\documentclass{article}
\usepackage{hyperref, graphicx}
\usepackage{amsmath}
\usepackage{tocloft}
\usepackage{amssymb}
\usepackage{slashed}
\usepackage{yfonts}
\usepackage[T1]{fontenc}
\usepackage[utf8]{inputenc}
\usepackage[english]{babel}
\usepackage{amsthm}
\usepackage{mathabx}
\usepackage{subcaption}
\usepackage{tikz-cd}
\usepackage{quiver}
\usepackage{mathabx}
\usepackage{stmaryrd}
\usepackage{braket}
\usepackage{physics}
\usepackage{makecell}
\usepackage{multirow,array}
\usepackage{adjustbox}
\usepackage[most]{tcolorbox}
\usepackage{geometry}
\usepackage[
backend=biber,
style=numeric-comp,
sorting=none,
]{biblatex}
\usepackage{url}

\makeatletter
\newcommand{\ostar}{\mathbin{\mathpalette\make@circled\star}}
\newcommand{\make@circled}[2]{%
  \ooalign{$\m@th#1\smallbigcirc{#1}$\cr\hidewidth$\m@th#1#2$\hidewidth\cr}%
}
\newcommand{\smallbigcirc}[1]{%
  \vcenter{\hbox{\scalebox{0.77778}{$\m@th#1\bigcirc$}}}%
}
\makeatother

\newtheorem{theorem}{Theorem}[section]

\newtheorem{corollary}[theorem]{Corollary}
\newtheorem{proposition}[theorem]{Proposition}
\newtheorem{conjecture}[theorem]{Conjecture}

\usepackage{authblk}
    \allowdisplaybreaks

\hypersetup{
pdfstartview = {FitH},
}
\hypersetup{
	colorlinks=true,         
	linkcolor=blue,          
	citecolor=red,        
	urlcolor=blue            
}

\DeclareMathAlphabet{\mathpzc}{OT1}{pzc}{m}{it}

\makeatletter
\def\@font@info#1{}
\makeatother

\theoremstyle{remark}
\newtheorem{remark}{Remark}[section]
\newenvironment{rmk}{\begin{remark}}{\hfill$\Diamond$\end{remark}}

\newtheorem{example}[remark]{Example}

\theoremstyle{definition}
\newtheorem{definition}[theorem]{Definition}

\usepackage{mathtools}

\usepackage{mathrsfs}

\numberwithin{equation}{section}

\newcommand{\be}{\begin{equation}}
\newcommand{\ee}{\end{equation}}
\newcommand{\ba}{\begin{aligned}}
\newcommand{\ea}{\end{aligned}}

\def\bbra{\langle\langle}
\def\kket{\rangle\rangle}
\def\Tr{\mathbb{T}\mathrm{r}}
\def\Dim{\mathbb{D}\mathrm{im}}

\def\bbC{{\mathbb{C}}}

\def\bbG{{\mathbb{G}}}

\def\bbZ{{\mathbb{Z}}}

\newcommand\ttC         {\mathtt{C}}

\newcommand\ttI         {\mathtt{I}}

\newcommand\ttU         {\mathtt{U}}
\newcommand\ttV         {\mathtt{V}}
\newcommand\ttW         {\mathtt{W}}
\newcommand\ttX         {\mathtt{X}}

\def\bbZ{\mathbb{Z}}
\def\R{\mathbb{R}}
\def\bbC{\mathbb{C}}

\newcommand{\cA}{{\cal A}}
\newcommand{\cC}{{\cal C}}
\newcommand{\cD}{{\cal D}}

\newcommand{\cH}{{\cal H}}

\newcommand{\cV}{{\cal V}}

\makeatletter
\newsavebox{\@brx}
\newcommand{\llangle}[1][]{\savebox{\@brx}{\(\m@th{#1\langle}\)}%
  \mathopen{\copy\@brx\kern-0.5\wd\@brx\usebox{\@brx}}}
\newcommand{\rrangle}[1][]{\savebox{\@brx}{\(\m@th{#1\rangle}\)}%
  \mathclose{\copy\@brx\kern-0.5\wd\@brx\usebox{\@brx}}}
\makeatother

\makeatletter
\newcommand{\xRrightarrow}[2][]{\ext@arrow 0359\Rrightarrowfill@{#1}{#2}}
\newcommand{\Rrightarrowfill@}{\arrowfill@\equiv\equiv\Rrightarrow}
\newcommand{\xLleftarrow}[2][]{\ext@arrow 3095\Lleftarrowfill@{#1}{#2}}
\newcommand{\Lleftarrowfill@}{\arrowfill@\Lleftarrow\equiv\equiv}
\makeatother

\newcommand{\id}{\operatorname{id}}

\begin{document}


\title{Categorical spin-networks: state-sum invariants of 4-manifolds from higher-gauge theory}

\bigskip

\author[1,2]{{\sf Hank Chen}\thanks{hank.chen@uwaterloo.ca}\thanks{chunhaochen@bimsa.cn}}
\author[3]{{\sf Zhian Jia}\thanks{giannjia@foxmail.com}}
\author[4]{{\sf Ran Luo}\thanks{ranluo@pku.edu.cn}}
\author[1]{{\sf Wei Cui}\thanks{cwei@bimsa.cn}}

\affil[1]{\small Beijing Institute of Mathematical Sciences and Applications, Beijing 101408, China}
\affil[2]{\small Yau Mathematical Sciences Center, Tsinghua University, Beijing 100084, China}
\affil[3]{\small Institute of Quantum Physics, School of Physics, Central South University,
Changsha 410012, China}
\affil[4]{\small School of Physics, Peking University, 209 Chengfu Road, Beijing 100871, China}

\date{\normalsize\today}

\maketitle

\bigskip

\begin{abstract}
In this work, we define $(3+1)$-dimensional topological quantum field theories (TQFTs) and their lattice realizations with input data given by a pivotal fusion 2-category; more precisely, our framework applies in wider generality to presemisimple locally fusion pivotal tensor $2$-categories, which are equipped with a so-called \textit{shadow trace}. 
We introduce a decoration of triangulated manifolds by the input data, and develop a $20j$-symbol formalism for computing 4-simplex scattering amplitudes. We verify the topological invariance of 20$j$-symbol summations under the 4-dimensional Pachner moves. Furthermore, we also construct explicitly a 3-dimensional Levin-Wen-type lattice Hamiltonian model realizing the state-sum invariant. Our construction can be understood as a higher-dimensional analogue of the spin-networks construction for the Turaev-Viro invariant. Examples arising from $2$-groups are discussed from the perspective of higher gauge theory.
\end{abstract}

\bigskip


\tableofcontents


\section{Introduction}
\label{intro}






Topological quantum field theories (TQFTs) and their lattice realizations have attracted significant attention over the past decades. Beyond their fundamental theoretical importance, they also play a central role in condensed matter physics, topological quantum computation, and topological quantum memory, see, e.g., \cite{wang2010topological,simon2023topoligical,Wen2004}.
For a given TQFT, a lattice realization consists of an equivalence class of gapped Hamiltonians and their ground-state Hilbert spaces, denoted by $(\boldsymbol H,\mathcal{H})$, for which the TQFT emerges as the corresponding low-energy effective theory. Typical examples include the Kitaev quantum double model \cite{Kitaev2003,Buerschaper2013a,chang2014kitaev,Meusburger:2016usq,jia2023boundary,Jia2023weak,Chen2z:2023}, which realizes the Dijkgraaf-Witten theory (or BF theory) \cite{dijkgraaf1990}, and the Levin-Wen string-net model \cite{Levin2004,kirillov2011stringnet,KitaevKong_2012,lan2014topological,Lin2014generalizations,jia2024weakTube,xi2021latticerealizationgeneralthreedimensional}, which realizes the Turaev-Viro-Barrett-Westbury TQFT \cite{Turaev:1992hq,Barrett1993}.

The $ (2+1)\text{d} $ Turaev-Viro-Barrett-Westbury TQFT is constructed from a spherical fusion category. When extending this framework to higher dimensions, higher fusion categories must be considered as algebraic input. Nevertheless, explicit constructions of both the state-sum TQFT and the corresponding lattice model in higher dimensions are still very rare and remain under intensive investigation. For $ (3+1)\text{d} $, the Douglas-Reutter construction \cite{Douglas:2018} is known to provide a state-sum TQFT based on a fusion $2$-category. The $ 3\text{d} $ string-net model for realizing such state-sum TQFTs is discussed in \cite{xi2021latticerealizationgeneralthreedimensional}, but the problem is far from settled. Note that in the physics literature \cite{Levin2004}, there is also a $ 3\text{d} $ string-net lattice model, but it is constructed from a fusion $1$-category and is therefore not genuinely $ (3+1)\text{d} $ in the general sense.

Motivated by the above considerations, the goal of this paper is to categorify the state-sum TQFT \cite{Turaev:1992hq,Barrett1993} and the corresponding Levin-Wen string-net construction \cite{Levin:2004mi} to (3+1)-dimensions, based on the data of a "sufficiently good" (see \S \ref{prelim} for details) linear tensor 2-category $\cC$. We pay particular attention to the case where  $\cC=\operatorname{2Rep}(H)=\mathsf{Mod}_\cV(H)$ is given by the exact 2-representations of a linear (Abelian) Hopf category $H$. Here,  by a "Hopf category" we specifically mean a Hopf algebra object in  an ambient {linear} symmetric monoidal closed bicomplete bicategory $\cV$; such higher algebraic structures have been studied in, for instance, \cite{neuchl1997representation,Huang2026-dr,Chen:2023tjf,Chen:2025?,Green:2023qqr}.

\begin{table}\hspace{-0.7cm}
    \centering
    \begin{tabular}{c|c|c|c}
     \hline
      \hline
         (3+1)d TQFT & gauge 2-group $\mathbb{G}$ & ambient $\cV$ & combinatorial invariants \\ 
         \hline
         \makecell{Dijkgraaf-Witten \\ \cite{Kapustin:2013uxa,Wen:2019,KongTianZhou:2020}} & finite & \makecell{Kapranov-Voevodsky \\ 2-vector spaces \\ $\cV=\mathsf{2Vect}$ \cite{Kapranov:1994}} & {\makecell{Yetter invariant \\ \cite{Bullivant:2016clk,Martins:2006hx,Bochniak:2020vil,Bullivant:2019tbp}}} \\ 
         \hline
        \makecell{2-Chern-Simons \\ \cite{Martins:2010ry,Mikovic:2016xmo,Baez:2002highergauge,Baez:2004in,chen:2022,Bochniak_2021,Chen1:2025?,chen2:2025}} & compact Lie & \makecell{Crane-Yetter \\ measureable categories \\ $\cV=\mathsf{Meas}$ \cite{Crane:2003gk,Yetter2003MeasurableC,Baez:2012,Kristel:2023gus} } & \makecell{{\color{blue} \textit{quantum} 20$j$-symbols} \cite{Liu:2024qth} \\ ("Crane-Yetter-like" \cite{Baez:1995ph})}\\
        \hline
      \hline
    \end{tabular}
    \caption{This table describes the layout for the construction of TQFT invariants from higher-group gauge theories, based on the 2-representation theory of finite (resp. compact Lie) 2-groups $\mathbb{G}$ within the context of the symmetric linear bicategory $\cV=\mathsf{2Vect}$ (resp. $\mathsf{Meas}$). The category $\ttI=\mathsf{Vect}$ of $\mathbb{C}$-vector spaces serves as the unit object; in the case of $\mathsf{Meas}$, these vector spaces can be infinite-dimensional.}
    \label{tab:4dtqfts}
\end{table}


The main inspiration for our construction is higher-gauge theory and its space of  \textit{2-holonomy} states. In this case, the context determines the ambient bicategory $\cV$ and the Hopf category $H\in\cV$ therein; see  Tab. \ref{tab:4dtqfts}. For example, $H$ is either the Karoubi complete category  $H=\mathsf{Vect}_\mathbb{G}\in\mathsf{2Vect}$ of finite-dimensional vector (or Hilbert) spaces graded by a finite 2-group $\mathbb{G}$ \cite{SOZER2023109155,Huang:2024},\footnote{Briefly, the objects and 1-morphisms of $H$ are graded by those of the 2-group $\mathbb{G}$.} or the Abelianized measureable category $H=\mathsf{Meas}_\mathbb{G} \in\mathsf{Meas}$ on a compact Lie 2-group $\mathbb{G}$ \cite{Chen1:2025?}.

\begin{rmk}\label{bicategorybgd}
     Generally, one think of $H_G$ as a category generated by "vector spaces $\underline{\bbC}_g$ localized at $g\in G$" (but this need not be literally true). In the setting where $G$ is an algebraic group scheme, this can be made precise by the quasicoherent \textit{skyscraper sheaves} $H_G=\mathsf{Sky}_G$ on $G$ \cite{Jia:2025vrj,Freed:2009qp}. As indicated in Tab. \ref{tab:4dtqfts}, there are measureable categories of Crane-Yetter which, by viewing $G$ (with its Haar-Radon measure $\mu$) as ringed by the commutative $W^*$-algebra $L^\infty(G,\mu)$, generalizes $\operatorname{QCoh}(G)\rightsquigarrow$ "$\operatorname{QCoh}_\mathsf{Meas}(G)$" to the locally compact Hausdorff setting. Similar geometric models for $\mathsf{Hilb}_\mathbb{G}$ have also appeared in  \cite{Stockall:2025ngz,Jia:2026vcr,2026arXiv260603368X}. A further generalization of $\mathsf{Meas}$, which consist of (projective locally-presentable) modules of non-commutative von Neumann algebras, has also been considered in \cite{Kristel:2023gus}. These are sheaf-theoretic versions of the bicommutant categories of Henriques-Penneys \cite{Henriques2017-gm}.
\end{rmk}


Denote by $\Gamma$ any oriented combinatorial graph complex embedded in a 3-manifold $M^3$; for instance, $\Gamma$ may be taken as the simplicial complex $\Delta$ underlying a triangulation of $M^3$. In higher-gauge theory, the fundamental degrees-of-freedom are the admissible configuration $\{(h_e\xrightarrow{b_t}h_{e'})\}_{(e\xrightarrow{t}e')\in\Delta^{(2)}}$ of $\mathbb{G}$-decorated 2-simplices $e\xrightarrow{f}e'\in\Delta^{(2)}$ \cite{Martins:2010ry,Martins:2006hx,Bochniak:2020vil,Bochniak_2021} --- aka. the so-called discrete \textit{2-holonomies} \cite{Chen1:2025?,schreiber2013connectionsnonabeliangerbesholonomy}. 

In Prop. 6.9 of \cite{chen2:2025}, it was shown that the quantum states associated to these $\mathbb{G}$-decorated polygons are modelled by (\textit{homogeneous morphisms} in) a linear Hopf monoidal category $H_\Delta^*$, which plays the role of a "categorified algebra of functions". Under some technical assumptions, $H^*_\Delta$ is given by functorial assignments of the morphisms in the Yoneda dual Hopf category $H^*\simeq\operatorname{Fun}_\cV(H,\mathsf{Vect})$ to the 2-skeleton of $\Delta$,
\begin{equation*}
    H^*_\Delta = \left\{ \Delta^{\leq 2} \rightarrow \operatorname{Mor}_\bbC (H^*)\right\},
\end{equation*}
where $\operatorname{Mor}_\bbC(C)$ denotes the direct sum of all hom-spaces of a linear category $C$.

Together with the categorical analogue of the \textit{Peter-Weyl theorem}, which is formally an equivalence
\begin{equation}
  H_\Delta^* \simeq \bigboxtimes_{t\in\Delta^{(2)}} \bigboxplus_{\mathtt{U}\in\operatorname{Irr}H}\mathtt{U}\boxtimes \mathtt{U}^\text{op}\, \qquad \underline{\text{in $\mathsf{2Vect}$}}\,,\label{formalPW}  
\end{equation}
higher-gauge theory then prompts us to assign (objects and morphisms of) indecomposable linear $H$-module categories $\ttU$ to (edges and faces of the 2-simplices in) the complex $\Delta$. See \S \ref{2pw} for a more precise formulation.

\medskip

This is the higher-categorical/-dimensional analogue of the spin-network/lattice Chern-Simons construction \cite{Baez:1994hx,Alekseev:1994pa,Alekseev:1994au}, where the $G$-holonomy states on the graph $\Delta$ form the function Hopf algebra $C(G^{\Delta^{(1)}}) \cong \bigotimes_{e\in\Delta^{(1)}}C(G)$. Indeed, the usual {Peter-Weyl theorem} \cite{Weyl1927} then gives $$C(G^{\Delta^{(1)}})\cong \bigotimes_{e\in\Delta^{(1)}} \bigoplus_{\rho\in\operatorname{Irr}G}\rho\otimes \rho^*\,,$$
which dictates that lattice $G$-gauge theory assigns irreducible representations of $G$ to the edges of $\Delta$. 

\medskip

However, here we wish to interpret \eqref{formalPW} and the corresponding decorated lattice in terms of the input 2-representation 2-category $\cC=\operatorname{2Rep}(H)$. This is achieved through the hom-objects $\ttU\boxtimes \ttU^\text{op}\simeq\operatorname{End}_H(\ttU)$ in $\cC$, whence the prescription \eqref{formalPW} is recast as assignments
\begin{equation*}
    K: \begin{pmatrix}
        \Delta^{(2)} \\ 
        \Delta^{(1)}\\
        \Delta^{(0)}
    \end{pmatrix}\to\begin{pmatrix}
        \text{2-morphisms of $\cC$}\\
        \text{1-morphisms of $\cC$}\\
        \text{objects of $\cC$}
    \end{pmatrix},\qquad \cC=\operatorname{2Rep}(H)\,
\end{equation*}
of 1- and 2-morphisms in $\cC$ to the simplices $\Delta$. For any sufficiently good 2-category $\cC$, this idea will serve as the foundation for the "$\cC$-decorated simplices" defined in  \S \ref{decorated2graphs}.
Furthermore, the maps determined from the Peter-Weyl theorem give rise to the higher {Clebsch-Gordan coefficients}.

\subsection{Relation to previous works: 4-manifold invariants}\label{diffinvars}
Given a positively/negatively-oriented 4-simplex $\sigma^\pm$ and a $\cC$-decoration $\mathcal{K}$ on it, we define its associated 4-simplex scattering amplitude $\langle\sigma^\pm(\mathcal{K})\rangle$ in \S \ref{20jsymbols}. They involve the so-called "20$j$-symbols" (\textbf{Definition \ref{20jsymboldef}}), which is explicitly evaluated under mild assumptions in \S \ref{eval}. Such 20$j$-symbols and their relevance to 4d state sum invariants have also been studied in \cite{Liu:2024qth}. 

However, since our construction of the 4-simplex amplitude differs from the conventional approach of \cite{Douglas:2018} (but is related to it; see \S \ref{discrete2gau}), we will explicitly prove their invariance under the 4-dimensional Pachner moves in \S \ref{pachner}. This  then allows us to define a state sum invariant \eqref{partition}  $$Z_\cC(M^4) = \sum_{\mathcal{K}}\prod_{\sigma\hookrightarrow T_{M^4}}\langle\sigma^\pm(\mathcal{K})\rangle$$
of oriented smooth 4-manifolds, where $T_{M^4}$ is any choice of triangulation on $M^4$. Here, the sum is taken over all \textit{simple} $\cC$-decorations $\mathcal{K}$ on the 4-simplices in $T_{M^4}$.

By considering the various loopings of $\cC$,
\begin{equation*}
    \Omega\cC = \operatorname{End}_\cC(\ttI),\qquad \Omega^2\cC = \operatorname{End}_{\operatorname{End}_\cC(\ttI)}(1_\ttI)\,,
\end{equation*}
we induce respectively $\cC$-decorations  $\Omega \cal K$  which assign the unit $\ttI\in\cC$ to all vertices in $\Delta$, or those $\Omega^2 \cal K$ which also assign the identity $1_\ttI: \ttI\to \ttI$ to every edge. We call the respective 4-simplex amplitudes $\langle\sigma^\pm({\Omega\mathcal{K}})\rangle\,,\langle\sigma^\pm(\Omega^2\mathcal{K})\rangle$ the 15$j$- and the 5$j$-symbols.
\[\begin{tikzcd}
	{\cal K} & {\Omega \cal K} & {\Omega^2 \cal K} \\
	20j & 15j & 5j
	\arrow["\langle\sigma^\pm\rangle",squiggly, from=1-1, to=2-1]
	\arrow["\langle\sigma^\pm\rangle",squiggly, from=1-2, to=2-2]
	\arrow["\langle\sigma^\pm\rangle",squiggly, from=1-3, to=2-3]
\end{tikzcd}\]
We justify the name "20$j$ symbols" in \S \ref{eval}, and the other names are due to the following facts:
\begin{enumerate}
    \item If $\Omega^2\cC\cong \bbC^\times$, then $\langle\sigma^+(\Omega^2\mathcal{K})\rangle$ produces the tetrahedral phases of Kashaev \cite{Kashaev2018}. 
    \item Given a tensor equivalence $\Omega\cC\simeq\operatorname{Rep}(U_q\mathfrak{sl}_2)$, $\langle\sigma^+(\Omega \mathcal{K})\rangle$ produces the 15$j$-symbols of Crane-Yetter \cite{Crane:1994ji}. This partially verifies the {assumption} that the Crane-Yetter TQFT is in actuality a certain higher-gauge theory.\footnote{Specifically, the conjecture is that the higher-gauge TQFT based on the inertia groupoid $\operatorname{Inn}SU(2) \simeq SU(2)//SU(2)$ is equivalent to the $SO(4)$-Crane-Yetter-Broda TQFT \cite{Baez:1995ph}. In terms of the categorical input (as described in \S 3.4.1 of \cite{Douglas:2018}), this is a tensor equivalence $\operatorname{2Rep}(\operatorname{Inn}SU_q(2))\simeq\Sigma\operatorname{Mod}(U_q\mathfrak{sl}_2)$; and indeed, one can verify that $\Omega\operatorname{2Rep}(\operatorname{Inn}SU_q(2))\simeq\operatorname{Mod}(U_q\mathfrak{sl}_2)$.}
\end{enumerate}
The property of "being reducible to the Kashaev and Crane-Yetter invariants" is also shared by the Hopf-triple invariant of \cite{Meusburger2025Categorical4I}.


\subsection{Relation to previous works: higher-gauge theory}\label{highergauexamples}
As mentioned in Tab. \ref{tab:4dtqfts}, when $\cC=\operatorname{2Rep}(\mathbb{G})$ is given by the 2-representations of a 2-group $\mathbb{G}$, the 4d TQFT $Z_\cC = Z_\mathbb{G}$ corresponds to lattice higher-gauge theory. We expand more on this here.

\subsubsection{Discrete 2-gauge theory}\label{discrete2gau}
For finite $\mathbb{G}$, we describe it in terms of its \textit{H{\'o}ang} data \cite{Ang2018,Baez2023HoangXS} $\mathbb{G} = (G,A,\tau)$ (or equivalently in the categorical extension notation $A[1]._\tau G$ \cite{Johnson-Freyd:2020}), where $\tau\in H^3(G,A)$ is its {Postnikov class}. We say $\mathbb{G}$ is \textit{split/strict} if $\tau=1.$

As mentioned in \S \ref{diffinvars} previously, our construction differ in the usual approach. Specifically for the Yetter-Dijkgraaf-Witten TQFT,  we work in the 2-representation/gauge symmetry basis \eqref{formalPW} $\cC= \operatorname{2Rep}(\mathbb{G})$ here as opposed to the 2-group algebra/global symmetry basis $\cD= \mathsf{2Vect}_\mathbb{G}$ as in \cite{Huang2026-dr,Kapustin:2013uxa} (see also \S 3.4.2 of \cite{Douglas:2018}). If we denote by $Z^\cD$ the TQFT resulting from the Douglas-Reutter construction \cite{Douglas:2018}, the underlying Morita duality \cite{Decoppet:2023bay}
\begin{equation}
    \operatorname{2Rep}(\mathbb{G}) \xrightarrow{\mathsf{2Vect}} \mathsf{2Vect}_\mathbb{G}\,\label{moritaeq}
\end{equation}
then provides two distinct ways to compute the same Yetter-Dijkgraaf-Witten invariant,
\begin{equation*}
    Z_{\operatorname{2Rep}(\mathbb{G})}(M^4)=\mathsf{YDW}_{\mathbb{G}}(M^4) = Z^{\mathsf{2Vect}_\mathbb{G}}(M^4)\,,
\end{equation*}
corresponding to the homotopy 2-type $\textbf{B}\mathbb{G}$. This duality is the foundation of the correspondence between the BF gauge-theoretic construction and the quantum double construction of the 4d Yetter-Dijkgraaf-Witten SymTFT \cite{Kapustin:2013uxa,Chen2z:2023,Huang:2025hpk,Zhang:2022rbg}. 

\medskip

In particular, by viewing the underlying 2-group $\mathbb{G}$ as a mixed 1-form and 0-form symmetry, the lattice model we construct in \S \ref{latticemodel} is an explicit microscopic realization of the 3d phase which one gets upon gauging these global symmetries \textit{simultaneously}.

Leveraging this duality, our 4d TQFT $Z_\cC=Z_\mathbb{G}$ gives the 4d Yetter-Dijkgraaf-Witten theory and produces many well-known examples upon reduction:
\[\begin{tikzcd}
	& \begin{array}{c} \substack{\mathbb{G} = A[1]._\tau G \\ \text{Yetter-Dijkgraaf-Witten theory}} \end{array} & \\
	\begin{array}{c} \substack{\mathbb{G} = \ast[1]. G \\ \text{4d Dijkgraaf-Witten theory}} \end{array} && \begin{array}{c} \substack{\mathbb{G}=A[1].\ast\\ \text{Walker-Wang model}} \end{array}
	\arrow[from=1-2, to=2-1]
	\arrow[from=1-2, to=2-3]
\end{tikzcd}\]
For instance, for $\mathbb{G}=\bbZ_2[1]._\tau \bbZ_2$ one produces the 4d toric code.

\medskip

We expect the \textbf{electromagnetic duality} between our 4d TQFT $Z_\cC$ and the Douglas-Reutter construction $Z^\cD$ to hold more generally, in analogy with the duality of the Kitaev quantum double model and the Levin-Wen model \cite{Buerschaper2009,Buerschaper:2010yf,Zoltan:2010,Wan:2020,Delcamp:2016yix}. 
\begin{conjecture}\label{2-duality}
    If $\cC,\cD$ are Morita dual fusion 2-categories, then  $$Z_\cC(M^4) = Z^\cD(M^4) \qquad \forall~ \text{closed }M^4\,$$
    computes the same state-sum invariant.
\end{conjecture}
\noindent This will be the subject of a future work; we will also say more about this in \S \ref{outlook}.

\subsubsection{Continuous 2-gauge theory}
As mentioned in \textit{Remark \ref{bicategorybgd}}, must of the results of this paper can also be applied to \textit{compact Lie 2-group} gauge theories; ie. 2-Chern-Simons theories, in accordance with Tab. \ref{tab:4dtqfts}. Specifically, we can set our 2-category $\cC$ to be the so-called  "finitary" 2-representations ({\` a} la  \cite{Macpherson20212RepresentationsAA}) of a given model for the infinite Hopf category $H=\mathsf{Hilb}_\mathbb{G}$; such a $\cC$ was studied in \cite{Chen:2025?}. The corresponding state sum invariant can be understood as a higher-group version of the Turaev-Viro invariant.  

For strict Lie $\mathbb{G}$, we describe it in terms of a Lie group crossed-module $\mathbb{G} = \mathsf{H}\xrightarrow{t}G$ \cite{Porst2008Strict2A}. We have the following two special \textit{split} cases: 
\begin{enumerate}
    \item \textbf{Skeletal 2-group ${t}=1$}: we have $\mathbb{G} = A[1]. G$ with $A$ Abelian and $\operatorname{dim}A = \operatorname{dim}G$. This corresponds to the \textbf{4d BF theory}, whose partition function was computed \cite{Schwarz1978-df,Schwarz1979-fa,Gegenberg:1993gd} to gives rise to the \textit{Ray-Singer torsion} $\tau(M^4)$ \cite{RAY1971145}. 
    \item \textbf{Inner automorphism 2-group ${t}=\id$}: we have $\mathbb{G}=\operatorname{Inn}G\simeq G//G$ (the loop/inertia groupoid of $G$). This case corresponds to the \textbf{4d BF-BB/twisted BF theory}, whose partition function was computed \cite{Baez:1995ph,Crane:1993cm} as the signature $\sigma(M^4)$ under a certain ansatz.
\end{enumerate}
These examples are also highly relevant for 4d quantum gravity \cite{Baez:1999sr,Girelli:2021zmt,Mikovic:2011si,Girelli:2007tt,Freidel:2012np}.

As such, for arbitrary values of $t\neq0,1$, the corresponding higher-gauge theory/partition function can be understood as a family of theories "interpolating" between 4d BF and BF-BB theories and their combinatorial invariants
\[\begin{tikzcd}
	0 &&&& 1 \\
	{BF\text{ theory}} && {\substack{\text{strict} \\ \text{2-Chern-Simons}}} && {BF\text{-}BB\text{ theory}} \\
	{\tau(M^4)} && {Z_{2CS}(M^4)} && {\sigma(M^4)}
	\arrow["{\mathsf{t}}", from=1-1, to=1-5]
	\arrow[no head, from=2-1, to=2-3]
	\arrow[squiggly, from=2-1, to=3-1]
	\arrow[no head, from=2-3, to=2-5]
	\arrow[squiggly, from=2-3, to=3-3]
	\arrow[squiggly, from=2-5, to=3-5]
\end{tikzcd}.\] 
Moreover, for the special values ${t}=1,\id$ (and only at these values), the underlying gauge Lie 2-algebras have canonical \textit{double structures} \cite{chen:2022,Chen2z:2023,Bai_2013} (ie. the corresponding 2-category $\cC$ is a Drinfeld centre). This explains why $Z_{2CS}$ at $\mathsf{t}=1,\id$, namely the BF SymTFT \cite{Jia:2025vrj,2026arXiv260603368X} and the Crane-Yetter TQFT, are \textit{invertible}, while in general $Z_{2CS}$ is not.




\subsubsection*{Acknowledgments}
H. C. is supported by the National Natural Science Foundation of China (grant No. W2533012). W. C. is supported by National Natural Science Foundation of China ( grant No. 42530108). 
Z. J. is supported by the National Natural Science Foundation of China (grant
No. 12605015), the start-up grant of Central South University (Grant No. 502045031) and by the Frontier Interdisciplinary Direction "Quantum Technologies" Project of Central South University (Grant No. 506010805). 
R.L. is supported by National Natural Science
Foundation of China under Grant No. 12422503.

\section{Preliminaries: pivotal tensor 2-categories and toroidal traces}\label{prelim}
Let us first introduce the categorical structures that we will be focusing on in the rest of this paper. More details on the following notions can be found in \cite{Douglas:2018,Chen:2025?,Chen2024ManifestlyUH}.

\subsection{Rigidity and planar-pivotality}
Recall that a 2-category $\cC$ is \textit{linear} iff it is 2-enriched in $\mathsf{Vect}$; namely that the 2-hom-spaces are vector spaces. A linear 2-category  $\cC$ with all (right) adjoints $F^\dagger\vdash F$ for its 1-morphisms $F:\ttU\to \ttV$ is called \textbf{planar-pivotal} iff the adjunction co/units
    $$e_F:F^\dagger\circ F\Rightarrow  1_\ttU,\qquad \iota_F: 1_\ttV\Rightarrow F\circ F^\dagger$$
    satisfy the snake equations 
    $$\id_F=(F\circ e_F)\ast(\iota_F\circ F)\,,$$ are compatible with the composition. Moreover, one can choose an invertible 2-isomorphism $\phi_F:F^{\dagger\dagger}\xRightarrow{\sim}F$
    for which
\begin{equation}
    e_F\ast(F^\dagger\circ \phi_F)\ast\iota_{F^\dagger} = (\operatorname{dim}F)\cdot \id_{1_\ttU}\label{1mordim}
\end{equation}
by definition computes the \textit{pivotal dimension}. By presemisimplicity, we have $\operatorname{dim}F\neq 0 $ for non-zero $F$.
    
    Given a linear 2-category $\cC$ with all adjoints, it is called \textbf{rigid tensor} iff $(\cC,\boxtimes,\ttI)$ is in addition monoidal, and each object $\ttU\in\cC$ has a dual $\ttU^*\in\cC$ equipped with the following co/evaluation 1-morphisms
    \begin{gather*}
        \operatorname{ev}_\ttU: \ttU^*\boxtimes \ttU\to \ttI,\qquad \operatorname{coe}_\ttU: \ttI\to \ttU\boxtimes \ttU^*\,,
    \end{gather*}
    equipped with their own adjoints and are compatible with the monoidal structure, as well as the following \textit{rigidity data}:
    \begin{enumerate}
    \item for each 1-morphism $F:\ttU\to \ttV$ and its dual-mate $F^*:\ttV^*\to \ttU^*$, we have the following 2-morphisms
     \begin{equation*}
              \operatorname{ev}_F: \operatorname{ev}_\ttU\circ (F^*\boxtimes \ttU)\Rightarrow \operatorname{ev}_{\ttV}\circ (\ttV^*\boxtimes F),\qquad \operatorname{coe}_F: (\ttU\boxtimes F^*)\circ \operatorname{coe}_\ttU\Rightarrow (F\boxtimes \ttV^*)\circ \operatorname{coe}_\ttV\,,
        \end{equation*}
        such that the collection $(\operatorname{ev}_\ttU,\operatorname{ev}_F)$ defines a \textit{2-cowedge} $\operatorname{ev}: -^*\boxtimes - \Rightarrow \ttI$, and the collection $(\operatorname{coe}_\ttU,\operatorname{coe}_F)$ defines a \textit{2-wedge} $\operatorname{coe}: \ttI \Rightarrow -\boxtimes -^*$ \cite{Loregian_2021};
        \item for each $\ttU\in\cC$, we have the following \textit{snakerator} 2-morphisms
        \begin{equation*}
            \rho_U:(\ttU\boxtimes \operatorname{ev}_\ttU)\circ(\operatorname{coe}_\ttU\boxtimes \ttU)\Rightarrow 1_\ttU,\qquad  \varrho_U: 1_{\ttU^*}\Rightarrow (\operatorname{ev}_\ttU\boxtimes \ttU^*)\circ(\ttU^*\boxtimes \operatorname{coe}_\ttU)
        \end{equation*}
        which satisfy the \textit{swallowtail equations}.
    \end{enumerate}
It would be convenient to write the adjunction data as a 1-op 2-functor $-^\dagger: \cC\to\cC^\text{1-op}$, and the duality as a op-monoidal 2-functor $-^*:\cC\to \cC^\text{m-op}$.

\subsubsection{Higher-dagger structure}
Given a tensor 2-category, we can introduce a ($*$-)involution $-^T: \cC\to \cC^\text{2-op}$ which we call the \textbf{dagger structure}. The compatibility between adjunction-mates and dagger structure is the strong-commutativity
$$ (-^{T})^{\text{1-op,2-op}} \circ -^\dagger = (-^{\dagger})^{\text{2-op}} \circ -^T\,,$$ which makes $\cC$ into a \textbf{2-dagger} tensor 2-category \cite{ferrer2024daggerncategories,stehouwer2023dagger}.

\begin{definition}\label{2unitary}
    Let $\cC$ be a dagger tensor 2-category. We say that a 2-morphism $\eta$ is \textbf{dagger unitary} iff it satisfies $\eta^T =\eta^{-1}$. 
\end{definition}

For a \textit{2}-dagger tensor 2-category $\cC$, the adjunction co/units $e_F,\iota_F$ are compatible with the dagger structure. Concretely,  up to the planar-pivotal 2-isomorphism $\phi_F: F^{\dagger\dagger}\cong F$ this means the following condition
\begin{equation}
        \iota_F^T\ast(\phi_F\circ F^\dagger )= e_{F^\dagger}\,\label{unitaryadjoint}
\end{equation}
from the definition of the pivotal dimension \eqref{1mordim}. Moreover, these co/units $\iota,e$ and the planar-pivotal structure $\phi$ can be chosen to be co/isometries
\begin{equation*}
    e_F\ast e_F^T = \iota_F^T\ast\iota_F = \operatorname{dim}F \cdot \id_{1_\ttU}
\end{equation*}
in terms of the 1-morphism dimension \eqref{1mordim}. We will always work with such a choice of unitary adjunctions in the following.








\begin{rmk}\label{2dagger}
    It is worth emphasizing that the 2-dagger structure mentioned above have a geometric origin in terms of string diagrams representations. The adjunction $-^\dagger$ corresponds to a planar $\pi$-rotation, while $-^T$ is a planar reflection. Their composite $-^{\dagger T}= -^{T \dagger}$ is therefore a planar orientation reversal; see also Def. 4.29 and Def. 5.3 of \cite{Liu:2024qth}.
\end{rmk}

Now specifically in the case $\cC=\operatorname{2Rep}(H)$, in order for $\cC$ to have 1-morphism adjoints, we need to consider {$\mathsf{Hilb}$-linear} $H$-modules $\operatorname{2Rep}(H) = \operatorname{Mod}_\mathsf{2Hilb}^u(H)$ in the 2-dagger symmetric bicategory $\cV=\mathsf{2Hilb}$ \cite{Chen2024ManifestlyUH}. One can replace all mention of $\mathsf{(2)Vect}$ by $\mathsf{(2)Hilb}$ in previous discussions.

\subsubsection{Pivotal structure}\label{pivotality}
Notice that a 2-dagger rigid tensor 2-category $\cC$ with all adjoints for 1-morphisms has equipped \textit{three} levels of duality, each acting on different layers of the 2-category $\cC$. A table summarizing their properties is given in Tab. \ref{tab:adjoints}. 

\begin{table}
    \centering
    \begin{tabular}{c|c|c|c}
        & objects & 1-morphisms & 2-morphisms \\
        \hline
         $-^T:\cC\to \cC^\text{2-op}$ & $\ttU$ & $F:\ttU\to \ttV$ & $\eta^T: G\Rightarrow F$ \\
         $-^\dagger:\cC\to \cC^\text{1-op}$ & $\ttU$ & $F^\dagger: \ttV\to \ttU$ & $\eta^\dagger: G^\dagger\Rightarrow F^\dagger$  \\
         $-^*:\cC\to \cC^\text{m-op}$ & $\ttU^*$ & $F^*: \ttV^*\to \ttU^*$ & $\eta^*: G^*\Rightarrow F^*$ 
    \end{tabular}
    \caption{The various duality structures endowed on a dagger rigid tensor 2-category $\cC$. By $\cC^{\text{1-op}}$ we mean $\cC$ with both op-1-/2-morphisms.}
    \label{tab:adjoints}
\end{table}
 
For the object-level duality $-^*$, its involutive-ness is \textit{datum} $-^{**}\simeq 1_\cC$. 
\begin{definition}\label{pivotal2cat}
    A \textbf{pivotal structure} on a unitary rigid tensor 2-category $\cC$ is a tuple $(p,P)$, where
\begin{equation*}
    p: -^{**}\xRightarrow{\sim} 1_{\cC}\,
\end{equation*}
is a monoidal 2-isomorphism and $P$ is a 3-isomorphism with components
\begin{gather}
   P_\ttU:\operatorname{coe}^\dagger_\ttU\circ (p_\ttU\boxtimes \ttU^*) \Rightarrow \operatorname{ev}_{\ttU^*}\,,
\end{gather}
witnessing the compatibility of $p$ against the rigid duality. 
\end{definition}
Importantly, if $\cC$ is dagger rigid, then the pivotal datum $(p,P)$ is in fact \textit{canonically} determined; indeed, we can define
\begin{equation*}
    p_\ttU=(\operatorname{ev}_{\ttU^*}\boxtimes\ttU)\circ (\ttU^{**}\boxtimes\operatorname{ev}_\ttU^\dagger ),\qquad P_\ttU :\operatorname{cev}^\dagger_\ttU\circ (p_\ttU\boxtimes \ttU^*) \Rightarrow \operatorname{ev}_{\ttU^*}
\end{equation*}
which makes $\cC$ pivotal \cite{Chen:2025?}. This structure $(p,P)$ is in general \textit{not} strict $p=\id$ as in \cite{Douglas:2018}. 




\subsubsection{Pivotal tensor 2-categories from 2-representations}\label{2reppivotal}
The construction of monoidal (and even braided) 2-categories from representations of higher Hopf algebras $H$ has been extensively studied from various angles \cite{BAEZ1996196,neuchl1997representation,gurski2006algebraic,Chen:2023tjf,Decoppet:2023bay,Green:2023qqr}. The central example in this paper is the 2-category $\cC=\operatorname{2Rep}(H)=\operatorname{Mod}^u_{\mathsf{2Hilb}}(H)$ of finite $\mathsf{Hilb}$-linear \textit{exact} left $H$-module $H^*$-categories in $\cV$ (or more generally, the so-called \textit{perfect} $H$-modules \cite{gainutdinov2026fullyexactfullydualizable}), and exact $H$-module functors between them \cite{fuchs2025spherical}. 

As the underlying Hopf category $H$ serves as the foundation of higher-gauge theory,\footnote{The idea that "4d lattice $\mathbb{G}$-gauge theory  \cite{Kapustin:2013uxa,Bochniak_2021,Bullivant:2016clk} determines a TQFT described by the tensor 2-category $\cC=\operatorname{2Rep}(H_\mathbb{G})$ through its observables" is a well documented folklore \cite{Kong:2014qka,Wen:2019}. This was verified for special cases in \cite{Huang:2025hpk,Zhang:2022rbg,Chen2z:2023}, where the braided monoidal data for $\cC$ was extracted directly out of the lattice gauge fields.} let us describe some relevant structures on $H$.
\begin{enumerate}
    \item Fact 1: $\cC$ has duals given by $\ttU^*\simeq \ttU^\text{op}$, which is made into a right $H^\text{m-op}$-module by the antipode $S :H\to H^\text{m-op}$. This endows $\cC$ with rigidity under the relative Deligne tensor product.
    \item Fact 2: there is an equivalence (of exact linear categories) between relative Deligne tensor products and exact functors,
    \begin{equation*}
        \ttU \boxtimes_H \ttV^\text{op} \simeq \operatorname{Fun}_H(\ttV,\ttU)=\operatorname{Hom}_\cC(\ttV,\ttU),\qquad \ttV,\ttU\in\cC\,,
    \end{equation*}
    which endows hom-objects to $\cC$. In particular, we have $\ttU\boxtimes_H \ttU^\text{op}\simeq\operatorname{End}_\cC(\ttU)$. 
    \item Fact 3: $\cC$ has all adjoints for the 1-morphisms in $\cC$. The double adjunction trivialization $F^{\dagger\dagger} \cong F$ is given by the relative Serre functor $\mathbb{S}$. In particular, all end-categories of $\cC$ are made pivotal monoidal in this way. 
\end{enumerate}
These facts can be found in \cite{fuchs2025spherical,douglas2020dualizable}, and also in \cite{Chen:2025?} in the manifestly unitary case.


In terms of lattice (higher-)gauge theory, it can be shown \cite{Chen1:2025?} that the Hopf structure of $H=H_\mathbb{G}$ have an interpretation in terms of the lattice geometry; this is in direct analogy in lattice Chern-Simons theory \cite{Alekseev:1994pa,Alekseev:1994au} (see also \cite{KitaevKong_2012}).

\medskip

For the rest of this paper, we will also assume that $\cC$ is presemisimple and locally finite idempotent complete \cite{Douglas:2018}. 
\begin{definition}\label{presemisimple}
    A tensor 2-category $\cC$ with all adjoints for 1-morphisms is called \textbf{presemisimple} iff 
    \begin{itemize}
        \item all objects in $\cC$ admit a decomposition into finite direct sums of simple objects, and
        \item all endomorphism categories $\operatorname{End}_\cC(\ttU)$ of simple objects $\ttU\in\cC$ are so-called \textbf{infusion}; that is, $\operatorname{End}_\cC(\ttU)$ is linear semisimple rigid monoidal with duals given by the adjunction $-^\dagger$, and the identity $1_\ttU$ is a simple 1-morphism.
    \end{itemize}
    $\cC$ is \textbf{locally finite idempotent complete} iff all hom-categories are finite idempotent complete.
\end{definition}
\noindent In particular, $\cC$ is \textit{locally pivotal fusion} (over $\mathbb{C}$), while finiteness at the object-level is not yet presumed. Such 2-representations $\cC=\operatorname{2Rep}(H)$ are also known by the name "{wide finitary}" \cite{Macpherson20212RepresentationsAA}. We will only require $\cC$ to have finitely many connected components in \S \ref{latticemodel}.

\subsection{Shadows and traces in a tensor 2-category}\label{shadows}
In explicit computations of the scattering amplitudes, we will require a higher-dimensional notion of traces. Specifically, we will need a sort of "toroidal" trace for cubical 2-cells $\eta:F_1\circ F_2\Rightarrow F_3\circ F_4$ in $\cC$, as opposed to "spherical" traces for globular 2-cells defined in \cite{Douglas:2018}. 

For this, we will adopt the formulation of \textit{shadows} given in \cite{ponto2013shadows}. 
\begin{definition}
    Let $\cC$ denote a bicategory, and let $\textbf{T}$ denote a fixed linear category. A \textbf{$\textbf{T}$-shadow} for $\cC$ is a \textit{linear} functor
    \begin{equation*}
        \bbra-\kket_\ttU: \operatorname{End}_\cC(\ttU)\to \textbf{T}
    \end{equation*}
    for each object $\ttU\in\cC$, equipped with natural isomorphisms 
    \begin{equation*}
        \theta_{F,G}: \bbra F\circ G\kket_\ttV \xrightarrow{\sim}\bbra G\circ F\kket_\ttU,\qquad \theta_{F,G}^{-1}=\theta_{G,F}
    \end{equation*}
    in $\textbf{T}$ for each $F:\ttU\to \ttV$ and $G:\ttV\to \ttU$. Moreover, for $H:\ttV\to \ttV$, the following diagrams commute,
\[\begin{tikzcd}
	{\bbra (F\circ G)\circ H\kket_\ttV} && {\bbra H\circ (F\circ G)\kket_\ttV} && {\bbra (H\circ F)\circ G\kket_\ttV} \\
	{\bbra F\circ (G\circ H)\kket_\ttV} && {\bbra (G\circ H)\circ F\kket_\ttU} && {\bbra G\circ (H\circ F)\kket_\ttU}
	\arrow["{\theta_{F\circ G,H}}", from=1-1, to=1-3]
	\arrow["{\bbra\alpha_{F,G,H}\kket}"', from=1-1, to=2-1]
	\arrow["{\bbra\alpha_{H,F,G}^{-1}\kket}", from=1-3, to=1-5]
	\arrow["{\theta_{H\circ F,G}}", from=1-5, to=2-5]
	\arrow["{\theta_{G\circ H,F}^{-1}}"', from=2-1, to=2-3]
	\arrow["{\bbra\alpha_{G,H,F}\kket}"', from=2-3, to=2-5]
\end{tikzcd}\]
\[\begin{tikzcd}
	{\bbra H\circ 1_\ttV\kket_\ttV} & {\bbra 1_\ttV\circ H\kket} \\
	{\bbra H\kket_\ttV}
	\arrow["{\theta_{F,G}}", from=1-1, to=1-2]
	\arrow[from=1-1, to=2-1]
	\arrow[from=1-2, to=2-1]
\end{tikzcd}\]
We call the tuple $(\cC,\bbra-\kket,\textbf{T})$ a \textbf{bicategory with ($\textbf{T}$-)shadow}.
\end{definition}
\noindent Strictly speaking, $\theta_{F,G}^{-1}=\theta_{G,F}$ need not be part of the definition, as it can be directly proven. We call $\theta$ the \textbf{cyclicity operator} for the shadow $\bbra-\kket$.

In contrast to the notion of a fibre 2-functor $\cC\to\mathsf{2Vect}$ \cite{Decoppet:2023bay}, a shadow structure with target $\textbf{T} = \mathsf{Vect}$ or $\mathsf{Hilb}$ --- see \textit{Example \ref{2charexample}} later --- is rather a sort of "local fibre functor" for $\cC$. Though, keep in mind that $\bbra-\kket$ are not required to be monoidal.

\subsubsection{Shadow traces}\label{shadotrace}
Provided $\cC$ has all 1-morphism adjoints, then all of its 1-morphisms are (left) dualizable and hence shadow traces can be defined.
\begin{definition}
    Let $(\cC,\bbra-\kket,\textbf{T})$ denote a 2-category with shadow and adjoints for 1-morphisms, and let $F:\ttU\to \ttV$, $G\in\operatorname{End}_\cC(\ttV),\, H\in\operatorname{End}_\cC(\ttU)$. The \textbf{shadow trace} of a 2-cell  $\eta: G\circ F\Rightarrow F\circ H$  is given by the composite in $\textbf{T}$:
    \begin{align}
        \operatorname{tr}_\textbf{T}(\eta)&: \bbra G\kket_\ttV \xrightarrow{\bbra \id_G\circ \iota_F\kket} \bbra G\circ (F\circ F^\dagger)\kket_\ttV \cong \bbra (G\circ F)\circ F^\dagger\kket_\ttV \nonumber\\
        &\qquad\qquad \quad\xrightarrow{\bbra\eta\circ F^\dagger\kket}\bbra  (F\circ H)\circ F^\dagger\kket_\ttV \xrightarrow{\theta_{F\circ H,F^\dagger}} \bbra F^\dagger\circ (F\circ H)\kket_\ttU \nonumber\\
        &\qquad \qquad \quad \cong  \bbra (F^\dagger\circ  F)\circ H\kket_\ttU\xrightarrow{\bbra e_F\circ H\kket} \bbra H\kket_\ttU\,,\label{shadowtrace}
    \end{align}
    where we have suppressed the morphisms induced by the composition associators $\alpha$. This morphism $ \bbra G\kket_\ttV\to  \bbra H\kket_\ttU$ in $\textbf{T}$ does not depend on the choice of the adjunction data $F^\dagger,\iota_F,e_F$.
\end{definition}
\noindent Clearly, this $\textbf{T}$-shadow trace can be applied to any cubical 2-morphism, as long as two of its opposite 1-morphisms coincide.



Let us collect some useful facts about the shadow trace. Their proofs can be found in \S 7 of \cite{ponto2013shadows}. 
\begin{theorem}\label{shadowproperties}
Let $(\cC,\bbra-\kket,\textbf{T})$ denote a bicategory with rigid adjunctions and a $\textbf{T}$-shadow. We shall suppress the composition associators $\alpha$ in the following.
    \begin{enumerate}
        \item \textbf{Tightening:} For 2-cells $\eta: G\circ F\Rightarrow F\circ H$, and $\nu: G'\Rightarrow G\,, \mu: H\Rightarrow H'$, we have 
        \begin{equation*}
            \bbra \mu\kket\circ \operatorname{tr}_\textbf{T}(\eta)\circ \bbra\nu\kket = \operatorname{tr}_\textbf{T}\left( (\nu\circ F)\ast \eta\ast (F\circ \mu) \right)
        \end{equation*}
        as a morphism $\bbra G'\kket_\ttV \to \bbra H'\kket_\ttU$ in $\textbf{T}$.
        \item \textbf{Cyclicity:} for 2-cells $\eta_1: G_1\circ F\Rightarrow F'\circ H_1$ and $\eta_2: G_2\circ F'\Rightarrow F\circ H_2$, we form the  composites 
        \begin{align*}
            \eta_1\ast_A \eta_2:\,& (G_1\circ G_2)\circ F'\cong G_1\circ( G_2\circ F')\xRightarrow{G_1\circ \eta_2}G_1\circ(F\circ H_2) \cong (G_1\circ F)\circ H_2\\
            &\qquad \xRightarrow{\eta_1\circ H_2} (F'\circ H_1)\circ H_2 \cong F'\circ (H_1\circ H_2)\\
            \eta_2\ast_A\eta_1 :\,& (G_2\circ G_1)\circ F\cong  G_2\circ (G_1\circ F)\xRightarrow{G_2\circ \eta_1} G_2\circ(F'\circ H_1) \cong (G_2\circ G')\circ H_1 \\
            &\qquad \xRightarrow{\eta_2\circ H_1} (F\circ H_2)\circ H_1\cong F\circ(H_2\circ H_1).
        \end{align*}
        The following diagram commutes,
\[\begin{tikzcd}
	{\bbra G_1\circ G_2\kket} && {\bbra H_1\circ H_2\kket} \\
	{\bbra G_2\circ G_1\kket} && {\bbra H_1\circ H_2\kket}
	\arrow["{\operatorname{tr}_{\bf T}(\eta_1\ast_A\eta_2)}", from=1-1, to=1-3]
	\arrow["{\theta_{G_1,G_2}}"', from=1-1, to=2-1]
	\arrow["{\theta_{H_1,H_2}}", from=1-3, to=2-3]
	\arrow["{\operatorname{tr}_{\bf T}(\eta_2\ast_A\eta_1)}"', from=2-1, to=2-3]
\end{tikzcd}\]
In particular if $\eta_1: F_1\Rightarrow F_2$ and $\eta_2: F_2\Rightarrow F_1$, then $$\operatorname{tr}_\textbf{T}(\eta_1\ast \eta_2)= \operatorname{tr}_\textbf{T}(\eta_2\ast\eta_1)\,.$$
        \item If $\eta: G\circ 1_\ttV\Rightarrow 1_\ttV\circ H$, then $\operatorname{tr}_\textbf{T}(\eta) = \bbra \eta\kket_\ttV$.
        \item \textbf{1-functoriality:} for 2-cells $\eta_1: G\circ F_1\Rightarrow F_1\circ H$ and $\eta_2: H\circ F_2\Rightarrow F_2\circ K$, we form the composite
        \begin{align*}
            \eta_2\ast_B\eta_1:\,& G\circ (F_1\circ F_2)\cong (G\circ F_1)\circ F_2\xRightarrow{\eta_1\circ F_2} (F_1\circ H)\circ F_2\cong F_1\circ (H\circ F_2)\\
            &\qquad \xRightarrow{F_1\circ \eta_2} F_1\circ(F_2\circ K)\cong (F_1\circ F_2)\circ K\,.
        \end{align*}
        Then we have $$\operatorname{tr}_\textbf{T}(\eta_2\ast_B\eta_1) = \operatorname{tr}_\textbf{T}(\eta_2)\circ \operatorname{tr}_\textbf{T}(\eta_1): \bbra G\kket \to \bbra H\kket \to\bbra K\kket\,.$$
        \item \textbf{Conjugation invariance:} The dagger adjunctions and the composition associators $\alpha $ in $\cC$  give rise to the linear natural isomorphism
\begin{align}
    \bar\cdot &:\operatorname{2Hom}_\cC(G\circ F,F'\circ G')\cong \operatorname{2Hom}_\cC(F,G^\dagger\circ (F'\circ G'))\nonumber\\
    &\qquad\qquad\stackrel{\alpha}{\cong} \operatorname{2Hom}_\cC(F,(G^\dagger\circ F')\circ G')\cong \operatorname{2Hom}_\cC( F\circ G'^\dagger,G^\dagger\circ F')
\label{conjugation}
\end{align}
such that, for a cubical 2-cell $\eta: G\circ F\Rightarrow F\circ H$ we have
        $$\operatorname{tr}_\textbf{T}(\eta) = \operatorname{tr}_\textbf{T}(\bar\eta): \bbra G\kket_\ttV \to \bbra H\kket_\ttU,$$
        where $ \operatorname{tr}_\textbf{T}(\bar\eta)$ is evaluated using the left-dualizability of $F^\dagger$.
    \end{enumerate}
\end{theorem}

Following \cite{Joyal_Street_Verity_1996,ponto2013shadows}, we call the tuple $(\cC,\bbra-\kket,\textbf{T})$ a \textbf{traced bicategory/2-category}.

\subsubsection{Toroidal traces}\label{torustrace}
To recap, upon applying the shadow trace to a cubical 2-cell $\eta: G\circ F \Rightarrow F\circ H$, we obtain a morphism $$\operatorname{tr}_\textbf{T}(\eta): \bbra G\kket_\ttV\to \bbra H\kket_\ttU$$ in $\textbf{T}$. Now suppose $\textbf{T}$ itself is pivotal monoidal, such that $\bbra G\kket_\ttV\cong \bbra H\kket_\ttU\in\textbf{T}$, then we can apply an ordinary trace to the shadow trace.

\begin{definition}
    Let $(\cC,\textbf{T},\bbra-\kket)$ denote a {traced monoidal bicategory} and suppose $\textbf{T}$ is linear monoidal. Let $\eta: G\circ F \Rightarrow F\circ G$ be a 2-cell in $\cC$ for which $X=\bbra G\kket_\ttV\in\textbf{T}$ is (right-)dualizable with the duality datum $(X^\vee,\mathfrak{i}_X,\mathfrak{e}_X)$. The \textbf{right toroidal trace} of $\eta$ is the right trace of the endomorphism $\operatorname{tr}_\textbf{T}(\eta): X\to X$,
    \begin{align*}
        \Tr_R(\eta) = \mathfrak{i}_{X^\vee}\circ (X^\vee\otimes \operatorname{tr}_\textbf{T}(\eta)) \circ \mathfrak{e}_X \in \operatorname{End}_\textbf{T}(e)\,,
    \end{align*}
    where $e\in \textbf{T}$ is the monoidal unit; similarly for the left-trace $\Tr_L(\eta)\in\operatorname{End}_\textbf{T}(e)$.
\end{definition}
\noindent  If $\textbf{T}$ is in addition spherical, then we call $       \Tr (\eta)= \Tr_R(\eta)=\Tr_L(\eta)$ the \textbf{toroidal trace} of $\eta$ --- this is a manifestation of "3-sphericality"; see \textit{Remark \ref{3sphericality}}. 

Let us prove formulas that exhibit the cyclic symmetries along the two "$A/B$-cycles". Consider 2-cells $\eta_1,\eta_2,\eta_3,\eta_4$ in $\cC$ in the following configuration
\[\begin{tikzcd}
	{\ttU} && {\ttV} && {\ttU} \\
	{\ttW} && {\ttW} && {\ttW} \\
	{\ttU} && {\ttV} && {\ttU}
	\arrow[""{name=0, anchor=center, inner sep=0}, "{f_1}", from=1-1, to=1-3]
	\arrow["{g_2}"', from=1-1, to=2-1]
	\arrow[""{name=1, anchor=center, inner sep=0}, "{f_2}", from=1-3, to=1-5]
	\arrow["{g_1}", from=1-3, to=2-3]
	\arrow["{g_2}", from=1-5, to=2-5]
	\arrow[""{name=2, anchor=center, inner sep=0}, "{f_3}"{description}, from=2-1, to=2-3]
	\arrow["{g_4}"', from=2-1, to=3-1]
	\arrow[""{name=3, anchor=center, inner sep=0}, "{f_4}"{description}, from=2-3, to=2-5]
	\arrow["{g_3}"', from=2-3, to=3-3]
	\arrow["{g_4}", from=2-5, to=3-5]
	\arrow[""{name=4, anchor=center, inner sep=0}, "{f_1}"', from=3-1, to=3-3]
	\arrow[""{name=5, anchor=center, inner sep=0}, "{f_2}"', from=3-3, to=3-5]
	\arrow["{\eta_1}", between={0.2}{0.8}, Rightarrow, from=0, to=2]
	\arrow["{\eta_2}", between={0.2}{0.8}, Rightarrow, from=1, to=3]
	\arrow["{\eta_3}"', between={0.2}{0.8}, Rightarrow, from=2, to=4]
	\arrow["{\eta_4}"', between={0.2}{0.8}, Rightarrow, from=3, to=5]
\end{tikzcd}\]
Given the functoriality and cyclicity of the shadow trace, we have the cycicity in the $A$-cycle:
\begin{align*}
    \operatorname{tr}_\textbf{T}\big((\eta_1\ast_A\eta_2)\ast_B(\eta_3\ast_A\eta_4)\big) &= \operatorname{tr}_\textbf{T}(\eta_1\ast_A\eta_2)\circ\operatorname{tr}_\textbf{T}(\eta_3\ast_A\eta_4) \\
    &= \big(\theta_{f_1,f_2}\operatorname{tr}_\textbf{T}(\eta_2\ast_A\eta_1)\theta_{f_3,f_4}^{-1}\big)\circ \big(\theta_{f_3,f_4}\operatorname{tr}_\textbf{T}(\eta_4\ast_A\eta_3)\theta_{f_1,f_2}^{-1}\big)\\
    &= \theta_{f_1,f_2}\big(\operatorname{tr}_\textbf{T}(\eta_2\ast_A\eta_1)\circ \operatorname{tr}_\textbf{T}(\eta_4\ast_A\eta_3)\big)\theta_{f_1,f_2}^{-1} \\ 
    &= \theta_{f_1,f_2}\big(\operatorname{tr}_\textbf{T}\big((\eta_2\ast_A\eta_1) \ast_B(\eta_4\ast_A\eta_3)\big)\big)\theta_{f_1,f_2}^{-1}\,,
\end{align*}
where we have used the invertibility of the cyclicity operator $\theta$. By taking the second trace in $\textbf{T}$, we get the cyclicity in the $B$-cycle:
\begin{align}
    &\Tr_R\big((\eta_1\ast_A\eta_2)\ast_B(\eta_3\ast_A\eta_4)\big) =  \Tr_R\big((\eta_2\ast_A\eta_1) \ast_B(\eta_4\ast_A\eta_3)\big) = \Tr_R\big((\eta_4\ast_A\eta_3)\ast_B(\eta_2\ast_A\eta_1) \big) \,.\label{cublicycle}
\end{align}
By the \textit{interchange law}
    \begin{equation*}
        (\eta_1\ast_A\eta_2)\ast_B(\eta_3\ast_A\eta_4) = (\eta_1\ast_B\eta_3)\ast_A(\eta_2\ast_B\eta_4)\,,
    \end{equation*}
\eqref{cublicycle} is also equivalent to 
\begin{equation*}
    \Tr_R\big((\eta_1\ast_B\eta_3)\ast_A(\eta_2\ast_B\eta_4)\big) =  \Tr_R\big((\eta_2\ast_B\eta_4) \ast_A(\eta_1\ast_B\eta_3)\big) = \Tr_R\big((\eta_4\ast_B\eta_2)\ast_A(\eta_3\ast_B\eta_1) \big) \, .
\end{equation*}
These equations can also be displayed as thus
\[\Tr_R\begin{tikzcd}
	\ttU & \ttV & \ttU \\
	\ttW & \ttW & \ttW \\
	\ttU & \ttV & \ttU
	\arrow[""{name=0, anchor=center, inner sep=0}, "{f_1}", from=1-1, to=1-2]
	\arrow["{g_2}"', from=1-1, to=2-1]
	\arrow[""{name=1, anchor=center, inner sep=0}, "{f_2}", from=1-2, to=1-3]
	\arrow[from=1-2, to=2-2]
	\arrow["{g_2}", from=1-3, to=2-3]
	\arrow[""{name=2, anchor=center, inner sep=0}, from=2-1, to=2-2]
	\arrow["{g_4}"', from=2-1, to=3-1]
	\arrow[""{name=3, anchor=center, inner sep=0}, from=2-2, to=2-3]
	\arrow[from=2-2, to=3-2]
	\arrow["{g_4}", from=2-3, to=3-3]
	\arrow[""{name=4, anchor=center, inner sep=0}, "{f_1}"', from=3-1, to=3-2]
	\arrow[""{name=5, anchor=center, inner sep=0}, "{f_2}"', from=3-2, to=3-3]
	\arrow["{\eta_1}", between={0.2}{0.8}, Rightarrow, from=0, to=2]
	\arrow["{\eta_2}", between={0.2}{0.8}, Rightarrow, from=1, to=3]
	\arrow["{\eta_3}"', between={0.2}{0.8}, Rightarrow, from=2, to=4]
	\arrow["{\eta_4}"', between={0.2}{0.8}, Rightarrow, from=3, to=5]
\end{tikzcd}=
\Tr_R\begin{tikzcd}
	\ttW & \ttW & \ttW \\
	\ttU & \ttV & \ttU \\
	\ttW & \ttW & \ttW
	\arrow[""{name=0, anchor=center, inner sep=0}, "{f_3}", from=1-1, to=1-2]
	\arrow["{g_4}"', from=1-1, to=2-1]
	\arrow[""{name=1, anchor=center, inner sep=0}, "{f_4}", from=1-2, to=1-3]
	\arrow[from=1-2, to=2-2]
	\arrow["{g_4}", from=1-3, to=2-3]
	\arrow[""{name=2, anchor=center, inner sep=0}, from=2-1, to=2-2]
	\arrow["{g_2}"', from=2-1, to=3-1]
	\arrow[""{name=3, anchor=center, inner sep=0}, from=2-2, to=2-3]
	\arrow[from=2-2, to=3-2]
	\arrow["{g_2}", from=2-3, to=3-3]
	\arrow[""{name=4, anchor=center, inner sep=0}, "{f_3}"', from=3-1, to=3-2]
	\arrow[""{name=5, anchor=center, inner sep=0}, "{f_4}"', from=3-2, to=3-3]
	\arrow["{\eta_3}"', between={0.2}{0.8}, Rightarrow, from=0, to=2]
	\arrow["{\eta_4}"', between={0.2}{0.8}, Rightarrow, from=1, to=3]
	\arrow["{\eta_1}"', between={0.2}{0.8}, Rightarrow, from=2, to=4]
	\arrow["{\eta_2}"', between={0.2}{0.8}, Rightarrow, from=3, to=5]
\end{tikzcd}=
\Tr_R\begin{tikzcd}
	\ttW & \ttW & \ttW \\
	\ttV & \ttU & \ttV \\
	\ttW & \ttW & \ttW
	\arrow[""{name=0, anchor=center, inner sep=0}, "{f_4}", from=1-1, to=1-2]
	\arrow["{g_1}"', from=1-1, to=2-1]
	\arrow[""{name=1, anchor=center, inner sep=0}, "{f_3}", from=1-2, to=1-3]
	\arrow[from=1-2, to=2-2]
	\arrow["{g_1}", from=1-3, to=2-3]
	\arrow[""{name=2, anchor=center, inner sep=0}, from=2-1, to=2-2]
	\arrow["{g_3}"', from=2-1, to=3-1]
	\arrow[""{name=3, anchor=center, inner sep=0}, from=2-2, to=2-3]
	\arrow[from=2-2, to=3-2]
	\arrow["{g_3}", from=2-3, to=3-3]
	\arrow[""{name=4, anchor=center, inner sep=0}, "{f_4}"', from=3-1, to=3-2]
	\arrow[""{name=5, anchor=center, inner sep=0}, "{f_3}"', from=3-2, to=3-3]
	\arrow["{\eta_4}"', between={0.2}{0.8}, Rightarrow, from=0, to=2]
	\arrow["{\eta_3}", between={0.2}{0.8}, Rightarrow, from=1, to=3]
	\arrow["{\eta_2}"', between={0.2}{0.8}, Rightarrow, from=2, to=4]
	\arrow["{\eta_1}"', between={0.2}{0.8}, Rightarrow, from=3, to=5]
\end{tikzcd}\]
This explains why $\Tr$ is called the \textit{toroidal trace}.


\medskip

\begin{example}\label{2charexample}
    Let us return to the example $\cC=\operatorname{2Rep}(H)$ of interest. For each $\ttU\in\cC$, consider the $\mathsf{Hilb}$-shadow given by the \textit{end}
    \begin{equation*}
        \bbra F\kket_\ttU = \int_{x\in \ttU} \operatorname{Hom}_\ttU(x,F(x))\cong \operatorname{Nat}(1_\ttU,F)
    \end{equation*}
    for each $F\in\operatorname{End}_\cC(\ttU)$.\footnote{In fact, for  2-representations $\cC=\operatorname{2Rep}(G)$ (in $\mathsf{2Vect}$) of a finite (2-)group $G$, the end is exactly the \textit{2-character} $$\chi:\operatorname{2Rep}(G)\to \mathsf{Vect},\qquad \chi_{(M,\rho)}(g) = \operatorname{Nat}(1_M,\rho(g))$$ introduced in \cite{Ganter:2006,Ganter:2014,Bartlett:2009PhD,Huang:2024}, and the toroidal trace of a 2-cell $\eta_{g,h}: \rho(g)\rho(h)\Rightarrow\rho(h)\rho(g)$ coming from mutually commuting elements $g,h\in G$ is precisely the \textit{joint character}. The cyclicity operator $\theta$ furnishes $\chi_{(M,\rho)}$ as a representation of the double $D(G)\simeq G//G$.} Following Prop. 3.8 of \cite{Ganter:2006}, we have a linear map $ \tilde\theta_{F,G}:\bbra F'\kket_\ttU \rightarrow \bbra G\circ F\circ G^\dagger\kket_\ttV$ for each $G: \ttU\to \ttV$, which precomposes a natural transformation $1_\ttU\Rightarrow F$ by the counit $\iota_G: 1_\ttV\Rightarrow G\circ G^\dagger$. By dagger adjunctions \eqref{unitaryadjoint}, this map $\tilde\theta$ is invertible and induces the requisite cyclicity operator $\theta$ (cf. Remark 4.7 of \cite{Huang:2024}). 
    
    Recall that, for simple $\ttU\in\cC$, its identity is simple; this is Prop. 1.2.14 in \cite{Douglas:2018}. Then, $\bbra 1_\ttU\kket_\ttU = \operatorname{Nat}(1_\ttU,1_\ttU)\cong \mathbb{C}$ is the simple unit in $\textbf{T}=\mathsf{Hilb}$. Now specifically for finite groups $H=\mathsf{Hilb}_G$, we have $\Omega\cC = \operatorname{End}_{\operatorname{2Rep}(G)}(\ttI)\simeq \operatorname{Rep}(G)$ \cite{Bartsch:2022mpm} whence 
    \begin{equation*}
        \bbra F\kket_\ttI \cong \operatorname{Nat}(1_\ttI,F)\cong\operatorname{Hom}_G(\mathbb{C},V_F) \cong V^G_F
    \end{equation*}
    is the $G$-coinvariants of the representation determined by $F: \ttI\to \ttI$. The shadow trace of a 2-cells $ \eta: F\circ H'\Rightarrow H\circ F$ for $H',H\in\Omega\cC$ can therefore be viewed as a linear map
    \begin{equation*}
        \operatorname{tr}_\mathsf{Vect}(\eta): V^{G}_{H'}\to V^{G}_H\,\,\quad \underline{\text{in } \mathsf{Hilb}}\,,
    \end{equation*}
    and the toroidal trace $\Tr(\eta)$ can then be seen as the usual trace of this matrix. The same computation can be performed for generic endomorphism categories $\operatorname{End}_H(\ttU)$ in $\cC=\operatorname{2Rep}(H).$
\end{example}

We will use the toroidal trace to compute 4-simplex scattering amplitudes in \S \ref{20jsymbols}.

\medskip

In summary, our input datum is a  tensor 2-category as follows
\begin{equation}
   \mathscr{C} =(\cC,\LaTeXunderbrace{\ttI,\boxtimes,\alpha,\pi}_{\text{monoidal}},\LaTeXunderbrace{\LaTeXunderbrace{-^*,\operatorname{ev},\operatorname{coe}}_{\text{rigid}},p,P}_{\text{pivotal}},\LaTeXunderbrace{\LaTeXunderbrace{-^\dagger,e,\iota}_{\text{adjoints}},-^T}_\text{2-dagger},\LaTeXunderbrace{\textbf{T},\bbra-\kket}_\text{"traced"})\,\label{input}
\end{equation}
which is also presemisimple and locally finite idempotent complete. Pivotal fusion 2-categories in the sense of \cite{Douglas:2018}, equipped with a compatible dagger structure, fit as an example.

\section{$\cC$-decorated simplices}\label{decorated2graphs}
Let $\cC$ denote a presemisimple 2-dagger pivotal 2-category $\cC $, eg. $\operatorname{2Rep}(H)$. We will not use its pivotal structure directly, but it will become relevant later.

We fix a simplicial complex $\Delta$ underlying a triangulation of the oriented 4-manifold $M^4$, and a total ordering of its vertices. More specifically, the complex $$\Delta=\Delta^{(2)} {\rightrightarrows}\Delta^{(1)}{\rightrightarrows} \Delta^{(0)}$$ is given by 2- and 1-simplices of $\Delta$, where each 2-simplex comes with a choice (consistent with the vertex ordering) of an oriented \textit{root edge} $e_\ast\in\Delta^{(1)}$.\footnote{The "double source" $ss(\Delta) = s(e_\ast)$ of a 2-simplex $\Delta\in\Delta^{(2)}$ is often called the \textit{root} of $\Delta$ \cite{Martins:2006hx,Bullivant:2019tbp,Bochniak_2021}.} This is an edge (or a composite of edges) determined by one of the face maps on $\Delta^{(2)}$.

The 1-simplex source/target maps $ \Delta^{(1)}\to \Delta^{(0)}$ and its composition law are obvious, while the 2-simplex source (resp. target) maps $\Delta^{(2)}\rightrightarrows \Delta^{(1)}$ are given respectively by the root edge (resp the composite of the remaining edge(s)). The horizontal composition of 2-simplices is given by a attaching along the source edges, while the vertical composition is given by attaching the source and target edge. 
\begin{definition}\label{2decorate}
    A \textbf{$\cC$-decoration on $\Delta$} is an assignment
    \begin{equation}
K([i]) = \ttU_{[i]},\qquad K([ij]) = F_{[ij]}: \ttU_{[i]}\to \ttU_{[j]},\qquad K([ijk]) = \eta_{[ijk]}: F_{[jk]}\circ F_{[ij]}\Rightarrow F_{[ik]}
\nonumber
\end{equation}
 such that, for each $i,j$,
   \begin{enumerate}
       \item $F_{[ji]} = F^\dagger_{[ij]}$ is the adjoint 1-morphism, 
       \item the transpose (2-morphism adjoint) of $ \eta_{[ijk]}$ is given by $\eta_{[ijk]}^T: F_{[ik]}\Rightarrow F_{[jk]}\circ F_{[ij]}$,
       \item the adjoint-mate of the 2-morphism $\eta_{[ijk]}$ is $\eta^\dagger_{[ijk]}=\eta_{[kji]}^T: F_{[ki]}\to F_{[ji]}\circ F_{[kj]}$,
       \item the identity 1-morphism $F_{[ii]} = 1_{\ttU_{[i]}}$ is (1-)self-adjoint,
       \item the identity 2-morphism $\eta_{[iij]} = \id_{F_{[ij]}}=\eta_{[ijj]}$ is (2-)self-adjoint, 
       \item gluing of 1-simplices is mapped to the composition of 1-morphisms,
       \item the horizontal/vertical gluing of 2-simplices is mapped to the horizontal/vertical composition $\ast_A,\ast_B$ of 2-morphisms (these are the definition of the "thorn/theta nets"; see \S \ref{scalaringredients}), and
       \item $e_{F_{[ij]}}=\eta_{[iji]},\,\iota_{F_{[ij]}}=\eta_{[jij]}^T$ are the counit/unit of the adjunction, for which
       \begin{equation*}
            \id^\dagger_{F_{[ij]}} = \id_{F_{[ji]}},\qquad e_{F_{[ij]}}^\dagger = \iota_{F_{[ji]}}: \id_{F_{[ji]}}\Rightarrow F_{[ji]}\circ F_{[ji]}^\dagger = F_{[ij]}^\dagger\circ F_{[ij]}\,.
       \end{equation*}
   \end{enumerate}
We denote such $\cC$-decorated simplices by $K:\Delta^{(2)}\to \cC$, and call it a \textbf{2-holonomy state} (see \S \ref{2pw}).
\end{definition}


\subsection{$\cC$-decorations on the PL 2-disc; states of 2-holonomies}\label{2pw} 
For this section, let $\mathcal{F}: \cC=\operatorname{2Rep}(H)\rightarrow\cV$ denote the forgetful {fibre 2-functor} \cite{Decoppet:2023bay}. The categorical Peter-Weyl theorem states that there is an equivalence
\begin{equation}
        H^*\simeq \int^{\ttU\in\operatorname{2Rep}(H)}\mathcal{F}(\ttU)\boxtimes \mathcal{F}(\ttU^\text{op}),\label{2pwthm}
\end{equation}
of Hopf algebra objects in the symmetric monoidal bicategory $\cV$, where the right-hand side involves the universal 2-coend \cite{Loregian_2021}. Provided $\operatorname{2Rep}(H)$ is rigid, the Yoneda dual version of this statement is in fact a manifestation of the \textit{2-Tannaka-Krein reconstruction} mentioned in \S \ref{2reppivotal},
\begin{equation*}
    H\simeq \int_{\ttU\in\operatorname{2Rep}(H)}\mathcal{F}(\ttU^\text{op}\boxtimes \ttU)\simeq \int_{\ttU\in\operatorname{2Rep}(H)}\operatorname{Hom}_H(\mathcal{F}\ttU,\mathcal{F}\ttU) \simeq \operatorname{Nat}(\mathcal{F},\mathcal{F})\,.
\end{equation*}
In the fusion case $\cV=\mathsf{2Vect}$, this has been established in \cite{Green:2023qqr}; see also \cite{huang2023tannaka} for the special case of finite 2-groups $H=\mathsf{Vect}_\mathbb{G}$ (ie. $\cC$ symmetric).

Now consider $\Delta=\Delta_{D^2}$ as the {simplicial complex} of a PL 2-disc $D^2$, consisting of one of each 0-/1-/2-simplex $[0],\ell,p$ respectively.  By specifying a $\cC$-decoration $K$ at the 0-simplex by $K([0]) = \ttU\in\operatorname{2Rep}(H)$, the PL 2-disc is decorated with a 2-endomorphism 
    \begin{equation*}
        K(p)= 
\begin{tikzcd}
	{} \\
	{\ttU}
	\arrow[between={0}{0.9}, Rightarrow, from=1-1, to=2-1]
	\arrow["{K(\ell)}", from=2-1, to=2-1, loop, in=50, out=130, distance=15mm]
\end{tikzcd}\in\operatorname{2End}_H(K(\ell))\,.
    \end{equation*}    
    By taking the direct sum over \textit{simple} such decorations $K=K_s$ (see \S \ref{tethilbertspace}) with $K_s([0])=\ttU$, we obtain a Hilbert space 
     \begin{equation*}
        \cH_{\Delta_{D^2},\ttU} =\bigoplus_{K_s \text{ simple}}K_s(p) \cong \bigoplus_{F\in\operatorname{Irr}(\ttC)}\operatorname{Mor}_\bbC(\ttC_F)\,,
    \end{equation*}
    where $\ttC=\operatorname{Hom}_\cC(\ttU,\ttU)$ is treated as a hom-object and $\ttC_F\subset \ttC$ is the full subcategory generated by the simple object $F$. By linearity and the local semisimplicity of $\cC$, we then have in turn
    \begin{equation}
        \cH_{\Delta_{D^2},\ttU} \cong \operatorname{Mor}_\bbC\left(\bigoplus_{F\in\operatorname{Irr}(\ttC)}\ttC_F\right)\cong \operatorname{Mor}_\bbC(\ttC)\cong \operatorname{Mor}_\bbC(\ttU\boxtimes \ttU^*)\,\label{internalU}
    \end{equation}
    where we have used the equivalence $\ttU\boxtimes \ttU^\text{op}\simeq\operatorname{Hom}_\cC(\ttU,\ttU)=\ttC$ as objects in $\cC$. 
    
    Upon forgetting the $H$-equivariance $\mathcal{F}:\operatorname{2Rep}(H)\to \mathsf{2Vect}$, the Peter-Weyl theorem \eqref{2pwthm} and the universal property of 2-coends give an embedding
    \begin{equation*}
        \mathcal{F}\ttU \boxtimes \mathcal{F}\ttU^\text{op}  \hookrightarrow \int^{\ttU\in\operatorname{2Rep}(H)}\mathcal{F}\ttU\boxtimes \mathcal{F}\ttU^\text{op} \simeq H^*\,,
    \end{equation*}
    which together with \eqref{internalU} tells us that the Hilbert space of the decorated 2-disc
    \begin{equation*}
        \cH_{\Delta_{D^2},\ttU} \hookrightarrow \operatorname{Mor}_\bbC(H^*)
    \end{equation*}
    embeds into the space of morphisms of the Hopf category $H^*$. This demonstrates that the $\cC$-decorations $K$ on the simplicial 2-disc can indeed be viewed as the data in higher lattice $\mathbb{G}$-gauge theory.

\subsection{States of decorated tetrahedra}\label{decoratedtet}
We now begin our construction of the quantum states on tetrahedra $\tau=[0123]$. Let $\Delta^{(2)}=\partial \tau$ denote its triangulated boundary.
\begin{enumerate}
    \item Decompose $\partial\tau = \Delta^+\cup_{\partial\Delta^\pm}\Delta^-$ into two 2-simplces:
    \begin{equation*}
        \Delta^+=[012]\cup_{[02]}[032],\qquad \Delta^-=[103]\cup_{[13]}[123].
    \end{equation*}
    Note each 1-simplex in $\partial\Delta^+$ also appears in $\partial\Delta^-$ (with possibly reversed orientation).
    \item Fix a $\cC$-decoration $K$ on $\Delta^{(2)}=\partial\tau$, with 1-morphism components $F_{[ij]}\in \operatorname{Hom}_\cC(\ttU_{[i]},\ttU_{[j]})$ for each 1-simplex $[ij]$ in $\partial\Delta^+\cong\partial\Delta^-$. Define the linear tensor product of 2-morphisms by
\begin{align}
        \eta_{\Delta^+} = \eta_{[032]}^T \otimes \eta_{[012]}&\in  \operatorname{2Hom}_\cC(F_{[02]},F_{[32]}\circ F_{[03]} )\otimes  \operatorname{2Hom}_\cC(F_{[12]}\circ F_{[01]},F_{[02]})\nonumber\\ 
        \eta_{\Delta^-} =\eta_{[123]}^T \otimes \eta_{[103]}& \in \operatorname{2Hom}_\cC(F_{[13]},F_{[03]}\circ F_{[10]})\otimes \operatorname{2Hom}_\cC(F_{[23]}\circ F_{[12]},F_{[13]}) \,,\label{Deltastates}
    \end{align}
    where "$\circ$" denotes the composition of 1-morphisms. These 2-simplex decorations are then combined along the attaching rules of the tetrahedron to form, up to the composition associator $\alpha$, the following \textbf{tetrahedron states},
    \begin{align}
        \big|\tau_{K}^+\big\rangle  &= (F_{[32]}\circ \eta_{\Delta^-}) \otimes \alpha_{F_{[32]},F_{[03]},F_{[10]}} \otimes (\eta_{\Delta^+}\circ F_{[10]})  \label{taustate+}\\
        \big|\tau_{K}^-\big\rangle  &= (F_{[23]}\circ \eta_{\Delta^+})\otimes\alpha_{F_{[23]},F_{[12]},F_{[01]}}\otimes (\eta_{\Delta^-}\circ F_{[01]})\label{taustate-}\,.
    \end{align}
    See also \textit{Remark \ref{internaltensor}} for a "$\cC$-internal" interpretation of these linear tensor products. The associator $F$-symbols $\alpha$ appearing here is the analogue of the Postnikov anomaly in lattice 2-gauge theory \cite{Kapustin:2013uxa,Wen:2019,Bullivant:2019tbp}. 
    
    \item By vertical 2-morphism composition (denoted $\ast$), the 2-simplex states \eqref{Deltastates} can be contracted to give
    \begin{align*}
        \sum\eta_{\Delta^+} = \eta_{[032]}^T \ast \eta_{[012]}&: F_{[12]}\circ F_{[01]}\Rightarrow F_{[02]}\Rightarrow F_{[32]}\circ F_{[03]} \nonumber\\ 
        \sum \eta_{\Delta^-} =\eta_{[123]}^T \ast\eta_{[103]}&: F_{[03]}\circ F_{[10]}\Rightarrow F_{[13]}\Rightarrow F_{[23]}\circ F_{[12]}\,,
    \end{align*}
and similarly for the tetrahedron states \eqref{taustate+}, \eqref{taustate-}:
\begin{align*}
    \sum \big|\tau_{K}^+\big\rangle  &= \left(F_{[32]}\circ \sum\eta_{\Delta^-}\right) \ast \alpha_{F_{[32]},F_{[03]},F_{[10]}} \ast \left(\sum\eta_{\Delta^+}\circ F_{[10]}\right) \nonumber\\
    &\qquad\qquad :F_{[12]}\circ F_{[01]}\circ F_{[10]}\Rightarrow F_{[32]}\circ  F_{[03]}\circ F_{[10]} \Rightarrow F_{[32]}\circ F_{[23]}\circ F_{[12]}\\
    \sum \big|\tau_{K}^-\big\rangle  &= \left(F_{[23]}\circ \sum\eta_{\Delta^+}\right)\ast\alpha_{F_{[23]},F_{[12]},F_{[01]}}\ast \left(\sum\eta_{\Delta^-}\circ F_{[01]}\right)\\
    &\qquad\qquad :  F_{[03]}\circ F_{[10]} \circ F_{[01]} \Rightarrow F_{[23]}\circ  F_{[12]}\circ F_{[01]}  \Rightarrow F_{[23]}\circ F_{[32]}\circ F_{[03]}\,.
\end{align*}
These contractions can be understood as the summation over all interior $\cC$-decorations.
\end{enumerate}

\begin{figure}
    \centering    \includegraphics[width=1\linewidth]{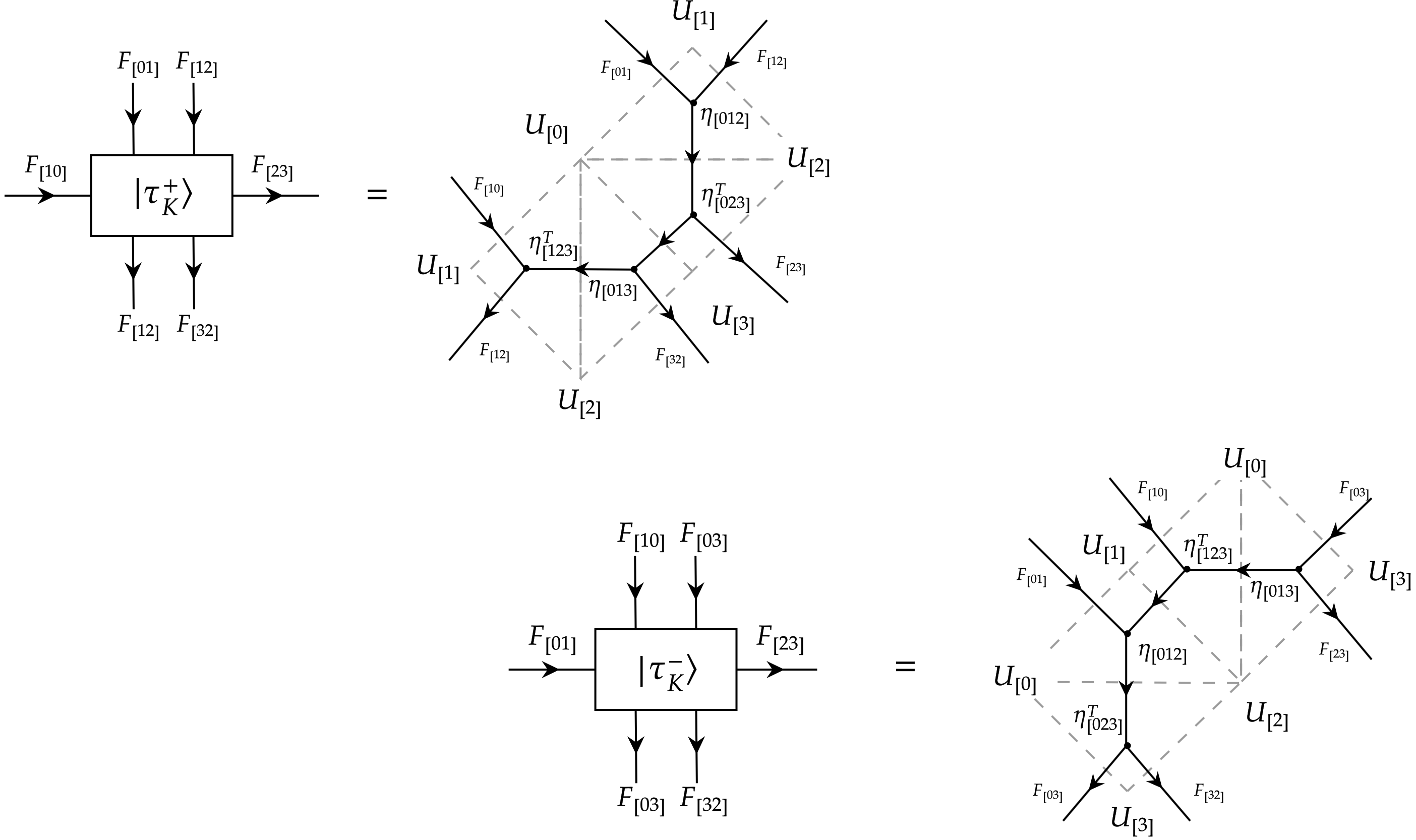}
    \caption{The tensor networks and the string diagrams corresponding to the (contracted) tetrahedron states \eqref{taustate+} (left) and \eqref{taustate-} (right). The regions are coloured by objects $\ttU_{[0]},\dots \ttU_{[3]}\in\cC$. The graphs given by the gray dashed lines depict the 2-simplices $\Delta^+,\Delta^-$.}
    \label{fig:planargraphs1}
\end{figure}

The tensor networks and string diagrams for the contracted tetrahedral states are shown in Fig. \ref{fig:planargraphs1}, and their  pasting diagrams are expressed as
\[\begin{tikzcd}
	&& {\ttU_{[1]}} \\
	& {\ttU_{[0]}} && {\ttU_{[2]}} \\
	{\ttU_{[1]}} && {\ttU_{[3]}} \\
	& {\ttU_{[2]}}
	\arrow["{F_{[12]}}", from=1-3, to=2-4]
	\arrow["{F_{[01]}}"{pos=0.3}, from=2-2, to=1-3]
	\arrow[""{name=0, anchor=center, inner sep=0}, from=2-2, to=2-4]
	\arrow["{F_{[03]}}"{description}, from=2-2, to=3-3]
	\arrow["{F_{[10]}}", from=3-1, to=2-2]
	\arrow[""{name=2, anchor=center, inner sep=0}, from=3-1, to=3-3]
	\arrow["{F_{[12]}}"', from=3-1, to=4-2]
	\arrow["{F_{[32]}}"', from=3-3, to=2-4]
	\arrow["{F_{[23]}}"'{pos=0.6}, from=4-2, to=3-3]
	\arrow["{\eta_{[012]}}"', between={0.1}{0.8}, Rightarrow, from=1-3, to=0]
	\arrow["{\eta_{[032]}^T}", between={0.2}{0.9}, Rightarrow, from=0, to=3-3]
	\arrow["{\eta_{[103]}}"', between={0.1}{0.8}, Rightarrow, from=2-2, to=2]
	\arrow["{\eta_{[123]}^T}", between={0.2}{0.9}, Rightarrow, from=2, to=4-2]
\end{tikzcd},\qquad \begin{tikzcd}
	&& {\ttU_{[0]}} \\
	& {\ttU_{[1]}} && {\ttU_{[3]}} \\
	{\ttU_{[0]}} && {\ttU_{[2]}} \\
	& {\ttU_{[3]}}
	\arrow["{F_{[03]}}", from=1-3, to=2-4]
	\arrow["{F_{[10]}}"{pos=0.6}, from=2-2, to=1-3]
	\arrow[""{name=0, anchor=center, inner sep=0}, from=2-2, to=2-4]
	\arrow["{F_{[12]}}"{description, pos=0.4}, from=2-2, to=3-3]
	\arrow["{F_{[01]}}", from=3-1, to=2-2]
	\arrow[""{name=2, anchor=center, inner sep=0}, from=3-1, to=3-3]
	\arrow["{F_{[03]}}"', from=3-1, to=4-2]
	\arrow["{F_{[23]}}"', from=3-3, to=2-4]
	\arrow["{F_{[32]}}"'{pos=0.6}, from=4-2, to=3-3]
	\arrow["{\eta_{[103]}}"', between={0.1}{0.9}, Rightarrow, from=1-3, to=0]
	\arrow["{\eta_{[123]}^T}", between={0.2}{0.9}, Rightarrow, from=0, to=3-3]
	\arrow["{\eta_{[012]}}"', between={0.2}{0.9}, Rightarrow, from=2-2, to=2]
	\arrow["{\eta_{[032]}^T}", between={0.2}{0.9}, Rightarrow, from=2, to=4-2]
\end{tikzcd}\]

Since $\cC$ is 2-$\mathsf{Vect}$-enriched, these states \eqref{taustate+}, \eqref{taustate-}  are vectors in some (Hilbert) space ${\cH}_{K(\tau)}^\pm\in\mathsf{Vect}$. In the case $\cC=\operatorname{2Rep}(H),$ these Hilbert spaces consist of \textit{2-intertwiners}.
\begin{proposition}\label{tethilbertspace}
    There are linear maps
    $${\cH}_{K(\tau)}^+\to \operatorname{2Hom}_\cC(F_{[12]},F_{[12]}),\qquad \cH^-_{\tau(K)}\to \operatorname{2Hom}_\cC(F_{[03]},F_{[03]})\,.$$
    See Fig. \ref{fig:planargraphs}.
\end{proposition}
\begin{proof}
    As described below \eqref{taustate+}, \eqref{taustate-}, the tetrahedron states can be contracted through the composition of 2-morphisms. This defines a pair of linear maps
    \begin{align*}
    \cH_{K(\tau)}^+ &\to \operatorname{2Hom}_\cC(F_{[12]}\circ F_{[01]}\circ F_{[10]},F_{[32]}\circ F_{[23]}\circ F_{[12]}): \big|\tau_{K}^+\big\rangle \mapsto \sum\big|\tau_{K}^+\big\rangle  \\
        \cH^-_{K(\tau)}&\to\operatorname{2Hom}_\cC(F_{[03]}\circ F_{[10]} \circ F_{[01]}, F_{[23]}\circ F_{[32]}\circ F_{[03]}):\big|\tau_{K}^-\big\rangle \mapsto \sum\big|\tau_{K}^-\big\rangle\,.
    \end{align*}
Now recall that, for each 1-simplex $[ij]$, the adjoint of $F_{[ij]}$  is given by  $F_{[ji]}$. Moreover, the co/units  $\iota_F,e_F$ of this adjunction $F_{[ji]}\vdash F_{[ij]}$ satisfy \eqref{unitaryadjoint}, hence pre/post composing with them induce linear isomorphisms of 2-hom spaces. These then give rise to the linear isomorphisms
\begin{align*}
     \sum\big|\tau_{K}^+\big\rangle &\mapsto (e_{F_{[23]}}\circ F_{[12]})\ast (\alpha_{F_{[32]},F_{[23]},F_{[12]}}^{-1})\ast \left(\sum\big|\tau_{K}^+\big\rangle \right) \ast (\alpha_{F_{[12]},F_{[01]},F_{[01]}^\dagger}^{-1})\ast(F_{[12]}\circ \iota_{F_{[01]}}) \\
     &\qquad \qquad \qquad \qquad \qquad \qquad\qquad\in \operatorname{2Hom}_\cC(F_{[12]},F_{[12]})\\ 
    \sum \big|\tau_{K}^-\big\rangle &\mapsto (e_{F_{[32]}}\circ F_{[03]})\ast (\alpha_{F_{[23]},F_{[32]},F_{[03]}}^{-1})\ast\left(\sum \big|\tau_{K}^-\big\rangle\right)\ast (\alpha_{F_{[03]},F_{[10]},F_{[10]}^\dagger}^{-1})\ast(F_{[03]}\circ \iota_{F_{[10]}})\\
    &\qquad \qquad \qquad \qquad \qquad \qquad\qquad\in \operatorname{2Hom}_\cC(F_{[03]},F_{[03]})\,.
\end{align*}    
    Composing the two above maps then proves the statement.
\end{proof}

\begin{figure}
    \centering
    \includegraphics[width=0.8\linewidth]{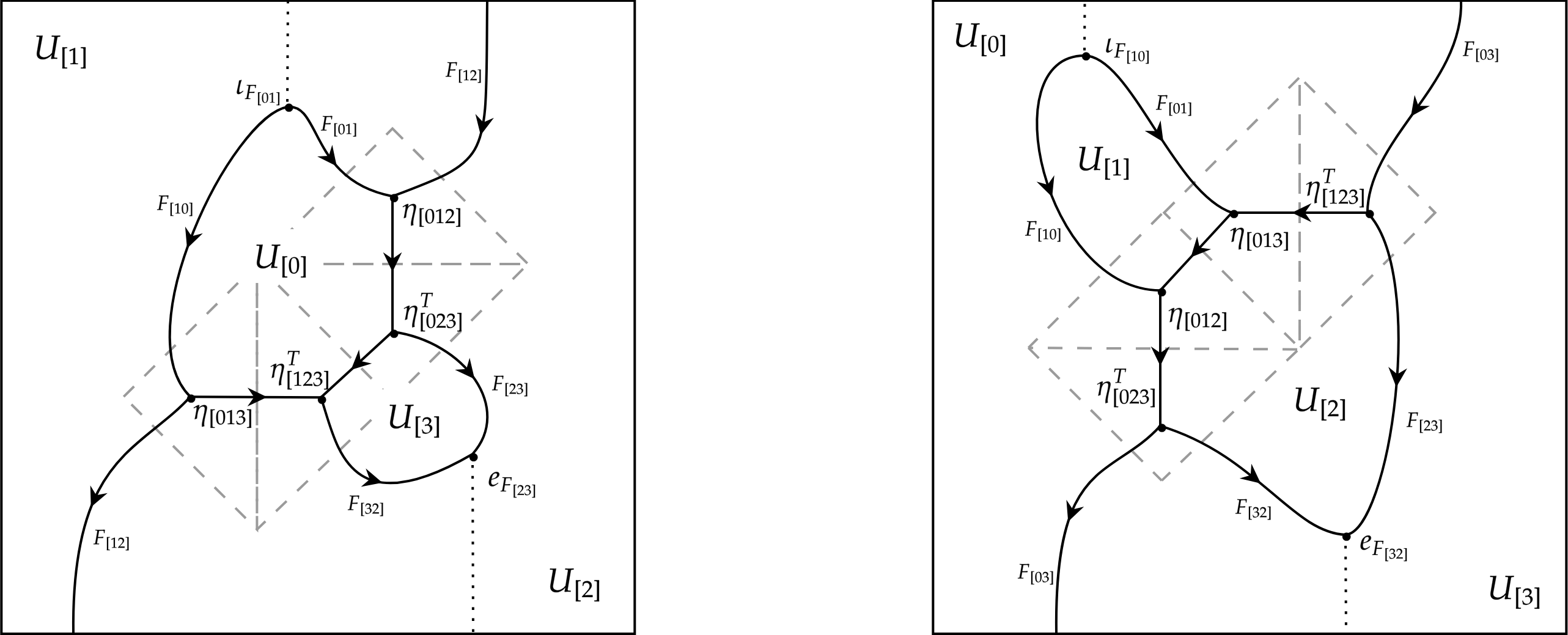}
    \caption{The tetrahedral states upon contracting pairs of its free ends as in \textbf{Proposition \ref{tethilbertspace}}. The dotted lines indicate identity 1-morphisms.} 
    \label{fig:planargraphs}
\end{figure}

If $F_{[12]}$ were a simple 1-morphism in $\cC$, then the linear map in \textbf{Proposition \ref{tethilbertspace}} extracts a scalar from the positively-oriented $\cC$-decorated tetrahedron
\begin{equation}
    \big|\tau_{K}^+\big\rangle\mapsto |K^+|\cdot \id_{F_{[12]}}\,,\qquad |K^+|\in\mathbb{C}\,;\label{tetrahedronscalar}
\end{equation}
similarly for the negatively-oriented decorated tetrahedron.\footnote{As inspired by the works on the 3d state-integral formulation of the Virasoro TQFT \cite{Andersen2018THETT,Garoufalidis:2014ifa,Shim:2026uxv,Yagi:2026nwu}, this scalar $\delta_{F_{[12]}}|K^+|$ may be referred to as the \textit{tetrahedral Boltzmann weight} of the lattice model. In contrast to those works, however, our Boltzmann weights do \textit{not} in general satisfy the (shaped) pentagon relation/3-2 Pachner move.} We will compute (the square of) this scalar $|K|=|K^+|$ in \textbf{Proposition \ref{sectionretraction}} later.

\subsection{Triangulated tetrahedral Hilbert space}\label{tethilbspace}
We call a $\cC$-decoration $K=K_s$ \textbf{simple} iff each $K([i])\in\cC$ and $K([ij])\in\operatorname{Hom}_\cC(K([i]),K([j]))$ are simple as objects/1-morphisms. For a triangulated closed oriented 3-manifold $(M^3,T_{M^3})$, a (simple) $\cC$-decoration $K: T_{M^3}\to \cC$ on the triangulation $T_{M^3}$ is a collection of (simple) $\cC$-decorations on all of its tetrahedra $\tau\subset T_{M^3}$.
\begin{definition}\label{tetspaces}
    For each oriented tetrahedron $\tau^\pm\subset T_{M^3}$, let $\cH_{K_s(\tau^\pm)} =\cH_{K_s(\tau)}^\pm$ denote the Hilbert space of 2-morphisms of the form \eqref{taustate+}, \eqref{taustate-}. The \textbf{triangulation Hilbert space on $M^3$} is given by
    \begin{equation}
        \cH_{T_{M^3}} = \bigoplus_{K_s: T_{M^3}\to \cC}\,\,\bigotimes_{\tau^\pm\subset T_{M^3}}\cH^\pm_{K_s(\tau)}\,,\label{totaltetspace}
    \end{equation}
    where the direct sum is taken over all simple $\cC$-decorations $K: T_{M^3}\to \cC$  on $T_{M^3}$.
\end{definition}
\noindent  Given a simple $\cC$-decoration $K=K_s$, the categorical Schur lemma (Prop. 1.2.19 in \cite{Douglas:2018}) implies that the tetrahedron Hilbert space $\cH_{K_s(\tau)}^\pm\neq 0$ unless $K=0$. 

However, $\cH_{K_s(\tau)}^\pm$ is \textit{not} necessarily 1-dimensional; the composition of simple 1-morphisms (which happens in both \eqref{Deltastates} and \eqref{taustate+}, \eqref{taustate-}) is in general not simple.  The best one can do is to expand 2-morphisms in a basis of simple 1-morphisms. We will demonstrate this in \S \ref{hamiltonian}.

\medskip

    The \textit{orientation-preserving} symmetries of the tetrahedron $\tau$, namely the alternating group $A_4$, acts by even permutations on the vertices of the tetrahedron. As such, tetrahedral Hilbert space should furnish an $A_4$-representation,
\begin{equation}
        f:A_4\to \operatorname{GL}(\cH_{K(\tau)}^\pm).\label{orientationrep}
    \end{equation}
We will argue in \textit{Remark \ref{branenetpivotal}} how this representation should arise canonically from the data of the tensor 2-category $\cC$. For now, we shall work under this assumption in the rest of the following.

\begin{proposition}\label{3dreversal}
    Suppose $\cC$ denote a 2-dagger rigid 2-category whose canonical pivotal structure induces  an $A_4$-representation \eqref{orientationrep}. Then, the conjugation $-^{\dagger T}:\cC\to \cC^\text{1-op,2-top}$ implements a 3-dimensional orientation reversal; namely we have a linear isomorphism
    \begin{equation*}
        (\cH^+_{K(\tau)})^{\dagger T}\xrightarrow{\sim}\cH_{\bar K(\tau)}^-
    \end{equation*}
    of the tetrahedral Hilbert spaces, where $\bar K$ denotes the negatively-oriented $\cC$-decoration.
\end{proposition}
\begin{proof}
    Recall form \textit{Remark \ref{2dagger}} that $-^{\dagger T}$ reverses the planar orientation of the pasting/string diagrams Fig. \ref{fig:planargraphs}. A direct comparison tells us that is an even cycle $x\in A_4$ on the vertex labels for which
    \begin{equation*}
        f_{x}\left(\big|\tau_{K}^+\big\rangle^{\dagger T}\right) = \big|\tau_{K}^-\big\rangle \,,
    \end{equation*}
    where $f:A_4\to \operatorname{GL}(\cH^+_\tau)$ is the linear isomorphism \eqref{orientationrep} representing the symmetries of the tetrahedron.
\end{proof}
\noindent We note that which even permutation $x\in A_4$ is used here depends on the way in which the vertices of the tetrahedron is labelled, but any linear ordering of the vertices has such an element $x$ implementing \textbf{Proposition \ref{3dreversal}}.


\begin{rmk}\label{internaltensor}
Consider the case $\cC=\operatorname{2Rep}(H)$ of main interest. The ($\mathbb{C}$-linear) tensor products involved in \S \ref{decoratedtet} can be understood as morphisms in the Delign tensor product of hom objects of $\cC$ --- for instance in \eqref{Deltastates}, we have
\begin{align*}
    \eta_{\Delta^+}\in \operatorname{Hom}_{\ttC\boxtimes \ttC}\left(F_{[02]}\boxtimes (F_{[12]}\circ F_{[01]}),(F_{[32]}\circ F_{[03]})\boxtimes F_{[02]}\right)\,,
\end{align*}
where $\ttC=\operatorname{Hom}_H(\ttU_{[0]},\ttU_{[2]})$ is understood as a hom-category object in $\cC$. Moreover, the 2-morphism adjunction $-^T:\ttC\to\ttC^\text{op}\simeq\ttC^*$ can be understood in terms of the rigid duality of the hom-object $\ttC\in\cC=\operatorname{2Rep}(H)$.  In other words, sections \S \ref{decoratedtet} and \S \ref{tethilbspace} can be completely understood in terms of the objects of $\cC=\operatorname{2Rep}(H)$. 
\end{rmk}

In \S \ref{3dlevinwen}, we will prove (\textbf{Theorem \ref{levinwentet}}) that the tetrahedron Hilbert space $\cH_\tau$ can be understood as a higher analogue of the intertwining spaces in the spin networks construction.

\section{Amplitude for 4-simplex scattering; the 20$j$-symbols}\label{20jsymbols}
Given the tetrahedral states defined above \eqref{taustate+}, \eqref{taustate-}, we now seek to construct the corresponding scattering states on a 4-simplex (aka. pentachoron, aka. hypertetrahedron, aka. pentatope, aka. pentahedroid) and compute its amplitude. To do so, we first briefly review the geometry of 4-simplices. 

Given a total ordering of the vertices, we denote by $\sigma = [01234]$ a positively oriented 4-simplex. Its boundary can be described as follows,
$$ \partial\sigma = \bigcup_{0\leq i\leq 4}[0\dots \hat i\dots 4]^{\epsilon_i} = \bigcup_{0\leq i\leq 4}\tau_i^{\epsilon_i}, $$
where "$\hat i$" denotes a missing index and $\epsilon_i=(-1)^i$ denotes their (mutual) orientations. These tetrahedra $\tau_i$ are glued upon a shared face, which is given by
\begin{equation*}
    \tau_i \cap \tau_j = \Delta_{ij}= [0\dots \hat i\dots \hat j\dots 4].
\end{equation*}
It is not hard to see that each (distinct) pair of tetrahedra $\tau_i,\tau_j\subset\partial\sigma$ share a unique face across which they are glued together.

The simplicial geometry of $\sigma$ is organized by its dual (non-planar, non-directed) 2-graph $\Gamma^5$, with (for $0\leq i<j\leq 4$) five vertices $v_i\in \Gamma^5$ each correspond to a tetrahedron $\tau_i\subset\partial\sigma$, and \textit{ten} links $e_{ij}$ each corresponding to the face $\Delta_{ij}$ shared between two tetrahedra $\tau_i,\tau_j$.
\begin{equation}\begin{tikzcd}
	&& {v_0} \\
	{v_1} &&&& {v_4} \\
	\\
	& {v_2} && {v_3}
	\arrow["{e_{02}}"{description}, no head, from=1-3, to=4-2]
	\arrow["{e_{03}}"{description}, no head, from=1-3, to=4-4]
	\arrow["{e_{01}}"{description}, no head, from=2-1, to=1-3]
	\arrow["{e_{14}}"{description}, no head, from=2-1, to=2-5]
	\arrow["{e_{12}}"{description}, no head, from=2-1, to=4-2]
	\arrow["{e_{13}}"{description}, no head, from=2-1, to=4-4]
	\arrow["{e_{40}}"',no head, from=2-5, to=1-3]
	\arrow["{e_{24}}"{description}, no head, from=2-5, to=4-2]
	\arrow["{e_{23}}"{description}, no head, from=4-2, to=4-4]
	\arrow["{e_{34}}"{description}, no head, from=4-4, to=2-5]
\end{tikzcd}\label{dual4-simplex}
\end{equation}
In the following, it will be useful to partition the set of edges of $\Gamma^5$,
\begin{equation*}
    E_\text{c} = \{e_{ij}\mid |i-j|=1\},\qquad E_\text{nc} = \{e_{ij}\mid |i-j|>1\}\,,
\end{equation*}
where $i,j\in\mathbb{Z}_4$ run over integer indices mod 4.

\medskip

To set the stage, given the 4-simplex $\sigma^\pm$ with a given orientation, we wish to construct in the following a Hilbert space $\cH_{\partial\sigma}^\pm$ associated to the boundary $\partial\sigma^\pm$ of the 4-simplex, such that the 4-simplex scattering amplitude 
\begin{equation*}
    \langle\sigma^\pm\rangle: \cH_{\partial\sigma}^\pm\to \bbC
\end{equation*}
can be understood as a linear functional. We call vectors $\psi\in \cH_{\partial\sigma}^\pm$ the {\it 4-simplex scattering states}.

\subsection{Ingredients for the 4-simplex scattering}\label{scalaringredients}
In order to construct the scattering Hilbert space, we first introduce the following ingredients. To express them in a way that is amenable to explicit computations, we shall make full use of the presemisimplicity of, and the unitary adjunctions on, $\cC$.

\begin{figure}
    \centering
    \includegraphics[width=0.8\linewidth]{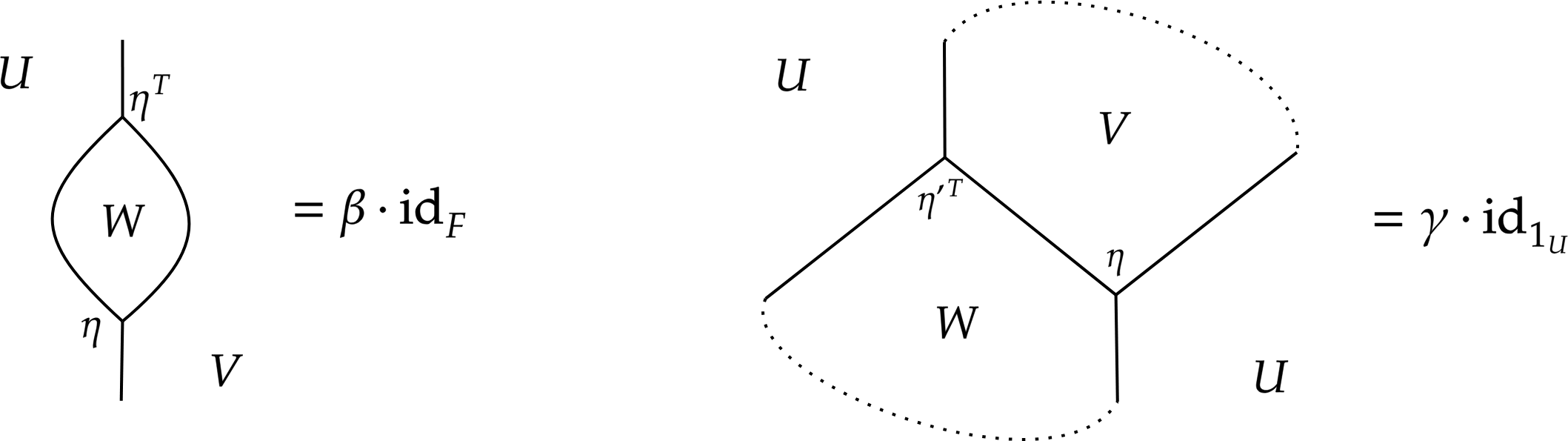}
    \caption{The planar string diagrams for the {\it \th}-net (left) and the $\theta$-net (right). They are named after their shapes. Similar $\theta$-nets have also appeared in the Crane-Yetter construction \cite{Crane:1993cm,Crane:1994ji,Crane:1993if}.}
    \label{fig:betagammascalar}
\end{figure}

Recall the composition of 2-cells along the "$A/B$-cycle" of the torus introduced in \S \ref{shadows}.
\begin{enumerate}
    \item \textbf{Scalars $\beta$ and the {\it \th}-nets ("thorn nets").} Consider composable 1-morphisms $\ttU\xrightarrow{G}\ttW\xrightarrow{H}\ttV$ and a 1-morphism $F: \ttU\to \ttV$. The pasting diagram which computes the vertical composition of a 2-morphism $\eta: H\circ G\Rightarrow F$ with its 2-dagger adjoint is given by  
\[\eta\ast_B\eta^T=\begin{tikzcd}
	\ttU && \ttV \\
	& {\ttW} \\
	\ttU && \ttV
	\arrow[""{name=0, anchor=center, inner sep=0}, "F", from=1-1, to=1-3]
	\arrow[from=1-1, to=2-2]
	\arrow["{1_\ttU}"', dotted, from=1-1, to=3-1]
	\arrow["{1_\ttV}", dotted, from=1-3, to=3-3]
	\arrow[from=2-2, to=1-3]
	\arrow[from=2-2, to=3-3]
	\arrow[from=3-1, to=2-2]
	\arrow[""{name=1, anchor=center, inner sep=0}, "F"', from=3-1, to=3-3]
	\arrow["{\eta^T}", between={0.2}{0.9}, Rightarrow, from=0, to=2-2]
	\arrow["\eta", between={0.1}{0.8}, Rightarrow, from=2-2, to=1]
\end{tikzcd}=
\begin{tikzcd}
	\ttU & \ttW & \ttV
	\arrow[from=1-1, to=1-2]
	\arrow[""{name=0, anchor=center, inner sep=0}, "F", curve={height=-24pt}, from=1-1, to=1-3]
	\arrow[""{name=1, anchor=center, inner sep=0}, "F"', curve={height=24pt}, from=1-1, to=1-3]
	\arrow[from=1-2, to=1-3]
	\arrow["{\eta^T}"', between={0.2}{1}, Rightarrow, from=0, to=1-2]
	\arrow["\eta", between={0}{0.9}, Rightarrow, from=1-2, to=1]
\end{tikzcd}\]
We call this composite a \textbf{{\it \th}-net}, and it is proportional to the identity,
\begin{equation}
    \eta\ast_B\eta^T = \beta\cdot\id_F\label{betascalar}
\end{equation}
for some scalar $\beta=|\eta|^2 $ which we call the "\textit{norm-squared}" of $\eta$. If $\cC$ is dagger unitary, $|\eta|^2=1$ for non-zero $\eta$. The pivotal string diagram for \eqref{betascalar} is displayed on the left side of Fig. \ref{fig:betagammascalar}.

\item \textbf{Scalars $\gamma$ and the $\theta$-nets ("theta nets").} Consider simple objects $\ttU,\ttV,\ttW\in\cC$ and the two unitary adjoint pairs $F^\dagger\vdash F:\ttV\to \ttU,\, G^\dagger\vdash G: \ttU\to \ttW$. For a 1-morphism $H:\ttW\to \ttV$, consider the 2-morphisms $\eta,\eta'^T$
in the following configuration 
\[\begin{tikzcd}
	{\ttU} && {\ttV} \\
	& {\ttW} && {\ttU}
	\arrow[""{name=0, anchor=center, inner sep=0}, from=1-1, to=1-3]
	\arrow[from=1-1, to=2-2]
	\arrow[from=1-3, to=2-4]
	\arrow[from=2-2, to=1-3]
	\arrow[""{name=1, anchor=center, inner sep=0}, from=2-2, to=2-4]
	\arrow["{\eta'^T}", between={0.2}{0.9}, Rightarrow, from=0, to=2-2]
	\arrow["{\eta}", between={0.1}{0.8}, Rightarrow, from=1-3, to=1]
\end{tikzcd}=
\begin{tikzcd}
	{\ttU} && {\ttV} \\
	{\ttW} && {\ttU}
	\arrow[""{name=0, anchor=center, inner sep=0}, from=1-1, to=1-3]
	\arrow[from=1-1, to=2-1]
	\arrow[from=1-3, to=2-3]
	\arrow[""{name=1, anchor=center, inner sep=0}, from=2-1, to=2-3]
	\arrow["{\eta\ast_A \eta'^T}"{description}, between={0.2}{0.9}, Rightarrow, from=0, to=1]
\end{tikzcd}\,,\]
where we have defined the composite
\begin{align*}
    \eta\ast_A\eta'^T: F \circ F^\dagger \xRightarrow{F\circ \eta'^T}F\circ (H\circ G)
    \xRightarrow{\alpha_{G,H,F}^{-1}} (F\circ  H)\circ G  \xRightarrow{\eta\circ G}G^\dagger \circ G\,
\end{align*}
which we call a \textbf{$\theta$-net}. By planar rigidity, we obtain a natural transformation
\begin{equation*}
    e_{G}\ast(\eta\ast_A\eta'^T)\ast \iota_{F}: 1_{\ttU}\Rightarrow 1_{\ttU}
\end{equation*}
which by presemisimplicity evaluates to an element in $\operatorname{End}(1_\ttU)\cong \bbC$.\footnote{This follows from Prop. 1.2.14 \cite{Douglas:2018}: given the object $\ttU$ is simple, so is its identity $1_\ttU$.} Provided $\eta,\,\eta'\neq 0$, this 2-morphism determines a non-zero complex number denoted by $\gamma\in\bbC^\times$:
\begin{equation}
    e_{G}\ast(\eta\ast_A\eta'^T)\ast \iota_{F} = {\gamma \cdot \id_{1_\ttU}}.\label{gammascalar}
\end{equation}
The pivotal string diagram for \eqref{gammascalar} is displayed on the right side of Fig. \ref{fig:betagammascalar}.
\item \textbf{Scalars $\delta$ and complete idempotents.}  Consider a rigid adjoint pair $F^\dagger\vdash F:\ttU\to \ttV$ of 1-morphisms. By dagger adjunction \eqref{unitaryadjoint}, the co/units $e_F,\iota_{F}$ form a section-retraction pair (cf. Def. 1.3.2 \cite{Douglas:2018}) up to the chosen 2-isomorphism $\phi_F:(F^\dagger)^\dagger\cong F$,
\begin{equation}
     \iota'_{F^\dagger}\ast e_F = p_F ,\qquad \iota'_{F^\dagger}= (F^\dagger\circ\phi_F)\ast\iota_{F^\dagger}\,\label{deltascalars}
\end{equation}
where  $p_F\in\operatorname{2Hom}_\cC(F^\dagger\circ F,F^\dagger\circ F)$ is an idempotent. By local idempotent completeness, these idempotents form a complete orthonormal basis 
\begin{equation}
    \sum_{F\in\operatorname{Irr}\cC(\ttU,\ttV)} p_F = \id_A,\qquad p_Fp_G = \delta_{F,G}p_F \label{deltadecouple}
\end{equation}
upon normalizing by the scalar $\delta_F = \operatorname{dim }F\in\bbC^\times$.  Here, $A$ is the universal coend
\begin{equation*}
    A= \int^{F\in\operatorname{Hom}_\cC(\ttU,\ttV)}F^\dagger\circ F \cong \bigoplus_{F\in\operatorname{Irr}\cC(\ttU,\ttV)} F^\dagger\circ F
\end{equation*}
in $\operatorname{End}_\cC(\ttU)$. 
Now \eqref{deltadecouple} allows us to perform an important graphical move: the sum over parallel surfaces decorated by adjoint 1-morphisms $F^\dagger\vdash F$ can be merged. This is displayed in Fig. \ref{fig:deltascalar}. 
\end{enumerate}

\begin{figure}
    \centering
    \includegraphics[width=0.7\linewidth]{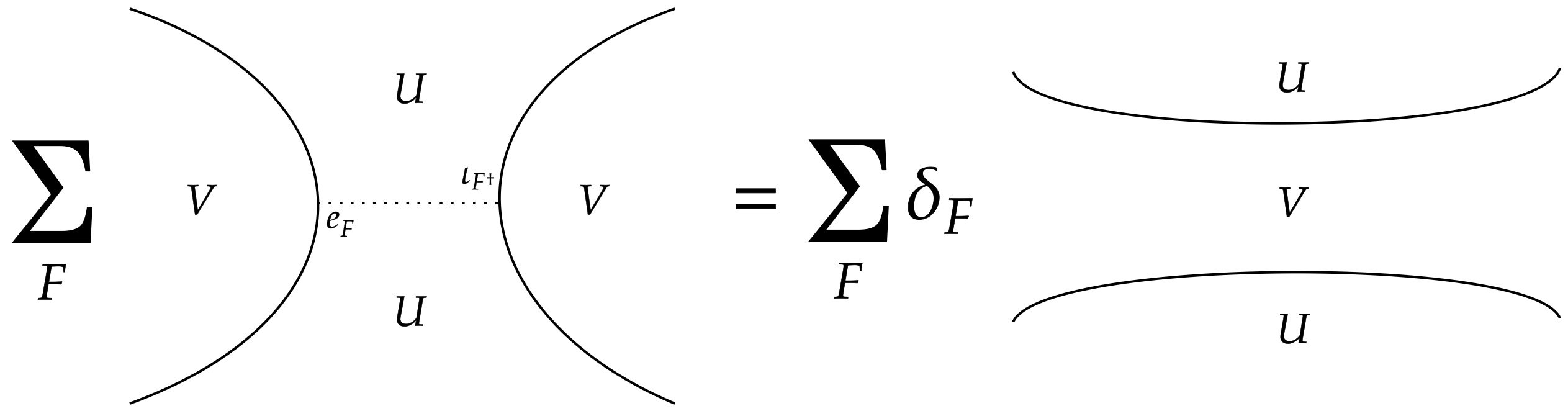}
    \caption{The surface diagram arising from the operation \eqref{deltadecouple}. The scalar $\delta_F=\operatorname{dim}F$ appears due to the normalization of the projectors $p_F$ \eqref{deltascalars}. This is precisely the {\bf Kirby colour} of \cite{Liu:2023dhj}.}
    \label{fig:deltascalar}
\end{figure}


We show in \S \ref{eval} that these scalars \eqref{betascalar}, \eqref{deltascalars}, \eqref{gammascalar} serve as the building blocks of the 4-simplex scattering amplitudes.

\subsection{Scattering states on the positively-oriented 4-simplex}
For each $0\leq i\leq 4$, we will fix notation such that $\tau^+_i=\tau_i = [jklm]$, where $j,k,l,m\neq i$ are ordered integers mod 4; in other words, we have explicitly
\begin{gather*}
    \tau_0= [1234],\qquad \tau_1= [2340],\qquad \tau_2=[3401],\\
    \tau_3=[4012],\qquad \tau_4=[0123].
\end{gather*}
The gluing law corresponding to the edge $e_{ij}$ are implemented through a PL attaching map $f_{{ij}}: \tau_i^{\epsilon_i} \to \tau_j^{\epsilon_j}$ in accordance with the graph $\Gamma^5$, where $f_{{ij}}$ is orientation-preserving iff $\epsilon_i\epsilon_j=+$, and orientation-reversing otherwise.

By considering the states \eqref{taustate+}, \eqref{taustate-} associated to $\tau_i$, we see that this PL homeomorphism corresponds to 2-morphism compositions. For instance, attaching $\tau_4^+,\tau_0^+$ at the face $[123]$ amounts to the contraction
\[\hspace{-0.9cm}\begin{tikzcd}
	&&& {\ttU_{[1]}} \\
	&& {\ttU_{[0]}} && {\ttU_{[2]}} \\
	& {\ttU_{[1]}} && {\ttU_{[3]}} \\
	&& {\ttU_{[2]}} \\
	& {\ttU_{[1]}} && {\ttU_{[3]}} \\
	{\ttU_{[2]}} && {\ttU_{[4]}} \\
	& {\ttU_{[3]}}
	\arrow[from=1-4, to=2-5]
	\arrow[from=2-3, to=1-4]
	\arrow[""{name=0, anchor=center, inner sep=0}, from=2-3, to=2-5]
	\arrow[from=2-3, to=3-4]
	\arrow[from=3-2, to=2-3]
	\arrow[""{name=1, anchor=center, inner sep=0}, from=3-2, to=3-4]
	\arrow[from=3-2, to=4-3]
	\arrow["{1_{\ttU_{[1]}}}"', dotted, from=3-2, to=5-2]
	\arrow[from=3-4, to=2-5]
	\arrow["{1_{\ttU_{[3]}}}", dotted, from=3-4, to=5-4]
	\arrow[from=4-3, to=3-4]
	\arrow[from=4-3, to=5-4]
	\arrow[from=5-2, to=4-3]
	\arrow[""{name=2, anchor=center, inner sep=0}, from=5-2, to=5-4]
	\arrow[from=5-2, to=6-3]
	\arrow[from=6-1, to=5-2]
	\arrow[""{name=3, anchor=center, inner sep=0}, from=6-1, to=6-3]
	\arrow[from=6-1, to=7-2]
	\arrow[from=6-3, to=5-4]
	\arrow[from=7-2, to=6-3]
	\arrow["{\eta_{[012]}}", between={0.1}{0.8}, Rightarrow, from=1-4, to=0]
	\arrow["{\eta_{[032]}^T}", between={0.2}{0.9}, Rightarrow, from=0, to=3-4]
	\arrow["{\eta_{[103]}}", between={0.1}{0.8}, Rightarrow, from=2-3, to=1]
	\arrow["{\eta_{[123]}^T}", between={0.2}{0.9}, Rightarrow, from=1, to=4-3]
	\arrow["{\eta_{[123]}}", between={0.1}{0.8}, Rightarrow, from=4-3, to=2]
	\arrow["{\eta_{143}^T}", between={0.2}{0.9}, Rightarrow, from=2, to=6-3]
	\arrow["{\eta_{[214]}}", between={0}{0.8}, Rightarrow, from=5-2, to=3]
	\arrow["{\eta_{[234]}^T}", between={0.2}{1}, Rightarrow, from=3, to=7-2]
\end{tikzcd} =
\begin{tikzcd}
	&&&& {\ttU_{[1]}} \\
	&&& {\ttU_{[0]}} && {\ttU_{[2]}} \\
	& {\ttU_{[1]}} &&& {\ttU_{[3]}} \\
	{\ttU_{[2]}} && {\ttU_{[4]}} \\
	& {\ttU_{[3]}}
	\arrow[from=1-5, to=2-6]
	\arrow[from=2-4, to=1-5]
	\arrow[""{name=0, anchor=center, inner sep=0}, from=2-4, to=2-6]
	\arrow[from=2-4, to=3-5]
	\arrow[from=3-2, to=2-4]
	\arrow[""{name=1, anchor=center, inner sep=0}, curve={height=-12pt}, from=3-2, to=3-5]
	\arrow[""{name=2, anchor=center, inner sep=0}, curve={height=12pt}, from=3-2, to=3-5]
	\arrow[from=3-2, to=4-3]
	\arrow[from=3-5, to=2-6]
	\arrow[from=4-1, to=3-2]
	\arrow[""{name=3, anchor=center, inner sep=0}, from=4-1, to=4-3]
	\arrow[from=4-1, to=5-2]
	\arrow[from=4-3, to=3-5]
	\arrow[from=5-2, to=4-3]
	\arrow["{\eta_{[012]}}", between={0.1}{0.8}, Rightarrow, from=1-5, to=0]
	\arrow["{\eta_{[032]}^T}", between={0.2}{0.9}, Rightarrow, from=0, to=3-5]
	\arrow["{\eta_{[103]}}"{pos=0}, between={0.1}{0.8}, Rightarrow, from=2-4, to=1]
	\arrow["{\eta_{[123]}^T\ast_B\eta_{[123]}}", shift right=4, between={0.2}{0.8}, Rightarrow, from=1, to=2]
	\arrow["{\eta_{[214]}}", between={0}{0.8}, Rightarrow, from=3-2, to=3]
	\arrow["{\eta_{[143]}^T}"'{pos=1}, between={0.2}{1}, Rightarrow, from=2, to=4-3]
	\arrow["{\eta_{[234]}^T}", between={0.2}{1}, Rightarrow, from=3, to=5-2]
\end{tikzcd}\]
On the other hand, attaching $\tau_4,\tau_2$ at the face $[301]$ corresponds to the \textit{horizontal} composition 
\begin{equation}\hspace{-1.7cm}\begin{tikzcd}
	&& {\ttU_{[4]}} && {\ttU_{[1]}} \\
	& {\ttU_{[3]}} && {\ttU_{[0]}} && {\ttU_{[2]}} \\
	{\ttU_{[4]}} && {\ttU_{[1]}} && {\ttU_{[3]}} \\
	& {\ttU_{[0]}} && {\ttU_{[2]}}
	\arrow[from=1-3, to=2-4]
	\arrow[from=1-5, to=2-6]
	\arrow[from=2-2, to=1-3]
	\arrow[""{name=0, anchor=center, inner sep=0}, from=2-2, to=2-4]
	\arrow[from=2-2, to=3-3]
	\arrow[from=2-4, to=1-5]
	\arrow[""{name=1, anchor=center, inner sep=0}, from=2-4, to=2-6]
	\arrow[from=2-4, to=3-5]
	\arrow[from=3-1, to=2-2]
	\arrow[""{name=2, anchor=center, inner sep=0}, from=3-1, to=3-3]
	\arrow[from=3-1, to=4-2]
	\arrow[from=3-3, to=2-4]
	\arrow[""{name=3, anchor=center, inner sep=0}, from=3-3, to=3-5]
	\arrow[from=3-3, to=4-4]
	\arrow[from=3-5, to=2-6]
	\arrow[from=4-2, to=3-3]
	\arrow[from=4-4, to=3-5]
	\arrow["{\eta_{[320]}}", between={0.1}{0.8}, Rightarrow, from=1-3, to=0]
	\arrow["{\eta_{[012]}}", between={0.1}{0.8}, Rightarrow, from=1-5, to=1]
	\arrow["{\eta_{[310]}^T}", between={0.2}{0.9}, Rightarrow, from=0, to=3-3]
	\arrow["{\eta_{[431]}}", between={0.1}{0.8}, Rightarrow, from=2-2, to=2]
	\arrow["{\eta_{[032]}^T}"', between={0.2}{0.9}, Rightarrow, from=1, to=3-5]
	\arrow["{\eta_{[103]}}"', between={0.1}{0.8}, Rightarrow, from=2-4, to=3]
	\arrow["{\eta_{[401]}^T}"', between={0.2}{0.9}, Rightarrow, from=2, to=4-2]
	\arrow["{\eta_{[123]}^T}", between={0.2}{1}, Rightarrow, from=3, to=4-4]
\end{tikzcd}=
\begin{tikzcd}
	&& {\ttU_{[4]}} && {\ttU_{[1]}} \\
	& {\ttU_{[3]}} && {\ttU_{[0]}} && {\ttU_{[2]}} \\
	{\ttU_{[4]}} && {\ttU_{[1]}} && {\ttU_{[3]}} \\
	& {\ttU_{[0]}} && {\ttU_{[2]}}
	\arrow[from=1-3, to=2-4]
	\arrow[from=1-5, to=2-6]
	\arrow[from=2-2, to=1-3]
	\arrow[""{name=0, anchor=center, inner sep=0}, from=2-2, to=2-4]
	\arrow[from=2-2, to=3-3]
	\arrow[from=2-4, to=1-5]
	\arrow[""{name=1, anchor=center, inner sep=0}, from=2-4, to=2-6]
	\arrow[from=2-4, to=3-5]
	\arrow[from=3-1, to=2-2]
	\arrow[""{name=2, anchor=center, inner sep=0}, from=3-1, to=3-3]
	\arrow[from=3-1, to=4-2]
	\arrow[""{name=3, anchor=center, inner sep=0}, from=3-3, to=3-5]
	\arrow[from=3-3, to=4-4]
	\arrow[from=3-5, to=2-6]
	\arrow[from=4-2, to=3-3]
	\arrow[from=4-4, to=3-5]
	\arrow["{\eta_{[320]}}", between={0.1}{0.8}, Rightarrow, from=1-3, to=0]
	\arrow["{\eta_{[012]}}", between={0.1}{0.8}, Rightarrow, from=1-5, to=1]
	\arrow["{\eta_{[431]}}", between={0.1}{0.8}, Rightarrow, from=2-2, to=2]
	\arrow["{\eta^T_{[310]}\ast_A\eta_{[103]}}"{description}, between={0.1}{0.8}, Rightarrow, from=0, to=3]
	\arrow["{\eta_{[032]}^T}"', between={0.2}{0.9}, Rightarrow, from=1, to=3-5]
	\arrow["{\eta_{[401]}^T}"', between={0.2}{0.9}, Rightarrow, from=2, to=4-2]
	\arrow["{\eta_{[123]}^T}", between={0.2}{1}, Rightarrow, from=3, to=4-4]
\end{tikzcd}\label{thetanethorizontalcomposite}\end{equation}
from which one can obtain the $\theta$-nets; see also \S \ref{pentagonatorscalars} later.

\begin{table}
    \centering
    \begin{tabular}{c|c|c}
         & $|i-j|=1$ & $|i-j|>1$\\
        \hline
        gluing type & cyclic $e_{ij}\in E_\text{c}$  & non-cyclic $e_{ij}\in E_\text{nc}$ \\ 
        composition type & {\it \th}-net $\eta^T\ast_B\eta$ & $\theta$-net $\eta'^T\ast_A \eta$\\
        attaching map $f_{{ij}}$ & reverses orientation & preserves orientation\\ 
        scalar factor & $\beta$ & $\gamma$
    \end{tabular}
    \caption{This table displays the 2-morphism compositions/nets, as well as the resulting scalar factors, which correspond to the gluing laws between the two tetrahedra $\tau_i,\tau_j$.}
    \label{tab:types}
\end{table}

By contracting all tetrahedron states on $\tau_0,\dots,\tau_4$ with respect to the prescription displayed in Tab. \ref{tab:types}, we complete the 4-simplex boundary:
\[\begin{tikzcd}
	&&&& {\ttU_{[3]}} && {\ttU_{[0]}} \\
	&&& {\ttU_{[2]}} && {\ttU_{[4]}} && {\ttU_{[1]}} \\
	&& {\ttU_{[3]}} && {\ttU_{[0]}} && {\ttU_{[2]}} \\
	{\ttU_{[4]}} &&& {\ttU_{[1]}} && {\ttU_{[3]}} \\
	& {\ttU_{[0]}} & {\ttU_{[2]}} && {\ttU_{[4]}} \\
	&&& {\ttU_{[3]}}
	\arrow[from=1-5, to=2-6]
	\arrow[from=1-7, to=2-8]
	\arrow[from=2-4, to=1-5]
	\arrow[""{name=0, anchor=center, inner sep=0}, from=2-4, to=2-6]
	\arrow[from=2-4, to=3-5]
	\arrow[from=2-6, to=1-7]
	\arrow[""{name=1, anchor=center, inner sep=0}, from=2-6, to=2-8]
	\arrow[from=2-6, to=3-7]
	\arrow[from=3-3, to=2-4]
	\arrow[""{name=2, anchor=center, inner sep=0}, curve={height=-6pt}, from=3-3, to=3-5]
	\arrow[""{name=3, anchor=center, inner sep=0}, curve={height=6pt}, from=3-3, to=3-5]
	\arrow[from=3-3, to=4-4]
	\arrow[""{name=4, anchor=center, inner sep=0}, curve={height=-6pt}, from=3-5, to=3-7]
	\arrow[""{name=5, anchor=center, inner sep=0}, curve={height=6pt}, from=3-5, to=3-7]
	\arrow[from=3-5, to=4-6]
	\arrow[from=3-7, to=2-8]
	\arrow[from=4-1, to=3-3]
	\arrow[""{name=6, anchor=center, inner sep=0}, from=4-1, to=4-4]
	\arrow[from=4-1, to=5-2]
	\arrow[""{name=7, anchor=center, inner sep=0}, curve={height=-6pt}, from=4-4, to=4-6]
	\arrow[""{name=8, anchor=center, inner sep=0}, curve={height=6pt}, from=4-4, to=4-6]
	\arrow[from=4-4, to=5-5]
	\arrow[from=4-6, to=3-7]
	\arrow[from=5-2, to=4-4]
	\arrow[from=5-3, to=4-4]
	\arrow[""{name=9, anchor=center, inner sep=0}, from=5-3, to=5-5]
	\arrow[from=5-3, to=6-4]
	\arrow[from=5-5, to=4-6]
	\arrow[from=6-4, to=5-5]
	\arrow["{\eta_{[234]}}", color={rgb,255:red,255;green,51;blue,68}, between={0.1}{0.8}, Rightarrow, from=1-5, to=0]
	\arrow["{\eta_{[401]}}"', color={rgb,255:red,51;green,54;blue,255}, between={0.1}{0.8}, Rightarrow, from=1-7, to=1]
	\arrow["{\eta_{[320]}}", color={rgb,255:red,214;green,153;blue,92}, between={0.1}{0.8}, Rightarrow, from=2-4, to=2]
	\arrow[between={0.2}{0.8}, Rightarrow, from=0, to=4]
	\arrow["{\eta_{[421]}^T}", color={rgb,255:red,0;green,219;blue,33}, between={0.2}{0.9}, Rightarrow, from=1, to=3-7]
	\arrow[between={0.2}{0.8}, Rightarrow, from=2, to=3]
	\arrow["{\eta_{[431]}}", color={rgb,255:red,150;green,51;blue,255}, between={0.1}{0.8}, Rightarrow, from=3-3, to=6]
	\arrow[between={0.2}{0.8}, Rightarrow, from=3, to=7]
	\arrow[between={0.2}{0.8}, Rightarrow, from=4, to=5]
	\arrow["{\eta_{[032]}^T}", color={rgb,255:red,214;green,153;blue,92}, between={0.2}{0.9}, Rightarrow, from=5, to=4-6]
	\arrow["{\eta_{[401]}^T}"', color={rgb,255:red,51;green,54;blue,255}, between={0.2}{0.9}, Rightarrow, from=6, to=5-2]
	\arrow[between={0.2}{0.8}, Rightarrow, from=7, to=8]
	\arrow["{\eta_{[214]}}"', color={rgb,255:red,0;green,219;blue,33}, between={0.2}{0.8}, Rightarrow, from=4-4, to=9]
	\arrow["{\eta_{[134]}^T}", color={rgb,255:red,150;green,51;blue,255}, between={0.2}{0.8}, Rightarrow, from=8, to=5-5]
	\arrow["{\eta_{[234]}^T}", color={rgb,255:red,255;green,51;blue,68}, between={0.2}{1}, Rightarrow, from=9, to=6-4]
\end{tikzcd}\]
where the pairs of 2-morphisms that can be contracted into \textit{\th}-nets or $\theta$-nets are colour-coded.


To achieve this contraction, we leverage the cyclic symmetry \eqref{cublicycle} afforded by the {toroidal trace} $\Tr$ \S \ref{torustrace}. Modulo $\operatorname{ker}\Tr$, the above 2-morphism is identified with the following,
\begin{equation} \psi_\mathcal{K} = 
\begin{tikzcd}
	&& {\ttU_{[2]}} && {\ttU_{[4]}} && {\ttU_{[1]}} \\
	\\
	{\ttU_{[3]}} && {\ttU_{[0]}} && {\ttU_{[2]}} && {\ttU_{[4]}} \\
	\\
	{\ttU_{[1]}} && {\ttU_{[3]}} && {\ttU_{[0]}} \\
	\\
	{\ttU_{[4]}} && {\ttU_{[1]}}
	\arrow[""{name=0, anchor=center, inner sep=0}, curve={height=-12pt}, from=1-3, to=1-5]
	\arrow[""{name=1, anchor=center, inner sep=0}, curve={height=12pt}, from=1-3, to=1-5]
	\arrow[from=1-3, to=3-3]
	\arrow[""{name=2, anchor=center, inner sep=0}, curve={height=-12pt}, from=1-5, to=1-7]
	\arrow[""{name=3, anchor=center, inner sep=0}, curve={height=12pt}, from=1-5, to=1-7]
	\arrow[from=1-5, to=3-5]
	\arrow[from=1-7, to=3-7]
	\arrow[""{name=4, anchor=center, inner sep=0}, curve={height=-12pt}, from=3-1, to=3-3]
	\arrow[""{name=5, anchor=center, inner sep=0}, curve={height=12pt}, from=3-1, to=3-3]
	\arrow[from=3-1, to=5-1]
	\arrow[""{name=6, anchor=center, inner sep=0}, curve={height=-12pt}, from=3-3, to=3-5]
	\arrow[""{name=7, anchor=center, inner sep=0}, curve={height=12pt}, from=3-3, to=3-5]
	\arrow[from=3-3, to=5-3]
	\arrow[""{name=8, anchor=center, inner sep=0}, from=3-5, to=3-7]
	\arrow[from=3-5, to=5-5]
	\arrow[""{name=9, anchor=center, inner sep=0}, curve={height=-12pt}, from=5-1, to=5-3]
	\arrow[""{name=10, anchor=center, inner sep=0}, curve={height=12pt}, from=5-1, to=5-3]
	\arrow[from=5-1, to=7-1]
	\arrow[""{name=11, anchor=center, inner sep=0}, from=5-3, to=5-5]
	\arrow[from=5-3, to=7-3]
	\arrow[""{name=12, anchor=center, inner sep=0}, from=7-1, to=7-3]
	\arrow[between={0.2}{0.8}, Rightarrow, from=0, to=1]
	\arrow["{\eta_{[204]}^T\ast_A \eta_{[042]}}"{description}, between={0.1}{0.9}, Rightarrow, from=1, to=6]
	\arrow[between={0.2}{0.8}, Rightarrow, from=2, to=3]
	\arrow["{\eta_{[421]}^T\ast_A \eta_{[214]}}"{description}, between={0.1}{0.9}, Rightarrow, from=3, to=8]
	\arrow[between={0.2}{0.8}, Rightarrow, from=4, to=5]
	\arrow["{\eta_{[310]}^T\ast_A\eta_{[103]}}"{description}, between={0.1}{0.9}, Rightarrow, from=5, to=9]
	\arrow[between={0.2}{0.8}, Rightarrow, from=6, to=7]
	\arrow["{\eta_{[032]}^T\ast_A \eta_{[320]}}"{description}, between={0.1}{0.9}, Rightarrow, from=7, to=11]
	\arrow[between={0.2}{0.8}, Rightarrow, from=9, to=10]
	\arrow["{\eta_{[134]}^T\ast_A \eta_{[431]}}"{description}, between={0.1}{0.9}, Rightarrow, from=10, to=12]
\end{tikzcd}\label{scatteringstate}
\end{equation}
The \textit{non-planar} tensor network and string diagram presentations for $\psi_\mathcal{K}$ are shown in Fig. \ref{fig:scatteringvector}. 

\begin{figure}
    \centering
    \includegraphics[width=1\linewidth]{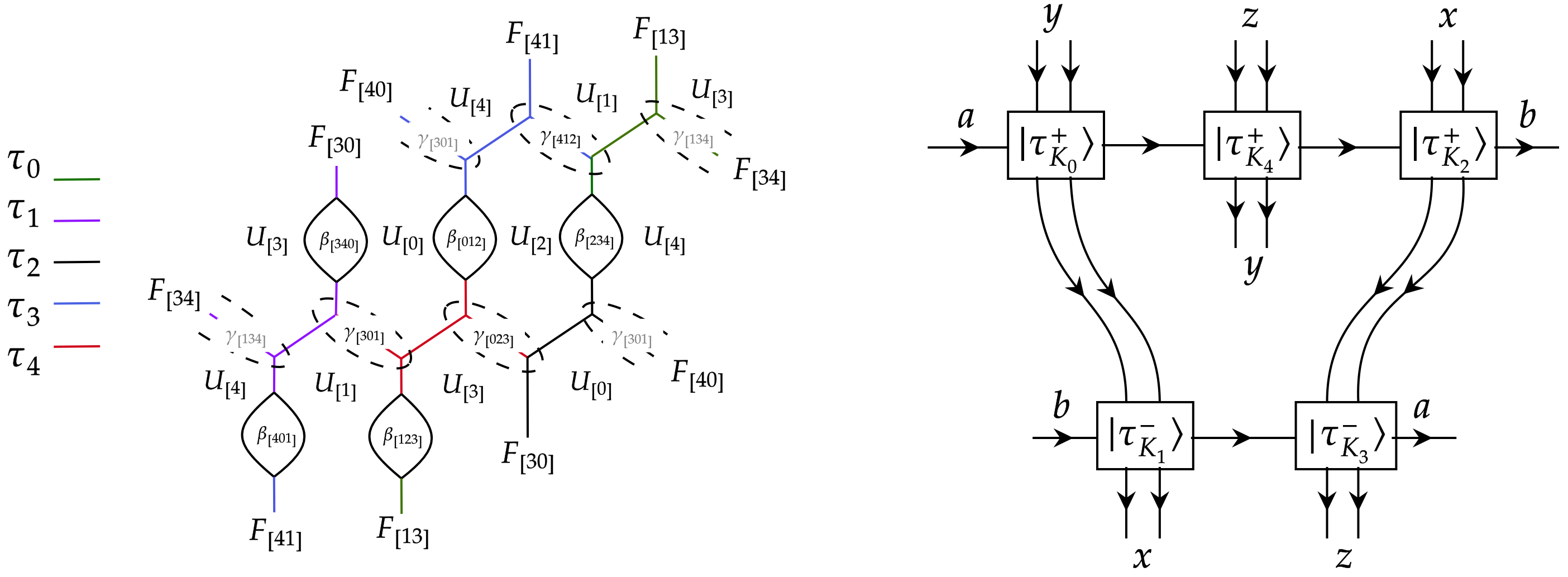}
    \caption{The string diagram (left) and tensor network (right) representing a 4-simplex scattering vector $\psi_\mathcal{K}\in\cH_{\partial\sigma}^+$ \eqref{scatteringstate}. On the left, the colours label contributions from the different tetrahedra. On the right, the symbols $a,b\,,x,y,z$ label the external tensor legs that can be contracted to form the scattering amplitude; the way these labels are contracted makes the diagram non-planar.}
    \label{fig:scatteringvector}
\end{figure}

\begin{definition}\label{4simpscattering}
    The \textbf{4-simplex scattering states} are given by 2-morphisms of the form  in \eqref{scatteringstate},
    \begin{equation*}
        \psi_{\mathcal{K}_s} \in \cH_{\partial\sigma}^+ \subset \bigotimes_{i=0}^5 \cH_{\tau_i}^{\epsilon_i}\,,
    \end{equation*}
    where $\mathcal{K}_s:\partial\sigma^+\to \cC$ is a simple  $\cC$-decoration on the boundary of a positively-oriented 4-simplex $\sigma^+$. 
\end{definition}
Throughout the following, we will consider the Hilbert space $\cH_{\partial\sigma}^+$ modulo $\operatorname{ker}\Tr$. 

\subsubsection{Recoupling of $\theta$-nets and the pentagonators}\label{pentagonatorscalars}
Experts on TQFTs may wonder where the \textit{pentagonator} datum $\pi$ of $\cC$ appears in the 4-simplex state $\psi_\mathcal{K}$, as it plays the role of the "10$j$-symbols" in the usual 4d state-sum construction \cite{Douglas:2018,Mackaay:ek,Crane:1994ty,Inamura:2023qzl}. Indeed, the 4-simplex scattering state $\psi_\mathcal{K}$ \eqref{scatteringstate} \textit{does} in fact involve the pentagonator $\pi$, albeit indirectly. 

To see this, let us first focus on pasting diagrams of the form 
\[\begin{tikzcd}
	& {\ttU_2} && {\ttU_4} \\
	{\ttU_1} && {\ttU_3} && {\ttU_5}
	\arrow[""{name=0, anchor=center, inner sep=0}, from=1-2, to=1-4]
	\arrow["{F_2}", from=1-2, to=2-3]
	\arrow["{F_4}", from=1-4, to=2-5]
	\arrow["{F_1}", from=2-1, to=1-2]
	\arrow[""{name=1, anchor=center, inner sep=0}, from=2-1, to=2-3]
	\arrow["{F_3}"', from=2-3, to=1-4]
	\arrow[""{name=2, anchor=center, inner sep=0}, from=2-3, to=2-5]
	\arrow["\eta_1",between={0}{0.8}, Rightarrow, from=1-2, to=1]
	\arrow["\eta_2",between={0.2}{1}, Rightarrow, from=0, to=2-3]
	\arrow["\eta_3",between={0}{0.8}, Rightarrow, from=1-4, to=2]
\end{tikzcd}\]
involving the horizontal composition $\eta_1\ast_A\eta_2\ast_A\eta_3$ of three 2-morphisms. There are various ways to evaluate such a pasting diagram, two of which is given by simply forming $\eta_1\ast_A\eta_2$ or $\eta_2\ast_A\eta_3$ first,
\[
\begin{tikzcd}
	& {\ttU_2} && {\ttU_4} \\
	{\ttU_1} && {\ttU_3} && {\ttU_5}
	\arrow[""{name=0, anchor=center, inner sep=0}, from=1-2, to=1-4]
	\arrow["{F_4}", from=1-4, to=2-5]
	\arrow["{F_1}", from=2-1, to=1-2]
	\arrow[""{name=1, anchor=center, inner sep=0}, from=2-1, to=2-3]
	\arrow["{F_3}"', from=2-3, to=1-4]
	\arrow[""{name=2, anchor=center, inner sep=0}, from=2-3, to=2-5]
	\arrow[between={0.2}{0.8}, Rightarrow, from=0, to=1]
	\arrow[between={0}{0.8}, Rightarrow, from=1-4, to=2]
\end{tikzcd},\qquad \begin{tikzcd}
	& {\ttU_2} && {\ttU_4} \\
	{\ttU_1} && {\ttU_3} && {\ttU_5}
	\arrow[""{name=0, anchor=center, inner sep=0}, from=1-2, to=1-4]
	\arrow["{F_2}", from=1-2, to=2-3]
	\arrow["{F_4}", from=1-4, to=2-5]
	\arrow["{F_1}", from=2-1, to=1-2]
	\arrow[""{name=1, anchor=center, inner sep=0}, from=2-1, to=2-3]
	\arrow[""{name=2, anchor=center, inner sep=0}, from=2-3, to=2-5]
	\arrow[between={0}{0.8}, Rightarrow, from=1-2, to=1]
	\arrow[between={0.2}{0.8}, Rightarrow, from=0, to=2]
\end{tikzcd}\]
which involve the following composition associators respectively
\begin{equation*}
    \alpha_{F_1,F_2,F_3}\circ F_4,\qquad F_1\circ\alpha_{F_2,F_3,F_4}\,.
\end{equation*}
Similarly, the remaining other ways to evaluate this diagram 

involve the associators 
\begin{equation*}
    \alpha_{F_1\circ F_2,F_3,F_4},\qquad \alpha_{F_1,F_2\circ F_3,F_4},\qquad \alpha_{F_1,F_2,F_3\circ F_4}.
\end{equation*}
The consistency of three horizontal composite 2-morphisms $\eta_1\ast_A\eta_2\ast_A\eta_3$ therefore hinges on the pentagon axiom satisfied by the composition associator $\alpha$. 

However, for self-enriched $\cC$ having hom-objects (or more generally, $\cC$ can be enriched in some other monoidal 2-category $\cV$), the pentagon axiom for $\alpha$ involves the {\it pentagonator} $\pi$ of $\cC$ itself \cite{GARNER20161},
    \begin{align}
        \pi_{\ttU_{12},\ttU_{23},\ttU_{34},\ttU_{45}}\ast (F_1\circ \alpha_{F_2,F_3,F_4})\,\ast &\, \alpha_{F_1,F_2\circ F_3,F_4}\ast (\alpha_{F_1,F_2,F_3}\circ F_4) \nonumber\\
&=\alpha_{F_1,F_2,F_3\circ F_4}\ast\alpha_{F_1\circ F_2,F_3,F_4} \,,\label{pentagonaxiom}
    \end{align}
    where we have used a shorthand $\ttU_{ij} = \operatorname{Hom}_\cC(\ttU_i,\ttU_j)=\cC(\ttU_i,\ttU_j)$ for the enriched hom's in $\cC$. What this means is that \textbf{a pentagonator $\pi$ hides in each evaluation of $\eta_1\ast_A\eta_2\ast_A\eta_3$}!

\medskip

The relevance of this fact for us lies in \eqref{thetanethorizontalcomposite}, where two tetrahedral states (say on $\tau_2,\tau_4$) are attached. Here, the goal is to evaluate the horizontal composition of the four relevant 2-morphisms to
\begin{equation*}
    \eta_{[431]}\ast_A( \eta_{[310]}^T\ast_A\eta_{[103]})\ast_A\eta_{[032]}^T
\end{equation*}
in order to form a $\theta$-net in the middle. For this, we must use the pentagon axioms  \eqref{pentagonaxiom} to rewrite composites of the associators in the tetrahedral states \eqref{taustate+}, which leads to the appearance of pentagonators in \eqref{thetanethorizontalcomposite}. More precisely, given the shorthand
$\pi_{\ttU_{[i][j]},\ttU_{[j][k]},\ttU_{[k][l]},\ttU_{[l] [j]}} = \pi_{[ijkl]} $
\eqref{thetanethorizontalcomposite} should take the form
\begin{equation}\pi_{[4310]}\ast \begin{tikzcd}
	&& {\ttU_{[4]}} && {\ttU_{[1]}} \\
	& {\ttU_{[3]}} && {\ttU_{[0]}} && {\ttU_{[2]}} \\
	{\ttU_{[4]}} && {\ttU_{[1]}} && {\ttU_{[3]}} \\
	& {\ttU_{[0]}} && {\ttU_{[2]}}
	\arrow[from=1-3, to=2-4]
	\arrow[from=1-5, to=2-6]
	\arrow[from=2-2, to=1-3]
	\arrow[""{name=0, anchor=center, inner sep=0}, from=2-2, to=2-4]
	\arrow[from=2-2, to=3-3]
	\arrow[from=2-4, to=1-5]
	\arrow[""{name=1, anchor=center, inner sep=0}, from=2-4, to=2-6]
	\arrow[from=2-4, to=3-5]
	\arrow[from=3-1, to=2-2]
	\arrow[""{name=2, anchor=center, inner sep=0}, from=3-1, to=3-3]
	\arrow[from=3-1, to=4-2]
	\arrow[""{name=3, anchor=center, inner sep=0}, from=3-3, to=3-5]
	\arrow[from=3-3, to=4-4]
	\arrow[from=3-5, to=2-6]
	\arrow[from=4-2, to=3-3]
	\arrow[from=4-4, to=3-5]
	\arrow["{\eta_{[320]}}", between={0.1}{0.8}, Rightarrow, from=1-3, to=0]
	\arrow["{\eta_{[012]}}", between={0.1}{0.8}, Rightarrow, from=1-5, to=1]
	\arrow["{\eta_{[431]}}", between={0.1}{0.8}, Rightarrow, from=2-2, to=2]
	\arrow["{\eta^T_{[310]}\ast_A\eta_{[103]}}"{description}, between={0.1}{0.8}, Rightarrow, from=0, to=3]
	\arrow["{\eta_{[032]}^T}"', between={0.2}{0.9}, Rightarrow, from=1, to=3-5]
	\arrow["{\eta_{[401]}^T}"', between={0.2}{0.9}, Rightarrow, from=2, to=4-2]
	\arrow["{\eta_{[123]}^T}", between={0.2}{1}, Rightarrow, from=3, to=4-4]
\end{tikzcd}\label{withpentagonators}\end{equation}
Note $\pi$ only depend on the hom-categories $\ttU_{[i][j]} = \operatorname{Hom}_\cC(\ttU_{[i]},\ttU_{[j]})$ themselves not the specific 1-morphisms $F_{[ij]}$. 

\medskip

We note here that the pentagonators are \textit{strict} $\pi=\id$ in the context of higher-gauge theory. A bit more generally, if the pentagnoators only have identity components, then we can define the  \textbf{$\pi$-scalars}
\begin{equation}
    \pi_{[ijkl]} = \pi_{ijkl} \cdot \id_{F_{[lj]}\circ F_{[kl]}\circ F_{[jk]}\circ F_{[ij]}}\,\label{piscalars}.
\end{equation}
The consistency of attaching these pentagonators $\pi$ to the $\theta$-net recoupling moves \eqref{thetanethorizontalcomposite} rests upon the \textit{associahedron equation} satisfied by $\pi$. This allows us to freely associate the $\theta$-nets; we will use this fact in \S \ref{localmoves}.

\subsubsection{4-simplex amplitudes as decorated 2-torii}\label{explicitdef}
We now return to the scattering amplitude and evaluate the toroidal trace $\Tr(\psi_\mathcal{K})$.

\begin{figure}
    \centering
    \includegraphics[width=0.55\linewidth]{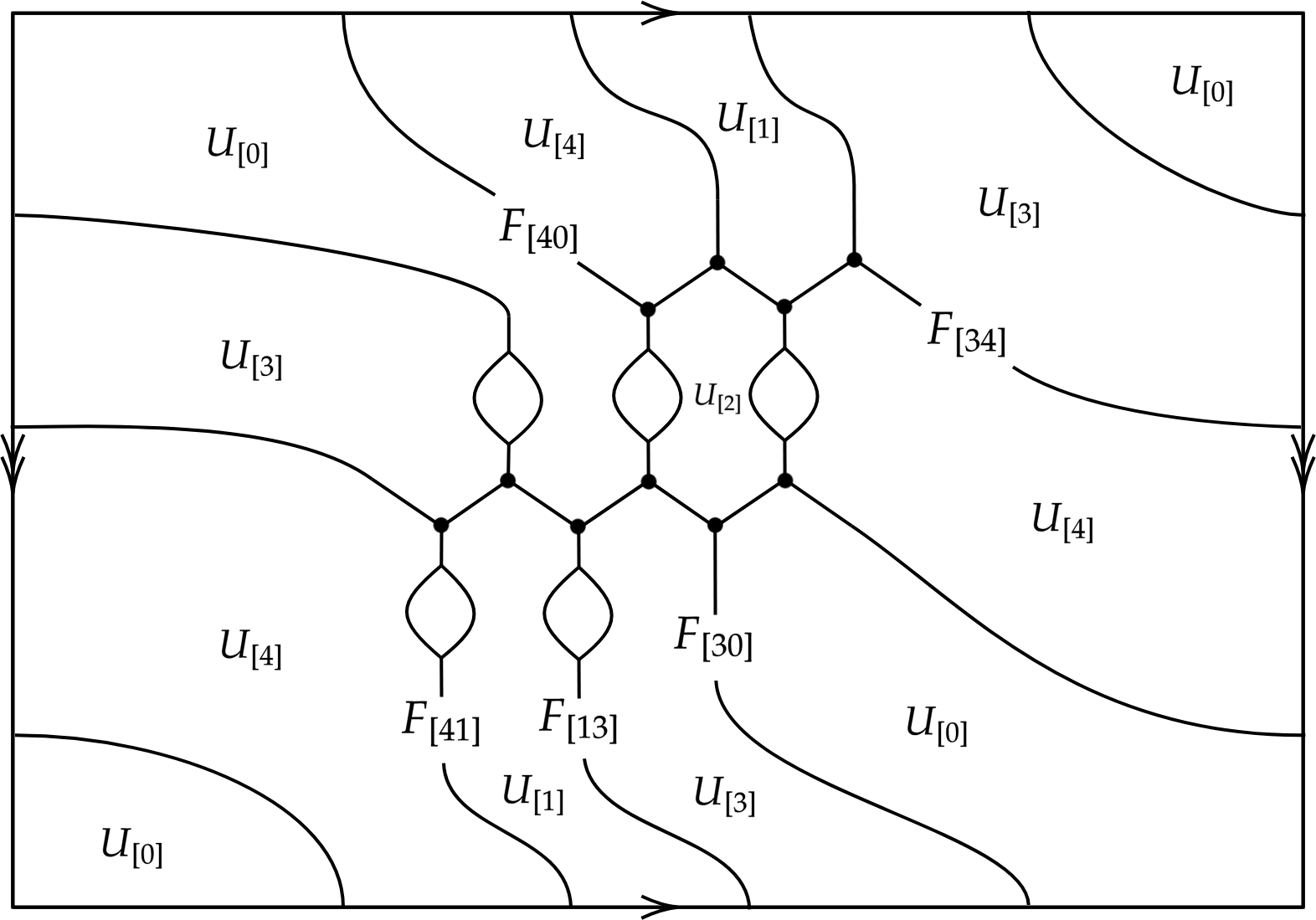}
    \caption{Presentation of the 4-simplex amplitude $\Tr(\psi_\mathcal{K})$ as a decorated  2-torus $\mathbb{T}^2=[0,1]^2/\sim$.}
    \label{fig:amplitudegraph}
\end{figure}

Leveraging the cyclicity of $\Tr$ (or from manipulating the string diagram Fig. \ref{fig:amplitudegraph} on a torus), we express the pasting diagram \eqref{scatteringstate} in the following form (neglecting the {\it \th}-nets for now)
\begin{equation}\left(\begin{tikzcd}
	& {U_2} && {U_4} & \\
	{\pi_{[2403]}} & {U_0} && {U_2} \\
	& {U_3} && {U_0} & {\pi_{[0231]}} \\
	{\pi_{[3014]}} & {U_1} && {U_3} \\
	& {U_4} && {U_1} & {\pi_{[1342]}} \\
	{\pi_{[4120]}} & {U_2} && {U_4}
	\arrow[""{name=0, anchor=center, inner sep=0}, from=1-2, to=1-4]
	\arrow[from=1-2, to=2-2]
	\arrow[from=1-4, to=2-4]
	\arrow["\ast"{description}, draw=none, from=2-1, to=2-2]
	\arrow[""{name=1, anchor=center, inner sep=0}, from=2-2, to=2-4]
	\arrow[from=2-2, to=3-2]
	\arrow[from=2-4, to=3-4]
	\arrow[""{name=2, anchor=center, inner sep=0}, from=3-2, to=3-4]
	\arrow[from=3-2, to=4-2]
	\arrow["\ast"{description}, draw=none, from=3-4, to=3-5]
	\arrow[from=3-4, to=4-4]
	\arrow["\ast"{description}, draw=none, from=4-1, to=4-2]
	\arrow[""{name=3, anchor=center, inner sep=0}, from=4-2, to=4-4]
	\arrow[from=4-2, to=5-2]
	\arrow[from=4-4, to=5-4]
	\arrow[""{name=4, anchor=center, inner sep=0}, from=5-2, to=5-4]
	\arrow[from=5-2, to=6-2]
	\arrow["\ast"{description}, draw=none, from=5-4, to=5-5]
	\arrow[from=5-4, to=6-4]
	\arrow["\ast"{description}, draw=none, from=6-1, to=6-2]
	\arrow[""{name=5, anchor=center, inner sep=0}, from=6-2, to=6-4]
	\arrow["{\eta_{[204]}^T\ast_A \eta_{[042]}}"{description}, between={0.1}{0.9}, Rightarrow, nfold, from=0, to=1]
	\arrow["{\eta_{[032]}^T\ast_A \eta_{[320]}}"{description}, between={0.1}{0.9}, Rightarrow, nfold, from=1, to=2]
	\arrow["{\eta_{[310]}^T\ast_A\eta_{[103]}}"{description}, between={0.1}{0.9}, Rightarrow, nfold, from=2, to=3]
	\arrow["{\eta_{[134]}^T\ast_A \eta_{[431]}}"{description}, between={0.1}{0.9}, Rightarrow, nfold, from=3, to=4]
	\arrow["{\eta_{[421]}^T\ast_A \eta_{[214]}}"{description}, between={0.1}{0.9}, Rightarrow, nfold, from=4, to=5]
\end{tikzcd}\right):G\circ F_{[24]} \Rightarrow F_{[24]}\circ H\,,\label{scatteringstate2} \end{equation}
where
\begin{equation*}
    G= F_{[14]}\circ F_{[31]}\circ F_{[03]}\circ F_{[20]}\circ F_{[42]},\qquad H= F_{[42]}\circ F_{[14]}\circ F_{[31]}\circ F_{[03]}\circ F_{[20]}\,,
\end{equation*}
Here, the pentagonators $\pi$ appear from \eqref{pentagonaxiom} to contract against the consecutive cubical 2-morphisms in \eqref{scatteringstate2}, in accordance to the recoupling theory of \S \ref{pentagonatorscalars}.

This brings $\psi_\mathcal{K}$ to a form on which the shadow trace can be applied (cf. \S \ref{shadotrace}), 
\begin{equation*}
    \operatorname{tr}_\textbf{T}(\psi_\mathcal{K}): \bbra G \kket_{\ttU_{[4]}} \to \bbra H\kket_{\ttU_{[2]}}\,.
\end{equation*}
Note the 1-morphisms $G,H$ here differ by a cyclic permutation of its composites, hence we can use the cyclicty operator $\theta$ to construct the endomorphism\footnote{Alternatively we can precompose with $\theta$ to achieve an endomorphism $\operatorname{tr}_\textbf{T}(\psi_\mathcal{K})\circ\theta\in\operatorname{End}_\textbf{T}(\bbra H \kket_{\ttU_{[2]}})$. But by cyclicity their $\textbf{T}$-traces coincide.}
\begin{equation}
    \operatorname{tr}'_\textbf{T}(\psi_\mathcal{K})=\theta_{F_{[42]},F_{[14]}\circ F_{[31]}\circ F_{[03]}\circ F_{[20]}} \circ \operatorname{tr}_\textbf{T}(\psi_\mathcal{K})\in\operatorname{End}_\textbf{T}(\bbra G \kket_{\ttU_{[4]}})\label{cycledmorphism}
\end{equation}
We shall define "$\Tr(\psi_\mathcal{K})$" as the trace of this endomorphism in $\textbf{T}$;      see also Fig. \ref{fig:amplitudegraph}.

\medskip


\begin{definition}\label{4simpscatterdef}
    The \textbf{4-simplex scattering amplitude} on (positively-oriented) $\sigma^+$ is given by the toroidal trace,
    \begin{equation}
        \langle\sigma^+\rangle:  \cH_{\partial\sigma}^+\to \operatorname{End}_\textbf{T}(e)\cong \bbC,\qquad \psi_{\mathcal{K}_s} \mapsto \langle\sigma^+(\mathcal{K}_s)\rangle =\Tr(\psi_{\mathcal{K}_s})\,. \label{amplitude}
    \end{equation}
\end{definition}
\noindent Now if the pentagonators have only identity components, then we can rewrite \eqref{scatteringstate2}. Put
\begin{align*}
    \prod_{\tau\in\sigma}\pi_\tau& = \pi_{3014}\pi_{2403}\pi_{1342}\pi_{0231}\pi_{4120}\\
    \prod_{v\in \sigma}\Dim(\ttU_v) &= \Dim(\ttU_{[0]})\Dim(\ttU_{[1]})\Dim(\ttU_{[2]})\Dim(\ttU_{[3]})\Dim(\ttU_{[4]})\,,
\end{align*}
where $\mathbb{D}\text{im}(\ttU) = \Tr(\id_{1_\ttU})$ denotes the \textbf{toroidal dimension} of $\ttU\in\cC$.
\begin{definition}\label{20jsymboldef}
    Given the $\pi$-scalars \eqref{piscalars}, we put $$\langle\sigma^+(\mathcal{K}_s)\rangle = \frac{\prod_{v\in \sigma}\Dim(\ttU_v)}{\prod_{\tau\in\sigma}\pi_\tau}\cdot|20j|_\sigma(\beta,\gamma,\delta)\,$$ 
    in terms of the \textbf{normalized 20$j$-symbols} $|20j|_\sigma(\beta,\gamma,\delta)$.
\end{definition}
The planar string diagram of the 20$j$-symbols is displayed in Fig. \ref{fig:amplitudegraph}.



\begin{rmk}\label{3sphericality}
    We shall assume that the 4-simplex amplitude $\langle\sigma^+(\mathcal{K})\rangle$ is independent on the choice of taking the toroidal trace from the left or from the right. Equivalently, this can be thought of as the invariance under isotopy of the decorated 2-torii Fig. \ref{fig:amplitudegraph} embedded in the 3-sphere $S^3$;  see Fig. \ref{fig:3spherical}. This is a certain manifestation of "\textbf{3-sphericality}" of the traced bicategory $(\cC,\bbra-\kket,\textbf{T})$. 
\end{rmk}


\begin{figure}
    \centering
    \includegraphics[width=0.6\linewidth]{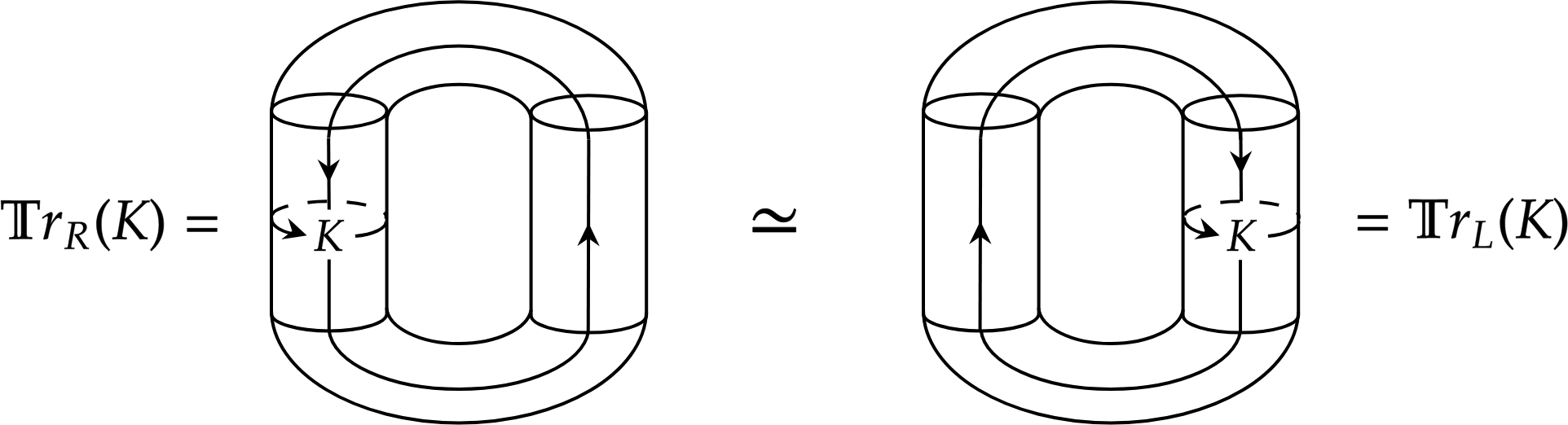}
    \caption{The agreement of the left- and right-toroidal trace of a cubical 2-morphism $K$ exhibits its \textit{3-sphericality}. This can be understood as the invariance of the decorated 2-torus (or any decorated surface) under isotopies of its embedding $\mathbb{T}^2\hookrightarrow S^3$ into the 3-sphere.}
    \label{fig:3spherical}
\end{figure}

\subsection{Evaluation of the Kirby-coloured 20$j$-symbols}\label{eval}
Define the following quantities,
\begin{align*}
    \prod_{\Delta\in \sigma_B}\beta_\Delta &=\beta_{[012]}\beta_{[123]}\beta_{[234]}\beta_{[340]}\beta_{[401]}  \\ 
    \prod_{\Delta\in \sigma_A} \gamma_\Delta & = \gamma_{[240]} \gamma_{[412]}\gamma_{[134]}  \gamma_{[301]}\gamma_{[023]}\\
     \prod_{e\in\sigma_A}\operatorname{dim}(F_e) &= \delta_{20}\delta_{42}\delta_{14}\delta_{31}\delta_{03} 
\end{align*}
where the $\beta,\gamma,\delta$-scalars were given in \eqref{betascalar}, \eqref{gammascalar}, \eqref{deltascalars} respectively.

By a  {\bf Kirby $\cC$-colour}, we mean a $\cC$-decoration for which all 1-simplices are decorated by the direct sum $\bigoplus_F F$ of admissible simple 1-morphisms. This name is inspired by the move Fig. \ref{fig:deltascalar} and \cite{Liu:2023dhj}.
\begin{theorem}\label{20jeval}
    Suppose $\cC$ is $\mathrm{strictly}$ planar-pivotal. Let $\mathcal K=\mathcal K_{\mathrm{Kirb}}$ be a {\bf Kirby $\cC$-colour} on the positively-oriented 4-simplex $\sigma^+$, then we have
    $$ |20j|_{\mathrm{Kirb}}(\beta,\gamma,\delta)= \sum_{\{F_e\}_{e\in\sigma_A}}\frac{\prod_{\Delta\in\sigma_B}\beta_\Delta \prod_{\Delta\in\sigma_A}\gamma_\Delta}{\prod_{e\in\sigma_A}\operatorname{dim}(F_e) }\,.$$ 
\end{theorem}
\begin{proof}
       We can immediately extract the factor $\prod_{\Delta\in\sigma_B}\beta_\Delta$ from the {\it \th}-nets \eqref{betascalar}. The main task is to extract the $\theta$-nets. To begin, consider the state \eqref{scatteringstate2} (without the pentagonators) with all 1-morphisms given by  $\bigoplus_F F$. We then insert resolutions of unity through the completeness of the idempotents $p_F$ \eqref{deltadecouple}. Writing $p_F$ in terms of the co/units of the adjunction \eqref{deltascalars} then yields,
{\footnotesize \[ \sum_{F_{[20]}}\frac{1}{\delta_{20}}\times \begin{tikzcd}
	& {\ttU_{[2]}} & {\ttU_{[4]}} \\
	{\ttU_{[2]}} && {\ttU_{[0]}} & {\ttU_{[2]}} \\
	{\ttU_{[0]}} && {\ttU_{[2]}} \\
	{\ttU_{[3]}} && {\ttU_{[0]}} \\
	{\ttU_{[1]}} && {\ttU_{[3]}} \\
	{\ttU_{[4]}} && {\ttU_{[1]}} \\
	{\ttU_{[2]}} && {\ttU_{[4]}}
	\arrow[""{name=0, anchor=center, inner sep=0}, from=1-2, to=1-3]
	\arrow[""{name=1, anchor=center, inner sep=0}, from=1-2, to=2-3]
	\arrow[""{name=2, anchor=center, inner sep=0}, curve={height=30pt}, dotted, from=1-2, to=2-4]
	\arrow[from=1-3, to=2-4]
	\arrow[from=2-1, to=3-1]
	\arrow[""{name=3, anchor=center, inner sep=0}, dotted, from=2-1, to=3-3]
	\arrow[""{name=4, anchor=center, inner sep=0}, from=2-3, to=2-4]
	\arrow[""{name=5, anchor=center, inner sep=0}, from=3-1, to=3-3]
	\arrow[from=3-1, to=4-1]
	\arrow[from=3-3, to=4-3]
	\arrow[""{name=6, anchor=center, inner sep=0}, from=4-1, to=4-3]
	\arrow[from=4-1, to=5-1]
	\arrow[from=4-3, to=5-3]
	\arrow[""{name=7, anchor=center, inner sep=0}, from=5-1, to=5-3]
	\arrow[from=5-1, to=6-1]
	\arrow[from=5-3, to=6-3]
	\arrow[""{name=8, anchor=center, inner sep=0}, from=6-1, to=6-3]
	\arrow[from=6-1, to=7-1]
	\arrow[from=6-3, to=7-3]
	\arrow[""{name=9, anchor=center, inner sep=0}, from=7-1, to=7-3]
	\arrow[between={0.3}{0.7}, Rightarrow, from=0, to=4]
	\arrow[between={0.3}{0.8}, Rightarrow, from=1, to=2]
	\arrow["\ast"{description}, draw=none, from=2, to=3]
	\arrow[between={0.2}{0.8}, Rightarrow, from=3, to=3-1]
	\arrow["{\eta_{[032]}^T\ast_A \eta_{[320]}}"{description}, between={0.1}{0.9}, Rightarrow, from=5, to=6]
	\arrow["{\eta_{[310]}^T\ast_A\eta_{[103]}}"{description}, between={0.1}{0.9}, Rightarrow, from=6, to=7]
	\arrow["{\eta_{[134]}^T\ast_A \eta_{[431]}}"{description}, between={0.1}{0.9}, Rightarrow, from=7, to=8]
	\arrow["{\eta_{[421]}^T\ast_A \eta_{[214]}}"{description}, between={0.1}{0.9}, Rightarrow, from=8, to=9]
\end{tikzcd}\]}
{\footnotesize\[ =\sum_{F_{[20]},F_{[03]},F_{[31]},F_{[14]}}\frac{1}{\delta_{20}\delta_{03}\delta_{31}\delta_{14}}\times\begin{tikzcd}
	&&&& {\ttU_{[2]}} & {\ttU_{[4]}} \\
	&&& {\ttU_{[2]}} && {\ttU_{[0]}} & {\ttU_{[2]}} \\
	&&& {\ttU_{[0]}} & {\ttU_{[2]}} \\
	&& {\ttU_{[0]}} && {\ttU_{[3]}} & {\ttU_{[0]}} \\
	&& {\ttU_{[3]}} & {\ttU_{[0]}} \\
	& {\ttU_{[3]}} && {\ttU_{[1]}} & {\ttU_{[3]}} \\
	& {\ttU_{[1]}} & {\ttU_{[3]}} \\
	{\ttU_{[1]}} && {\ttU_{[4]}} & {\ttU_{[1]}} \\
	{\ttU_{[4]}} & {\ttU_{[1]}} \\
	& {\ttU_{[2]}} & {\ttU_{[4]}}
	\arrow[""{name=0, anchor=center, inner sep=0}, from=1-5, to=1-6]
	\arrow[""{name=1, anchor=center, inner sep=0}, from=1-5, to=2-6]
	\arrow[""{name=2, anchor=center, inner sep=0}, curve={height=30pt}, dotted, from=1-5, to=2-7]
	\arrow[from=1-6, to=2-7]
	\arrow[from=2-4, to=3-4]
	\arrow[""{name=3, anchor=center, inner sep=0}, dotted, from=2-4, to=3-5]
	\arrow[""{name=4, anchor=center, inner sep=0}, from=2-6, to=2-7]
	\arrow[""{name=5, anchor=center, inner sep=0}, from=3-4, to=3-5]
	\arrow[""{name=6, anchor=center, inner sep=0}, from=3-4, to=4-5]
	\arrow[""{name=7, anchor=center, inner sep=0}, curve={height=30pt}, dotted, from=3-4, to=4-6]
	\arrow[from=3-5, to=4-6]
	\arrow[from=4-3, to=5-3]
	\arrow[""{name=8, anchor=center, inner sep=0}, dotted, from=4-3, to=5-4]
	\arrow[""{name=9, anchor=center, inner sep=0}, from=4-5, to=4-6]
	\arrow[""{name=10, anchor=center, inner sep=0}, from=5-3, to=5-4]
	\arrow[""{name=11, anchor=center, inner sep=0}, from=5-3, to=6-4]
	\arrow[""{name=12, anchor=center, inner sep=0}, curve={height=30pt}, dotted, from=5-3, to=6-5]
	\arrow[from=5-4, to=6-5]
	\arrow[from=6-2, to=7-2]
	\arrow[""{name=13, anchor=center, inner sep=0}, dotted, from=6-2, to=7-3]
	\arrow[""{name=14, anchor=center, inner sep=0}, from=6-4, to=6-5]
	\arrow[""{name=15, anchor=center, inner sep=0}, from=7-2, to=7-3]
	\arrow[""{name=16, anchor=center, inner sep=0}, from=7-2, to=8-3]
	\arrow[""{name=17, anchor=center, inner sep=0}, curve={height=30pt}, dotted, from=7-2, to=8-4]
	\arrow[from=7-3, to=8-4]
	\arrow[from=8-1, to=9-1]
	\arrow[""{name=18, anchor=center, inner sep=0}, dotted, from=8-1, to=9-2]
	\arrow[""{name=19, anchor=center, inner sep=0}, from=8-3, to=8-4]
	\arrow[""{name=20, anchor=center, inner sep=0}, from=9-1, to=9-2]
	\arrow[from=9-1, to=10-2]
	\arrow[from=9-2, to=10-3]
	\arrow[""{name=21, anchor=center, inner sep=0}, from=10-2, to=10-3]
	\arrow[between={0.3}{0.7}, Rightarrow, from=0, to=4]
	\arrow[between={0.3}{0.8}, Rightarrow, from=1, to=2]
	\arrow["\ast"{description}, draw=none, from=2, to=3]
	\arrow[between={0.2}{0.8}, Rightarrow, from=3, to=3-4]
	\arrow[between={0.3}{0.7}, Rightarrow, from=5, to=9]
	\arrow[between={0.2}{0.8}, Rightarrow, from=6, to=7]
	\arrow["\ast"{description}, draw=none, from=8, to=7]
	\arrow[between={0.2}{1}, Rightarrow, from=8, to=5-3]
	\arrow[between={0.3}{0.7}, Rightarrow, from=10, to=14]
	\arrow[between={0.2}{0.8}, Rightarrow, from=11, to=12]
	\arrow["\ast"{description}, draw=none, from=12, to=13]
	\arrow[between={0.2}{1}, Rightarrow, from=13, to=7-2]
	\arrow[between={0.3}{0.7}, Rightarrow, from=15, to=19]
	\arrow[between={0.2}{0.8}, Rightarrow, from=16, to=17]
	\arrow["\ast"{description}, draw=none, from=17, to=18]
	\arrow[between={0.2}{1}, Rightarrow, from=18, to=9-1]
	\arrow[between={0.2}{0.8}, Rightarrow, from=20, to=21]
\end{tikzcd}\]}
Decoupling the final $\id_{F_{[42]}}$, we have 
{\footnotesize \[=\sum_{\{F_e\}_{e\in\sigma_A}}\frac{1}{\prod_{e\in\sigma_A}\delta_e}\times \iota_{F_{[24]}}\ast \begin{tikzcd}
	&&&& {\ttU_{[2]}} & {\ttU_{[4]}} \\
	&&& {\ttU_{[2]}} && {\ttU_{[0]}} & {\ttU_{[2]}} \\
	&&& {\ttU_{[0]}} & {\ttU_{[2]}} \\
	&& {\ttU_{[0]}} && {\ttU_{[3]}} & {\ttU_{[0]}} \\
	&& {\ttU_{[3]}} & {\ttU_{[0]}} \\
	& {\ttU_{[3]}} && {\ttU_{[1]}} & {\ttU_{[3]}} \\
	& {\ttU_{[1]}} & {\ttU_{[3]}} \\
	{\ttU_{[1]}} && {\ttU_{[4]}} & {\ttU_{[1]}} \\
	{\ttU_{[4]}} & {\ttU_{[1]}} \\
	& {\ttU_{[2]}} & {\ttU_{[4]}}
	\arrow[""{name=0, anchor=center, inner sep=0}, from=1-5, to=1-6]
	\arrow[""{name=1, anchor=center, inner sep=0}, from=1-5, to=2-6]
	\arrow[""{name=2, anchor=center, inner sep=0}, curve={height=30pt}, dotted, from=1-5, to=2-7]
	\arrow[from=1-6, to=2-7]
	\arrow[from=2-4, to=3-4]
	\arrow[""{name=3, anchor=center, inner sep=0}, dotted, from=2-4, to=3-5]
	\arrow[""{name=4, anchor=center, inner sep=0}, from=2-6, to=2-7]
	\arrow[""{name=5, anchor=center, inner sep=0}, from=3-4, to=3-5]
	\arrow[""{name=6, anchor=center, inner sep=0}, from=3-4, to=4-5]
	\arrow[""{name=7, anchor=center, inner sep=0}, curve={height=30pt}, dotted, from=3-4, to=4-6]
	\arrow[from=3-5, to=4-6]
	\arrow[from=4-3, to=5-3]
	\arrow[""{name=8, anchor=center, inner sep=0}, dotted, from=4-3, to=5-4]
	\arrow[""{name=9, anchor=center, inner sep=0}, from=4-5, to=4-6]
	\arrow[""{name=10, anchor=center, inner sep=0}, from=5-3, to=5-4]
	\arrow[""{name=11, anchor=center, inner sep=0}, from=5-3, to=6-4]
	\arrow[""{name=12, anchor=center, inner sep=0}, curve={height=30pt}, dotted, from=5-3, to=6-5]
	\arrow[from=5-4, to=6-5]
	\arrow[from=6-2, to=7-2]
	\arrow[""{name=13, anchor=center, inner sep=0}, dotted, from=6-2, to=7-3]
	\arrow[""{name=14, anchor=center, inner sep=0}, from=6-4, to=6-5]
	\arrow[""{name=15, anchor=center, inner sep=0}, from=7-2, to=7-3]
	\arrow[""{name=16, anchor=center, inner sep=0}, from=7-2, to=8-3]
	\arrow[""{name=17, anchor=center, inner sep=0}, curve={height=30pt}, dotted, from=7-2, to=8-4]
	\arrow[from=7-3, to=8-4]
	\arrow[from=8-1, to=9-1]
	\arrow[""{name=18, anchor=center, inner sep=0}, dotted, from=8-1, to=9-2]
	\arrow[""{name=19, anchor=center, inner sep=0}, from=8-3, to=8-4]
	\arrow[""{name=20, anchor=center, inner sep=0}, from=9-1, to=9-2]
	\arrow[""{name=21, anchor=center, inner sep=0}, from=9-1, to=10-2]
	\arrow[""{name=22, anchor=center, inner sep=0}, curve={height=30pt}, dotted, from=9-1, to=10-3]
	\arrow[from=9-2, to=10-3]
	\arrow[""{name=23, anchor=center, inner sep=0}, from=10-2, to=10-3]
	\arrow[between={0.3}{0.7}, Rightarrow, from=0, to=4]
	\arrow[between={0.3}{0.8}, Rightarrow, from=1, to=2]
	\arrow["\ast"{description}, draw=none, from=2, to=3]
	\arrow[between={0.2}{0.8}, Rightarrow, from=3, to=3-4]
	\arrow[between={0.3}{0.7}, Rightarrow, from=5, to=9]
	\arrow[between={0.2}{0.8}, Rightarrow, from=6, to=7]
	\arrow["\ast"{description}, draw=none, from=8, to=7]
	\arrow[between={0.2}{1}, Rightarrow, from=8, to=5-3]
	\arrow[between={0.3}{0.7}, Rightarrow, from=10, to=14]
	\arrow[between={0.2}{0.8}, Rightarrow, from=11, to=12]
	\arrow["\ast"{description}, draw=none, from=12, to=13]
	\arrow[between={0.2}{1}, Rightarrow, from=13, to=7-2]
	\arrow[between={0.3}{0.7}, Rightarrow, from=15, to=19]
	\arrow[between={0.2}{0.8}, Rightarrow, from=16, to=17]
	\arrow["\ast"{description}, draw=none, from=17, to=18]
	\arrow[between={0.2}{1}, Rightarrow, from=18, to=9-1]
	\arrow[between={0.2}{0.8}, Rightarrow, from=20, to=23]
	\arrow[between={0.2}{0.8}, Rightarrow, from=21, to=22]
\end{tikzcd}\]}
We now apply the toroidal trace and use its cyclicity to bring the trailing $\iota_{F_{[24]}}$ to the top right,
{\footnotesize \[\sum_{\{F_e\}_{e\in\sigma_A}}\frac{1}{\prod_{e\in\sigma_A}\delta_e} \times\Tr\left(\begin{tikzcd}
	&&&& {\ttU_{[4]}} \\
	&&&& {\ttU_{[2]}} & {\ttU_{[4]}} \\
	&&& {\ttU_{[2]}} && {\ttU_{[0]}} & {\ttU_{[2]}} \\
	&&& {\ttU_{[0]}} & {\ttU_{[2]}} \\
	&& {\ttU_{[0]}} && {\ttU_{[3]}} & {\ttU_{[0]}} \\
	&& {\ttU_{[3]}} & {\ttU_{[0]}} \\
	& {\ttU_{[3]}} && {\ttU_{[1]}} & {\ttU_{[3]}} \\
	& {\ttU_{[1]}} & {\ttU_{[3]}} \\
	{\ttU_{[1]}} && {\ttU_{[4]}} & {\ttU_{[1]}} \\
	{\ttU_{[4]}} & {\ttU_{[1]}} \\
	& {\ttU_{[2]}} & {\ttU_{[4]}}
	\arrow[from=1-5, to=2-5]
	\arrow[""{name=0, anchor=center, inner sep=0}, dotted, from=1-5, to=2-6]
	\arrow[""{name=1, anchor=center, inner sep=0}, from=2-5, to=2-6]
	\arrow[""{name=2, anchor=center, inner sep=0}, from=2-5, to=3-6]
	\arrow[""{name=3, anchor=center, inner sep=0}, curve={height=30pt}, dotted, from=2-5, to=3-7]
	\arrow[from=2-6, to=3-7]
	\arrow[from=3-4, to=4-4]
	\arrow[""{name=4, anchor=center, inner sep=0}, dotted, from=3-4, to=4-5]
	\arrow[""{name=5, anchor=center, inner sep=0}, from=3-6, to=3-7]
	\arrow[""{name=6, anchor=center, inner sep=0}, from=4-4, to=4-5]
	\arrow[""{name=7, anchor=center, inner sep=0}, from=4-4, to=5-5]
	\arrow[""{name=8, anchor=center, inner sep=0}, curve={height=30pt}, dotted, from=4-4, to=5-6]
	\arrow[from=4-5, to=5-6]
	\arrow[from=5-3, to=6-3]
	\arrow[""{name=9, anchor=center, inner sep=0}, dotted, from=5-3, to=6-4]
	\arrow[""{name=10, anchor=center, inner sep=0}, from=5-5, to=5-6]
	\arrow[""{name=11, anchor=center, inner sep=0}, from=6-3, to=6-4]
	\arrow[""{name=12, anchor=center, inner sep=0}, from=6-3, to=7-4]
	\arrow[""{name=13, anchor=center, inner sep=0}, curve={height=30pt}, dotted, from=6-3, to=7-5]
	\arrow[from=6-4, to=7-5]
	\arrow[from=7-2, to=8-2]
	\arrow[""{name=14, anchor=center, inner sep=0}, dotted, from=7-2, to=8-3]
	\arrow[""{name=15, anchor=center, inner sep=0}, from=7-4, to=7-5]
	\arrow[""{name=16, anchor=center, inner sep=0}, from=8-2, to=8-3]
	\arrow[""{name=17, anchor=center, inner sep=0}, from=8-2, to=9-3]
	\arrow[""{name=18, anchor=center, inner sep=0}, curve={height=30pt}, dotted, from=8-2, to=9-4]
	\arrow[from=8-3, to=9-4]
	\arrow[from=9-1, to=10-1]
	\arrow[""{name=19, anchor=center, inner sep=0}, dotted, from=9-1, to=10-2]
	\arrow[""{name=20, anchor=center, inner sep=0}, from=9-3, to=9-4]
	\arrow[""{name=21, anchor=center, inner sep=0}, from=10-1, to=10-2]
	\arrow[""{name=22, anchor=center, inner sep=0}, from=10-1, to=11-2]
	\arrow[""{name=23, anchor=center, inner sep=0}, curve={height=30pt}, dotted, from=10-1, to=11-3]
	\arrow[from=10-2, to=11-3]
	\arrow[""{name=24, anchor=center, inner sep=0}, from=11-2, to=11-3]
	\arrow[between={0.2}{1}, Rightarrow, from=0, to=2-5]
	\arrow[between={0.3}{0.7}, Rightarrow, from=1, to=5]
	\arrow[between={0.3}{0.8}, Rightarrow, from=2, to=3]
	\arrow["\ast"{description}, draw=none, from=3, to=4]
	\arrow[between={0.2}{0.8}, Rightarrow, from=4, to=4-4]
	\arrow[between={0.3}{0.7}, Rightarrow, from=6, to=10]
	\arrow[between={0.2}{0.8}, Rightarrow, from=7, to=8]
	\arrow["\ast"{description}, draw=none, from=9, to=8]
	\arrow[between={0.2}{1}, Rightarrow, from=9, to=6-3]
	\arrow[between={0.3}{0.7}, Rightarrow, from=11, to=15]
	\arrow[between={0.2}{0.8}, Rightarrow, from=12, to=13]
	\arrow["\ast"{description}, draw=none, from=13, to=14]
	\arrow[between={0.2}{1}, Rightarrow, from=14, to=8-2]
	\arrow[between={0.3}{0.7}, Rightarrow, from=16, to=20]
	\arrow[between={0.2}{0.8}, Rightarrow, from=17, to=18]
	\arrow["\ast"{description}, draw=none, from=18, to=19]
	\arrow[between={0.2}{1}, Rightarrow, from=19, to=10-1]
	\arrow[between={0.2}{0.8}, Rightarrow, from=21, to=24]
	\arrow[between={0.2}{0.8}, Rightarrow, from=22, to=23]
\end{tikzcd}\right)\]}
By dagger adjunction $e_F^T  = \iota_{F^\dagger} $ \eqref{unitaryadjoint}, the planar snake equation allows us to replace $F^\dagger\circ \iota_{F^\dagger}$ with the dagger conjugate of $e_{F}^T\circ F$. Inserting this into the above diagram then  makes the $\gamma$-scalars manifest,
{\footnotesize\[\hspace{-0.3cm}\sum_{\{F_e\}_{e\in\sigma_A}}\frac{1}{\prod_{e\in\sigma_A}\delta_e}\times \Tr \left(\begin{tikzcd}
	&&&&&&&& {\ttU_2} & {\ttU_4} \\
	&&&&&& {U_0} & {\ttU_2} && {U_0} & {\ttU_2} \\
	&&&& {\ttU_3} & {U_0} && {\ttU_3} & {U_0} \\
	&& {\ttU_1} & {\ttU_3} && {\ttU_1} & {\ttU_3} \\
	{\ttU_4} & {\ttU_1} && {\ttU_4} & {\ttU_1} \\
	& {\ttU_2} & {\ttU_4}
	\arrow[""{name=0, anchor=center, inner sep=0}, from=1-9, to=1-10]
	\arrow[draw={rgb,255:red,163;green,163;blue,163}, curve={height=6pt}, from=1-9, to=2-7]
	\arrow[""{name=1, anchor=center, inner sep=0}, from=1-9, to=2-10]
	\arrow[""{name=2, anchor=center, inner sep=0}, curve={height=30pt}, dotted, from=1-9, to=2-11]
	\arrow[""{name=3, anchor=center, inner sep=0}, curve={height=-30pt}, dotted, from=1-9, to=2-11]
	\arrow[""{name=4, anchor=center, inner sep=0}, from=1-10, to=2-11]
	\arrow[""{name=5, anchor=center, inner sep=0}, from=2-7, to=2-8]
	\arrow[draw={rgb,255:red,163;green,163;blue,163}, curve={height=6pt}, from=2-7, to=3-5]
	\arrow[""{name=6, anchor=center, inner sep=0}, from=2-7, to=3-8]
	\arrow[""{name=7, anchor=center, inner sep=0}, curve={height=30pt}, dotted, from=2-7, to=3-9]
	\arrow[""{name=8, anchor=center, inner sep=0}, curve={height=-30pt}, dotted, from=2-7, to=3-9]
	\arrow[""{name=9, anchor=center, inner sep=0}, from=2-8, to=3-9]
	\arrow[""{name=10, anchor=center, inner sep=0}, from=2-10, to=2-11]
	\arrow[""{name=11, anchor=center, inner sep=0}, from=3-5, to=3-6]
	\arrow[draw={rgb,255:red,163;green,163;blue,163}, curve={height=6pt}, from=3-5, to=4-3]
	\arrow[""{name=12, anchor=center, inner sep=0}, from=3-5, to=4-6]
	\arrow[""{name=13, anchor=center, inner sep=0}, curve={height=30pt}, dotted, from=3-5, to=4-7]
	\arrow[""{name=14, anchor=center, inner sep=0}, curve={height=-30pt}, dotted, from=3-5, to=4-7]
	\arrow[""{name=15, anchor=center, inner sep=0}, from=3-6, to=4-7]
	\arrow[""{name=16, anchor=center, inner sep=0}, from=3-8, to=3-9]
	\arrow[""{name=17, anchor=center, inner sep=0}, from=4-3, to=4-4]
	\arrow[draw={rgb,255:red,163;green,163;blue,163}, curve={height=6pt}, from=4-3, to=5-1]
	\arrow[""{name=18, anchor=center, inner sep=0}, from=4-3, to=5-4]
	\arrow[""{name=19, anchor=center, inner sep=0}, curve={height=30pt}, dotted, from=4-3, to=5-5]
	\arrow[""{name=20, anchor=center, inner sep=0}, curve={height=-30pt}, dotted, from=4-3, to=5-5]
	\arrow[""{name=21, anchor=center, inner sep=0}, from=4-4, to=5-5]
	\arrow[""{name=22, anchor=center, inner sep=0}, from=4-6, to=4-7]
	\arrow[""{name=23, anchor=center, inner sep=0}, from=5-1, to=5-2]
	\arrow[""{name=24, anchor=center, inner sep=0}, from=5-1, to=6-2]
	\arrow[""{name=25, anchor=center, inner sep=0}, curve={height=30pt}, dotted, from=5-1, to=6-3]
	\arrow[""{name=26, anchor=center, inner sep=0}, curve={height=-30pt}, dotted, from=5-1, to=6-3]
	\arrow[""{name=27, anchor=center, inner sep=0}, from=5-2, to=6-3]
	\arrow[""{name=28, anchor=center, inner sep=0}, from=5-4, to=5-5]
	\arrow[""{name=29, anchor=center, inner sep=0}, from=6-2, to=6-3]
	\arrow[between={0.3}{0.7}, Rightarrow, from=0, to=10]
	\arrow[between={0.3}{0.8}, Rightarrow, from=1, to=2]
	\arrow[between={0.2}{0.8}, Rightarrow, from=3, to=4]
	\arrow["\ast"{description}, draw=none, from=8, to=2]
	\arrow[between={0.3}{0.7}, Rightarrow, from=5, to=16]
	\arrow[between={0.2}{0.8}, Rightarrow, from=6, to=7]
	\arrow[between={0.2}{0.8}, Rightarrow, from=8, to=9]
	\arrow["\ast"{description}, draw=none, from=14, to=7]
	\arrow[between={0.3}{0.7}, Rightarrow, from=11, to=22]
	\arrow[between={0.2}{0.8}, Rightarrow, from=12, to=13]
	\arrow[between={0.2}{0.8}, Rightarrow, from=14, to=15]
	\arrow["\ast"{description}, draw=none, from=20, to=13]
	\arrow[between={0.3}{0.7}, Rightarrow, from=17, to=28]
	\arrow[between={0.2}{0.8}, Rightarrow, from=18, to=19]
	\arrow[between={0.2}{0.8}, Rightarrow, from=20, to=21]
	\arrow["\ast"{description}, draw=none, from=26, to=19]
	\arrow[between={0.2}{0.8}, Rightarrow, from=23, to=29]
	\arrow[between={0.2}{0.8}, Rightarrow, from=24, to=25]
	\arrow[between={0.2}{0.8}, Rightarrow, from=26, to=27]
\end{tikzcd}\right)\]}
Evaluating with \eqref{gammascalar} and using the tightening of the shadow trace (1. and 3. of \textbf{Theorem \ref{shadowproperties}}) then gives
\begin{align*}
&=\sum_{\{F_e\}_{e\in\sigma_A}}\prod_{e\in \sigma_A} \frac{1}{\delta_{e}}\prod_{\Delta\in\sigma_A}\gamma_\Delta \times \Tr\left(
\begin{tikzcd}[ampersand replacement=\&]
	{} \& {} \& {} \& {} \& {} \\
	{\ttU_2} \& {U_0} \& {\ttU_3} \& {\ttU_1} \& {\ttU_4}
	\arrow["{\id_{1_{\ttU_{[2]}}}}"{description, pos=0.3}, draw=none, from=1-1, to=2-1]
	\arrow["{\id_{1_{\ttU_{[0]}}}}"{description, pos=0.3}, draw=none, from=1-2, to=2-2]
	\arrow["{\id_{1_{\ttU_{[3]}}}}"{description, pos=0.3}, draw=none, from=1-3, to=2-3]
	\arrow["{\id_{1_{\ttU_{[1]}}}}"{description, pos=0.3}, draw=none, from=1-4, to=2-4]
	\arrow["{\id_{1_{\ttU_{[4]}}}}"{description, pos=0.2}, draw=none, from=1-5, to=2-5]
	\arrow[from=2-1, to=2-1, loop, in=50, out=130, distance=15mm]
	\arrow["{F_{[20]}}"', draw={rgb,255:red,163;green,163;blue,163}, curve={height=6pt}, from=2-1, to=2-2]
	\arrow[from=2-2, to=2-2, loop, in=50, out=130, distance=15mm]
	\arrow["{F_{[03]}}"', draw={rgb,255:red,163;green,163;blue,163}, curve={height=6pt}, from=2-2, to=2-3]
	\arrow[from=2-3, to=2-3, loop, in=50, out=130, distance=15mm]
	\arrow["{F_{[31]}}"', draw={rgb,255:red,163;green,163;blue,163}, curve={height=6pt}, from=2-3, to=2-4]
	\arrow[from=2-4, to=2-4, loop, in=50, out=130, distance=15mm]
	\arrow["{F_{[14]}}"', draw={rgb,255:red,163;green,163;blue,163}, curve={height=6pt}, from=2-4, to=2-5]
	\arrow[from=2-5, to=2-5, loop, in=50, out=130, distance=15mm]
\end{tikzcd}\right)\\
  &= \sum_{\{F_e\}_{e\in\sigma_A}}\prod_{e\in \sigma_A} \frac{1}{\delta_{e}}\prod_{\Delta\in\sigma_A}\gamma_\Delta \times\prod_{v\in \sigma}\Dim(\ttU_v)\,.
\end{align*}
Bringing out the toroidal dimensions, we finally have
\begin{align*}
     |20j|_\mathrm{Kirb}(\beta,\gamma,\delta)&= \sum_{\{F_e\}_{e\in\sigma_A}} \frac{\prod_{\Delta\in\sigma_B}\beta_\Delta \prod_{\Delta\in\sigma_A}\gamma_\Delta }{\prod_{e\in \sigma_A} \delta_{e}} \,,
\end{align*}
as desired.
\end{proof}
\noindent When the pentagonators of $\cC$ has only identity components, then this result allows us to explicitly compute the amplitude on the 4-sphere (see \S \ref{pachner}). It also explains the name "20$j$-symbols" --- its dependence on the vertices decorations drops, and it only depends on the 10 faces and 10 edges of the 4-simplex.

\subsection{Reflection-positivity}\label{reflectionpositive}
Let us examine the dependence of the scattering states on the orientation. 
\begin{proposition}\label{4dorientation}
    The scattering amplitude associated to a $\cC$-decoration $\bar{\mathcal{K}}$ on a negatively-oriented 4-simplex $\sigma^-$ is given by 
    \begin{equation*}
        \langle\sigma^-(\bar{\mathcal{K}})\rangle = \Tr(\psi_{\mathcal{K}}^{\dagger T}) 
    \end{equation*}
    for some $\psi_{\mathcal{K}}\in \cH_\tau^+$.
\end{proposition}
\begin{proof}
    Since a 4-simplex with the opposite orientation, $\sigma^-$, naturally has oppositely-oriented boundary, this follows from \textbf{Proposition \ref{3dreversal}}. There exists an orientation-preserving rotation $x\in A_4$ for which we have a linear isomorphism
    \begin{equation*}
        \big(\cH_{\partial\sigma}^+\big)^{\dagger T} \xrightarrow{\sim}  \cH_{\partial\sigma}^-\,,\qquad \psi_{\mathcal{K}}\mapsto R_x\psi_\mathcal{K}R_x^{T}
    \end{equation*}
    represented by some orthogonal matrix $R_x$. By taking a torodidal trace, its cyclicity and the orthogonality of $x\in A_4\subset SO(3)$ then gives 
    \begin{equation*}
        \Tr(\psi_{\bar{\mathcal{K}}}) = \Tr(R_x^{T}\psi_{\mathcal{K}}R_x)= \Tr(\psi_\mathcal{K}^{\dagger T})\,.
    \end{equation*}
\end{proof}
\noindent Due to this result, we will denote simply the negatively-oriented 4-simplex amplitude by $\langle\sigma^-(\bar{\mathcal{K}})\rangle = \langle\sigma(\bar{\mathcal{K}})\rangle$.

\medskip

We now wish to use \textbf{Proposition \ref{4dorientation}} to characterize the reflection-positivity of the scattering amplitudes. For this, we will assume that there exists a $\textbf{T}$-shadow $\bbra-\kket^\text{op}$  on $\cC^\text{1-op}$ such that we have a \textit{shadow 2-functor} \cite{ponto2013shadows}
\begin{equation}
  (-^\dagger ,-^\vee,d):(\cC,\textbf{T},\bbra-\kket) \to (\cC^\text{1-op},\textbf{T}^\text{m-op},\bbra-\kket^\text{op});\label{opshadow}  
\end{equation}
that is, for each $F\in\operatorname{End}_\cC(\ttU)$ and $\ttU\in\cC$ there exist a natural isomorphism $d_F:\bbra F^\dagger\kket_\ttU^\text{op} \cong (\bbra F\kket_\ttU)^\vee$ in $\textbf{T}$ which commutes with the cyclicity operator
\[\begin{tikzcd}
	{\bbra G^\dagger \circ F^\dagger\kket^\text{op}_\ttV} && { \bbra F^\dagger\circ G^\dagger\kket^\text{op}_\ttU} \\
	{(\bbra F\circ G\kket_\ttV)^\vee} && {(\bbra G\circ F\kket_\ttU)^\vee}
	\arrow["{\theta_{G^\dagger,F^\dagger}}", from=1-1, to=1-3]
	\arrow["{d_{F\circ G}}"', from=1-1, to=2-1]
	\arrow["{d_{G\circ F}}", from=1-3, to=2-3]
	\arrow["{\theta_{G,F}^\vee}"', from=2-1, to=2-3]
\end{tikzcd}\]
 for each $F: \ttU\to \ttV$ and $G:\ttV\to \ttU$. See \textit{Example \ref{serrefunctor}}. We do \textit{not}, however, require any other data (eg. monoidal product or the co/evaluations) to be preserved.

Consider two cubical 2-morphisms in $\cC$ of the form $\eta: G\circ F\Rightarrow F\circ H$ and $\tilde \eta: F^\dagger\circ G^\dagger\Rightarrow H^\dagger \circ F^\dagger$. Beginning from their shadow traces
\begin{align*}
    \operatorname{tr}_\textbf{T}\left(\eta\right)&: \bbra G\kket_\ttV \to \bbra H\kket_\ttU \\
    \operatorname{tr}_\textbf{T}\left(\tilde \eta\right) &:\bbra G^\dagger\kket_\ttV^\text{op}\to \bbra H^\dagger\kket_\ttU^\text{op}\,,
\end{align*}
we can form the following scalar-valued functional (see also Fig. \ref{fig:positivetorus})
\begin{align}
    Z(\eta,\tilde\eta) &= \mathfrak{e}_{\bbra H\kket^\vee} \circ \big(\operatorname{tr}_\textbf{T}\left(\eta\right)\otimes d_H\operatorname{tr}_\textbf{T}\left(\tilde\eta\right)d_G^{-1}\big ) \circ \mathfrak{i}_{\bbra G\kket} \,.\label{functional}
\end{align}

\begin{figure}
    \centering
    \includegraphics[width=0.35\linewidth]{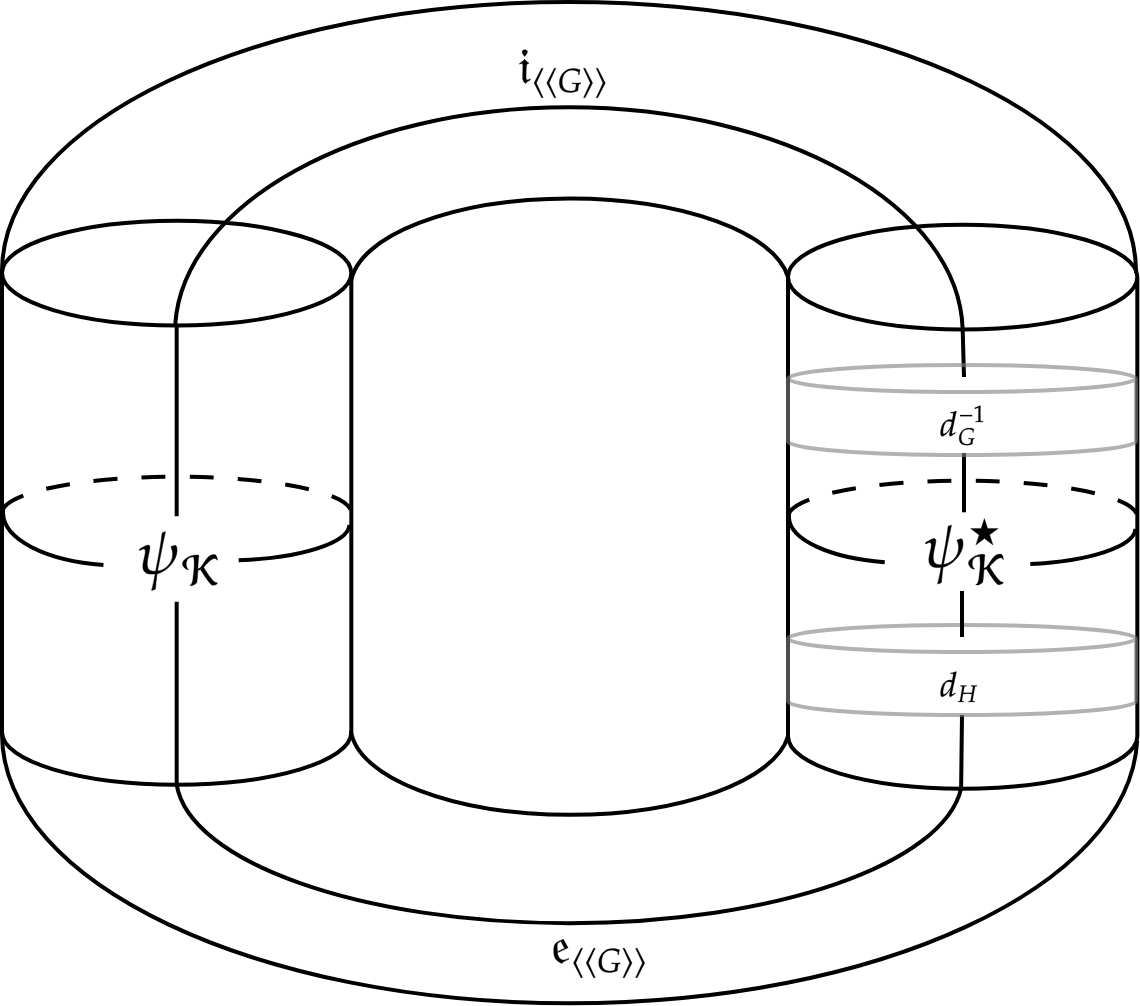}
    \caption{The presentation of the functional $Z$ \eqref{functional} as a surface graph on a 2-torus --- presumed to be invariant under isotopies in $S^3$. For the 4-simplex state $\psi_\mathcal{K}$, the 1-morphisms $G,H$ are given in \eqref{scatteringstate2}. This quantity is crucial for reflection-positivity.}
    \label{fig:positivetorus}
\end{figure}

Now consider two $\cC$-decorated 4-simplices $\mathcal{K}_{1,2}:\sigma_{1,2}\to \cC$ sharing a boundary tetrahedron. Given the $\cC$-decoration on this boundary tetrahedron is the same, the corresponding scattering states $\psi_{\mathcal{K}_1},\psi_{\mathcal{K}_2}$ share legs that be contracted and put onto a decorated torus; see Fig. \ref{fig:positivetorus}. The inner product
    $$\langle \sigma_1(\mathcal{K}_1),\sigma_2(\mathcal{K}_2)\rangle=Z(\psi_{\mathcal{K}_1},\psi_{\mathcal{K}_2}^\star) = \Tr \left(\Psi_{\psi_{\mathcal{K}_1},\psi_{\mathcal{K}_2}} \right)\,,$$
defined by the functional $Z$ \eqref{functional}, is therefore symmetric and non-degenerate modulo $\operatorname{ker}\Tr$; cf. the universal construction of \cite{BLANCHET1995883}.


\begin{definition}
    Denote by $-^{\star} = -^{\dagger T}$. We say that the functional $Z$  \eqref{functional} is \textbf{Hermitian} iff for each pair $\eta,\tilde\eta$ of cubical 2-morphisms of the form above, we have
    \begin{equation*}
        \overline{Z(\eta,\tilde\eta)} = Z(\tilde\eta^\star,\eta^\star),
    \end{equation*}
    where $\bar z$ denotes complex conjugation. The graphical presentation of the pairing $Z$ is given in Fig. \ref{fig:positivetorus}.
\end{definition}

\begin{example}\label{serrefunctor}
Consider once again the example of $\cC=\operatorname{2Rep}(H)$, which we recall consist of $\mathsf{Hilb}$-linear unitary \textit{exact} $H$-modules. Exactness provides a (unique) \textit{relative Serre functor} $\mathbb{S}=\mathbb{S}_\ttU: \ttU\xrightarrow{\sim} \ttU$ \cite{fuchs2025spherical} and a natural isomorphism
    \begin{equation*}
        (s_U)_{x,y}:\operatorname{Hom}_\ttU(x,y)^\vee\cong \operatorname{Hom}_\ttU(y,\mathbb{S}(x))
    \end{equation*}
    which dualizes the hom-spaces of $\ttU$. Together with the $\mathsf{Hilb}$-shadow $\bbra-\kket_\ttU: F\mapsto  \operatorname{Nat}(1_\ttU,F)$ as given in \textit{Example \ref{2charexample}}, we will also consider on $\cC^\text{1-op}$ the $\mathsf{Hilb}$-shadow given by the \textit{coend}
    \begin{equation*}
        \bbra \tilde F\kket_\ttU^\text{op} = \int^{x\in \ttU}\operatorname{Hom}_\ttU(x,\tilde F(x)),\qquad\tilde F\in\operatorname{End}_{\cC^\text{1-op}}(\ttU)\,.
    \end{equation*}
    By definition of the adjunction $F\vdash F^\dagger$ and using the Serre functor, we have
    \begin{align*}
        \operatorname{Nat}(1_\ttU,F)^\vee \cong \int^{x\in \ttU}\operatorname{Hom}_\ttU(x,F(x))^\vee&\cong \int^{x\in \ttU}\operatorname{Hom}_\ttU(F^\dagger(x),x)^\vee\\
        &\cong\int^{x\in \ttU}\operatorname{Hom}_\ttU(x,\mathbb{S}F^\dagger(x))=\bbra \mathbb{S} F^\dagger\kket^\text{op}.
    \end{align*}
    Now given the trivialization\footnote{This $\psi$ has been studied in relation to oriented fully-extended 2d TQFTs \cite{fuchs2025spherical,gainutdinov2026fullyexactfullydualizable,Schommer-Pries:2011rhj,douglas2020dualizable}.} $\psi_U:\mathbb{S}\simeq 1_\ttU$  provided by the $H^*$-structure underlying $\mathtt{U}\in\mathsf{2Hilb}$,  we can construct from $s_U$ and $\psi_U$ the natural isomorphism
    \begin{equation*}
        d_U:\bbra F^\dagger\kket_\ttU^\text{op} \cong \bbra \mathbb{S}F^\dagger\kket_\ttU^\text{op}\cong  \operatorname{Nat}(1_\ttU,F)^\vee = (\bbra F\kket_\ttU)^\vee
    \end{equation*}
    needed to form the functional $Z$.
\end{example}

Now consider $\tilde\eta=\eta^\star$. If $Z$ were Hermitian, then $Z(\eta,\eta^\star)\in\R$ is in fact real-valued. We can therefore define the following (see also \S 4.3 of \cite{Liu:2024qth}).
\begin{definition}\label{refpos}
    We say that a Hermitian $Z$ is \textbf{reflection-positive} iff, for every $\cC$-decorated 4-simplex $\mathcal{K}: \partial\sigma^+ \to \cC$, we have
    \begin{equation*}
        Z(\psi_\mathcal{K},\psi_\mathcal{K}^\star)\geq 0
    \end{equation*}
    for each scattering vector $\psi_\mathcal{K}\in\cH_{\sigma}^+$, with equality iff $\mathcal{K}=0$. In this case, we put
    \begin{equation}
        |\langle\sigma(\mathcal{K})\rangle|^2=Z(\psi_\mathcal{K},\psi_\mathcal{K}^\star).\label{normsquared}
    \end{equation}
\end{definition}
\noindent Together with \textbf{Proposition \ref{4dorientation}}, we then have:
\begin{corollary}\label{refpos2}
    Suppose $Z$ \eqref{functional} is reflection-positive, then $\langle\sigma^-(\bar{\mathcal{K}})\rangle = \overline{\langle \sigma^+(\mathcal{K})\rangle}$. 
\end{corollary}

\section{Topological invariance of the state-sum}\label{pachner}
Let $\mathscr{C}$ be as in \eqref{input}. Through the heavy lifting of \S \ref{decorated2graphs} and \S \ref{20jsymbols}, we define an assignment of
\begin{enumerate}
    \item the total tetrahedral Hilbert space \eqref{totaltetspace}
\begin{equation*}
        Z_\mathscr{C}(M^3,T_{M^3}) = \cH_{T_{M^3}}= \bigoplus_{K:T_{M^3}\to \cC}\bigotimes_{\tau^\pm\subset M^3}\cH_{K(\tau)}^\pm \,
    \end{equation*} 
    to each triangulated oriented closed 3-manifold $(M^3,T_{M^3})$, and
    \item a state-sum over the 4-simplex amplitudes \eqref{amplitude} 
\begin{equation}
        Z_\mathscr{C}(M^4,T_{M^4}) = \frac{1}{{N}}\sum_{\mathcal{K}:T_{M^4}\to \cC}\prod_{\sigma^\pm\subset T_{M^4}}\langle\sigma^\pm(\mathcal{K})\rangle\,\label{partition}
    \end{equation}
    to each triangulated closed oriented 4-manifold $(M^4,T_{M^4})$, where we have set $\sigma^-(\mathcal{K})=\sigma(\bar{\mathcal{K}})$ thanks to \textbf{Proposition \ref{4dorientation}}. 
\end{enumerate} 
The normalization constant $N$ can be chosen to be the value of \eqref{4dorientation} on the 4-sphere with $T_{S^4} = \sigma\cup\bar\sigma$, which by \textbf{Definition \ref{4simpscatterdef}} and reflection-positivity \eqref{normsquared} gives $${N}= Z_\mathscr{C}(S^4,T_{S^4})=\sum_\mathcal{K}|\langle\sigma(\mathcal{K})\rangle|^2\,;$$ see also \textit{Remark \ref{amplitudenorm}}.

\medskip

The main goal of the section is the following.
\begin{theorem}\label{invariant}
    The 4d partition function $Z_\mathscr{C}(M^4,T_{M^4})=Z_\mathscr{C}(M^4)$ \eqref{partition}  does not depend on the triangulation of $M^4$, and is therefore a topological invariant.
\end{theorem}
\begin{proof}
    It suffices to prove the invariance of 20$j$-symbol summations \eqref{partition} under the 4-dimensional Pachner moves. This is done in \S \ref{pachnerinvariance} below.
\end{proof}

\subsection{$\delta$-recoupling and invertible 4-simplex summation}\label{localmoves}
Before we proceed to prove \textbf{Theorem \ref{invariant}}, we first prove a few identities that allows us to perform local moves on the tensor network diagrams. These concern, for a (simple) $\cC$-decoration $K: \tau\coprod\tau'\to \cC$ of two tetrahedra $\tau,\tau'$ sharing three vertex labels. 

\begin{figure}
    \centering
    \includegraphics[width=0.95\linewidth]{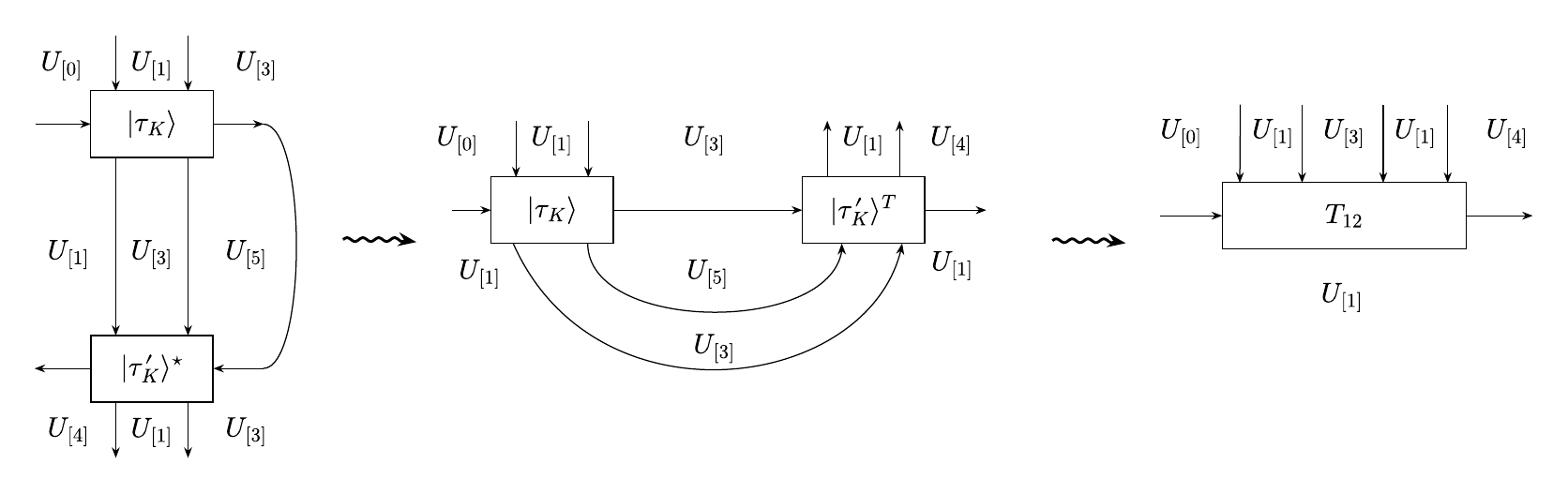}
    \caption{The left-side of the figure displays the tensor network diagram representing the $\beta$-$\gamma$ contraction $\bar K \# K$ for the choices $\tau=\tau_{24}=[0135],\, \tau' =\tau_{02}= [3451]$ of the labelled tetrahedra. The 1-morphism labels are omitted for clarity. We have also used the pivotality of $\cC$ to reorient part of the diagram. The tensor $T_{12}$ is introduced to simplify the notation.}
    \label{fig:contracted}
\end{figure}

Under an appropriate application of co/evaluations $e,\iota$, we can form the "partial contraction" $\bar K' \# K$ of the corresponding (contracted) tetrahedral states $K' = |\tau'_{K}\rangle,\, K= |\tau_K\rangle$ along their shared legs; see Fig. \ref{fig:contracted}.  Here, $\bar K'=|\tau'_K\rangle^\star$ denotes the conjugate/orientation reversed $\cC$-decoration.

In the following, we will assume that the contraction of tensor network diagrams can be associated freely.\footnote{Given the recoupling theory of \S \ref{pentagonatorscalars}, this requires the pentagonator $\pi$ to satisfy the associahedron equation.} All tensor network diagrams below will implicitly have its decorations summed over.

\medskip

For each triple $(i,j,k)$ of indices with $0\leq i<j<k\leq 5$, we denote by $\{\tau_x(ijk)\}_{x=a,b,c}$ the set of three (and exactly three) tetrahedra which share $i,j,k$ as vertex labels (i.e. $a\in\mathbb{Z}_5$ but $a\neq i,j,k$). 
\begin{proposition}
    Let $\kappa\to\cC$ denote a $\cC$-decoration on the 5-simplex $\kappa=[012345]$. For each triple $(i,j,k)$ with $0\leq i<j<k\leq 5$, we denote for $0\leq a<b\leq 5$ by 
    \begin{equation*}
        K_{ab} = \bar K_a(ijk)\# K_b(ijk),\qquad a,b\not\in\{i,j,k\}
    \end{equation*}
    the $\beta$-$\gamma$ contraction of the states corresponding to the tetrahedra $\tau_x(ijk)$. Then, we have the identity
    \begin{equation}
        \sum_{a,b,c}K_{ab} \ast K_{bc}\ast K_{ca} = (\delta_{ij}\delta_{jk})^{-2}\sum_{a,b,c} K_{aa}\ast K_{bb}\ast  K_{cc}\label{3-3identity}
    \end{equation}
    where $\delta_{ij}=\operatorname{dim}F_{[ij]}$ are the $\delta$-scalars described in \S \ref{scalaringredients}.
\end{proposition}
\begin{proof}

Fix $i,j,k$, we denote by $K_a=K_a(ijk)$. Then the identity \eqref{3-3identity} can be obtained from the computation
\begin{align*}
        \sum_{a,b,c}K_{ab} \ast K_{bc}\ast K_{ca} &= \sum_{a,b,c}  (\bar K_a\#  K_b)\ast(\bar K_b\# K_c )\ast (\bar K_c\# K_a) \\
        &= \sum_{a,b,c} ( K_a \ast\bar K_a)\#(K_b\ast \bar K_b)\# (K_c \ast  \bar K_c) \\
        &\stackrel{\text{fig. } \ref{fig:33proof}}{ =} \delta_{ij}^{-2}\delta_{jk}^{-2} \sum_{a,b,c} (K_a\# \bar K_a)\ast (K_b \#  \bar K_b)\ast (\bar K_c\# K_c) \\ 
        &= \delta_{ij}^{-2}\delta_{jk}^{-2} \sum_{a,b,c} K_{aa}\ast K_{bb}\ast K_{cc}\,,
\end{align*}
where in the second line we have used the associativity of the 2-morphism composition and the associahedron equation for the pentagonator $\pi$. For the equality in the third line, we prove it diagrammatically as displayed in Fig. \ref{fig:33proof}.

\begin{figure}
    \centering
    \includegraphics[width=0.85\linewidth]{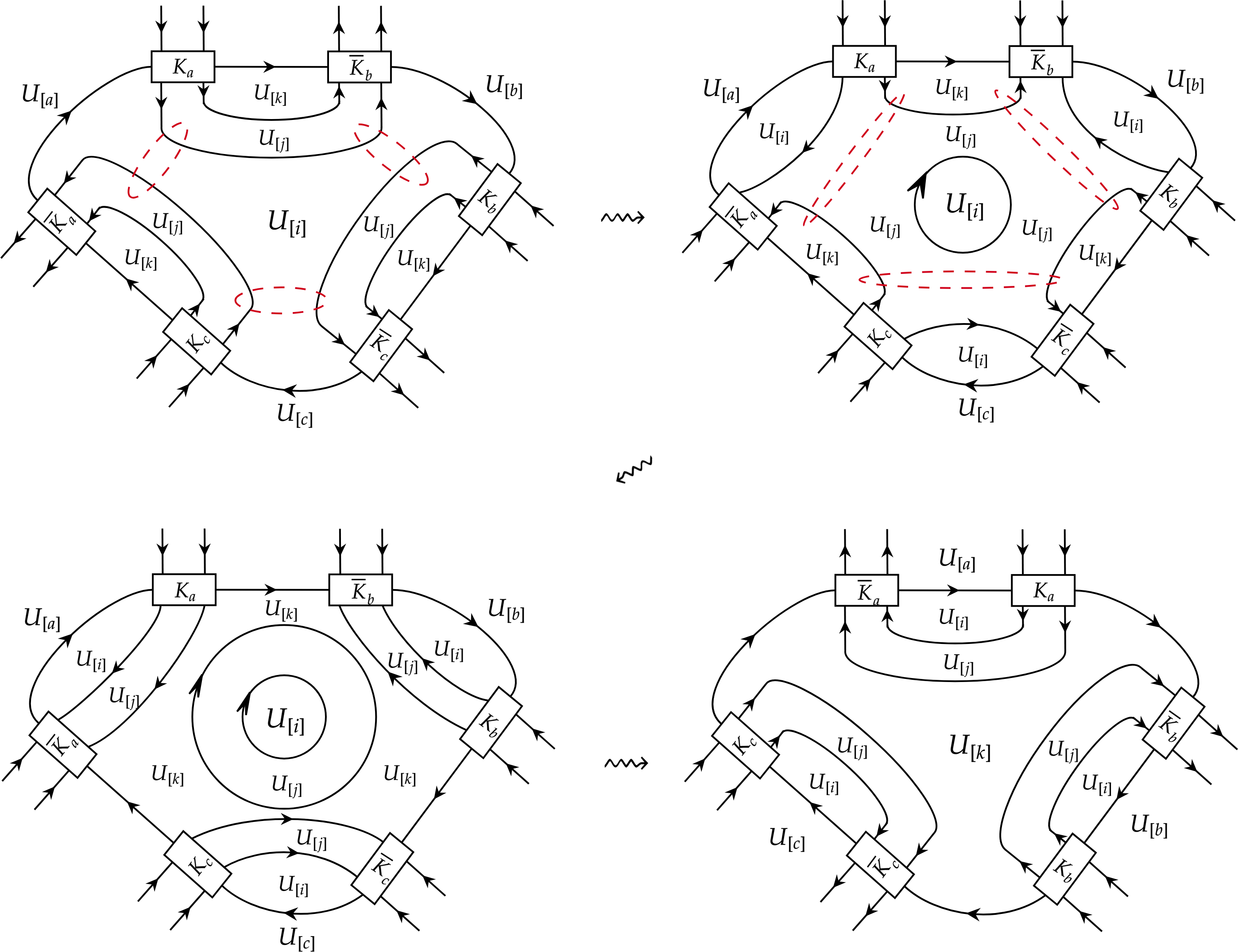}
    \caption{This figure displays the diagrammatic proof of \eqref{3-3identity}, but with the scalar factors appearing in each step suppressed. Here, each red circle is a inverse local $\delta$-move as in Fig. \ref{fig:deltascalar}, and the oriented circles are evaluated to $\delta_{ij}=\operatorname{dim}F_{[ij]},\,\delta_{jk}=\operatorname{dim}F_{[jk]}$ by using planar-pivotality.}
    \label{fig:33proof}
\end{figure}
\end{proof}


\begin{figure}
        \centering
        \includegraphics[width=0.65\linewidth]{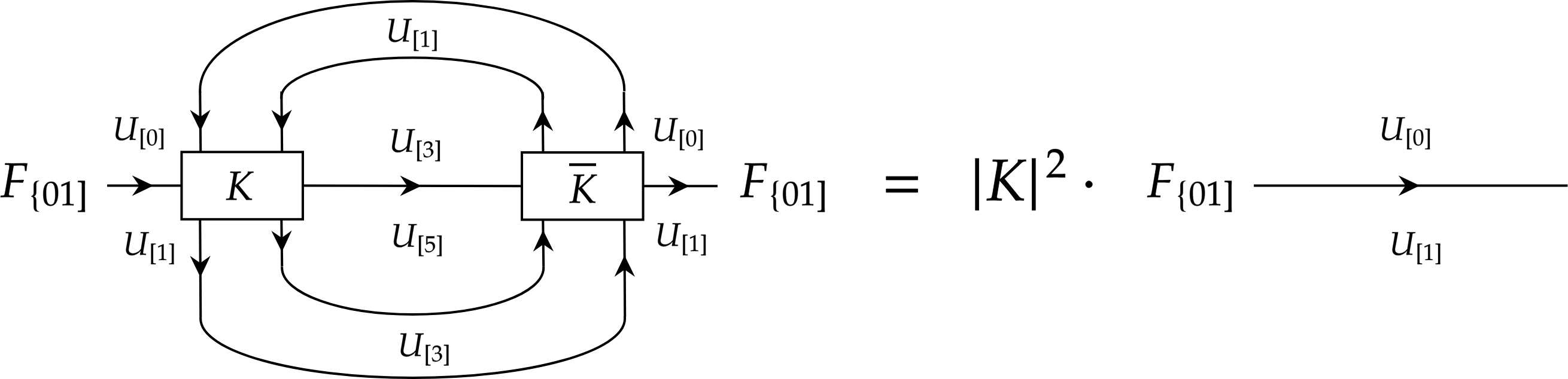}
        \caption{The tensor network presentation of the full contraction of a $\cC$-decorated tetrahedron $K$ with its orientation-reversal $\bar K$.}
        \label{fig:contracted3simplex}
    \end{figure}
    
We now prove the analogue of Corollary 4.3.10 of \cite{douglas2020dualizable}.
\begin{proposition}\label{sectionretraction}
    Let $\mathcal{K}: \sigma\to \cC$ denote a simple $\cC$-decorated 4-simplex and $\bar{\mathcal{K}}$ its orientation-reversal. Then upon summing over $\mathcal{K}$, the pair $\sum_\mathcal{K}(\sigma(\mathcal{K}),\sigma(\bar{\mathcal{K}}))$ can be made into a section-retraction pair.
\end{proposition}
\begin{proof}
    For a $\cC$-decorated tetrahedron $K$ and its orientation-reversal $\bar K$, we can form its "full contraction" by evaluating off all of its matched legs; this is displayed in Fig. \ref{fig:contracted3simplex}. By dagger adjunctions \eqref{unitaryadjoint}, this gives rise to some scalar multiple $|K|^2\in\mathbb{C}$ of the identity 2-endomorphism, which is nothing but the square of \eqref{tetrahedronscalar}. By simplicity of the $\cC$-decoration, $|K|^2\neq 0$ does not vanish. 
    
    Now to prove the statement, consider a $\cC$-decorated 4-simplex $\mathcal{K}:\sigma=\sigma_5\to\cC$. To be explicit, we pick a specific presentation of $\mathcal{K}$ as the tensor network diagram on the right-hand side of Fig. \ref{fig:amplitudegraph}, but other configurations (modulo $\operatorname{ker}\Tr$; see also \S \ref{explicitdef}) yield the same result. 

    \begin{figure}
        \centering
        \includegraphics[width=0.8\linewidth]{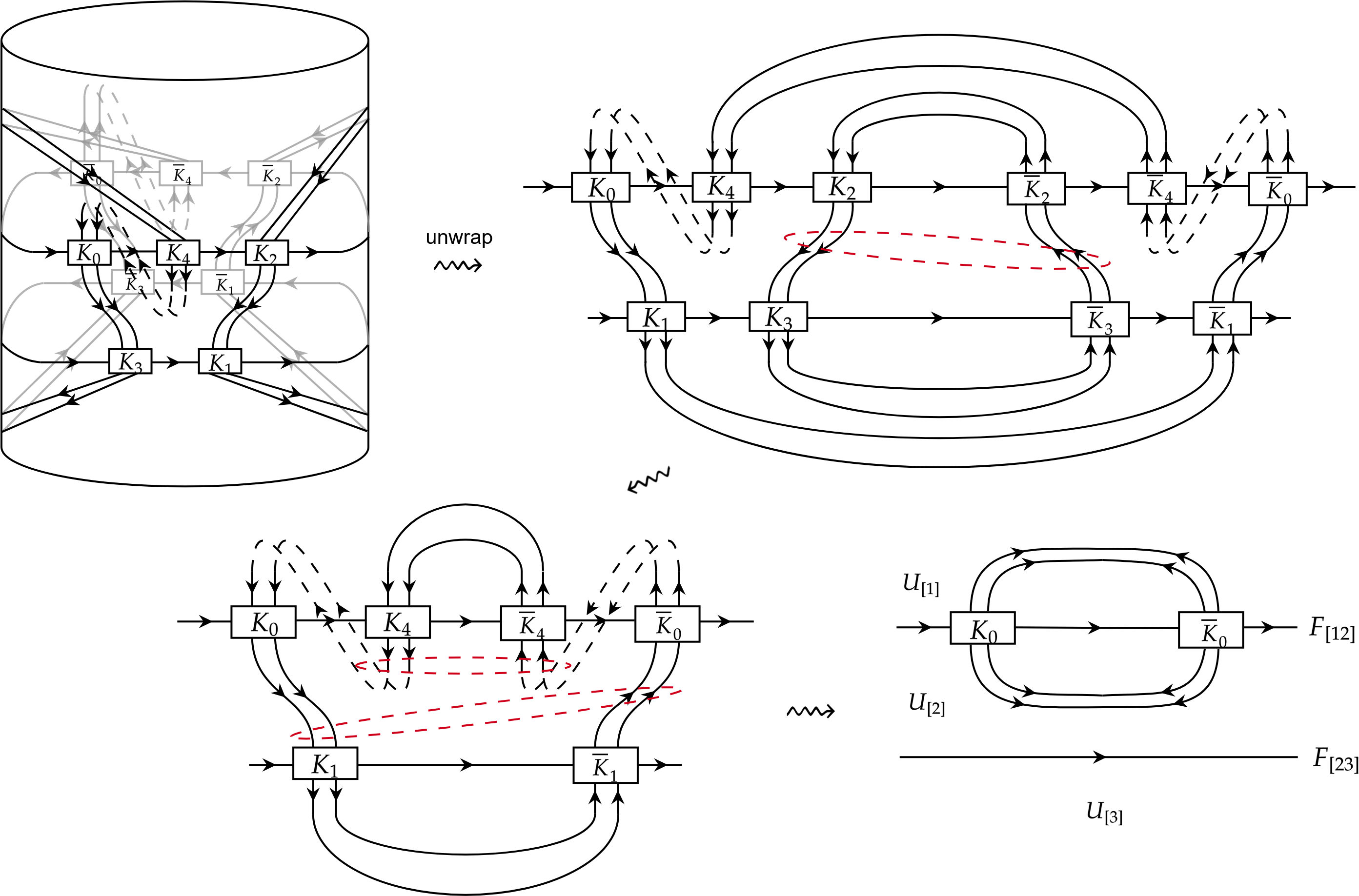}
        \caption{The diagrammatic manipulation which demonstrates \eqref{invertibility}. The dashed red circles are applications of an inverse $\delta$-move as in Fig. \ref{deltascalars}, and the resulting scalars, ie. $\delta$'s and $|K|$'s, are suppressed for clarity.}
        \label{fig:retsecfigure}
    \end{figure}
    
    The strategy is to contract all legs between $K_i,\bar K_i$ for each tetrahedron $K_i\subset \partial\sigma$ under a \textbf{T}-shadow. This is displayed on the top left of Fig. \ref{fig:retsecfigure}.  By following the diagrammatic manipulations shown, we see that we have 
        \begin{equation}
        \sum_\mathcal{K}\operatorname{tr}_\textbf{T}(\sigma(\mathcal{K})\ast\sigma(\bar{\mathcal{K}})) = N\cdot \id\label{invertibility}
    \end{equation}
    with the normalization constant given by 
    \begin{equation}    N=\sum_\mathcal{K}\delta_{50}^{-1}\delta_{01}^{-1}\delta_{12}^{-1}\delta_{23}^{-1}\delta_{34}^{-1}\delta_{45}^{-1}\prod_{\tau_i\subset\partial\sigma}|K_i|^2= \sum_\mathcal{K}\delta_{50}^{-1}\delta_{01}^{-1}\delta_{12}^{-1}\delta_{23}^{-1}\delta_{34}^{-1}\delta_{45}^{-1}\prod_{\tau\in\sigma_5}|K(\tau)|^2\neq0 \,.\label{constant}
    \end{equation}
    By choosing a square-root for $N$, we see that $\sum_\mathcal{K}1/\sqrt{N}(\sigma(\mathcal{K}),1/\sqrt{N}\sigma(\bar{\mathcal{K}}))$ is a section-retraction pair.
\end{proof}


\begin{rmk}\label{amplitudenorm}
    We claim that reflection-positivity uniquely determines the positive square root $\sqrt{N}$. To see this, consider the full trace of the quantity \eqref{invertibility} (up to the natural isomorphism $d$) $$N=\sum_\mathcal{K}\Tr(\sigma(\mathcal{K})\ast\sigma(\bar{\mathcal{K}})) = \sum_\mathcal{K}\operatorname{Tr} \left[\operatorname{tr}_\textbf{T}(\sigma(\mathcal{K})\ast\sigma(\bar{\mathcal{K}}))\right] = \sum_\mathcal{K}|\langle\sigma(\mathcal{K})\rangle|^2\,,$$ which is precisely the squared 4-simplex summation/the normalization constant in \eqref{partition}. This is computed by recoupling the $\theta$-nets, or diagrammatically by closing the torus in Fig. \ref{fig:retsecfigure} and wrapping $\sigma(\bar{\mathcal{K}})$ to the front of the other side to obtain Fig. \ref{fig:positivetorus}.
\end{rmk}

\subsection{Invariance under the 4d Pachner moves}\label{pachnerinvariance}

In 4-dimensions, Pachner moves are associated with bipartitions of the boundary ($\varepsilon_i=\pm$ denotes the orientation of the 4-simplex $\sigma_i$)
\begin{equation*}
    \partial \kappa = \bigcup_{i=0}^5 \sigma_i^{\varepsilon_i},\qquad \sigma_i = [0\cdots \hat i\cdots 5]
\end{equation*}
of an ordered 5-simplex $\kappa=[012345]$ into two non-empty disjoint sets $(k,l)$ of 4-simplices \cite{Kashaev2018,Kashaev:2015}. For $k+l=6$, we have $(3,3)$, $(2,4)$ and $(1,5)$.

For each $\sigma_i$, we denote by $\mathcal{K}_i: \sigma_i\to \cC$ the corresponding $\cC$-decorated 4-simplex, and denote by $K_{ij}:\tau_{ij}=\partial_j\sigma_i\to \cC$  the corresponding decorated boundary tetrahedra.


\subsubsection{3-3 Pachner move}
We begin with the 3-3 Pachner move. Fix the $\cC$-decorated 4-simplices $\mathcal{K}_i$, $i=0,\dots, 5$ lying on the boundary of the 5-simplex $\kappa=[012345]$. We recall that the corresponding 4-simplex scattering amplitudes $\langle\sigma(\mathcal{K}_i)\rangle$ can be represented as $\cC$-decorated torii, as in Fig. \ref{fig:amplitudegraph}.

In the following, we shall always assume that these $\cC$-decorated torii are embedded in $S^3$ (see also \textit{Remark \ref{3sphericality}}). By the notation "$\sum$" below we mean a summation over all relevant decorations. 
\begin{theorem}
    Up to isotopy, the 4-simplex scattering amplitudes satisfy the \textbf{3-3 Pachner move}
    \begin{equation}
        \sum \langle\sigma({\mathcal{K}_1}) \ast  \sigma({\mathcal{K}_3}) \ast  \sigma({\mathcal{K}_5})\rangle = \sum (\delta_{02}\delta_{24}\delta_{13}\delta_{35})^2  \langle\sigma({\mathcal{K}_0})\ast\sigma({\mathcal{K}_2})\ast\sigma({\mathcal{K}_4})\rangle\,.\label{33move}
    \end{equation}
\end{theorem}
\begin{proof}
    To prove this, we express the left hand side of the equation \eqref{33move} as diagram of tensor network in Fig. \ref{fig:135step1}. Let $i=1,j=3,k=5$ and $a=0,b=2,c=4$ in \eqref{3-3identity}.  The factor $(\delta_{13}\delta_{35})^2$ appears upon applying the inverse of the move \eqref{3-3identity} to the 1-3-5 contraction.

        \begin{figure}[t]
        \centering
\includegraphics[width=1.0\linewidth]{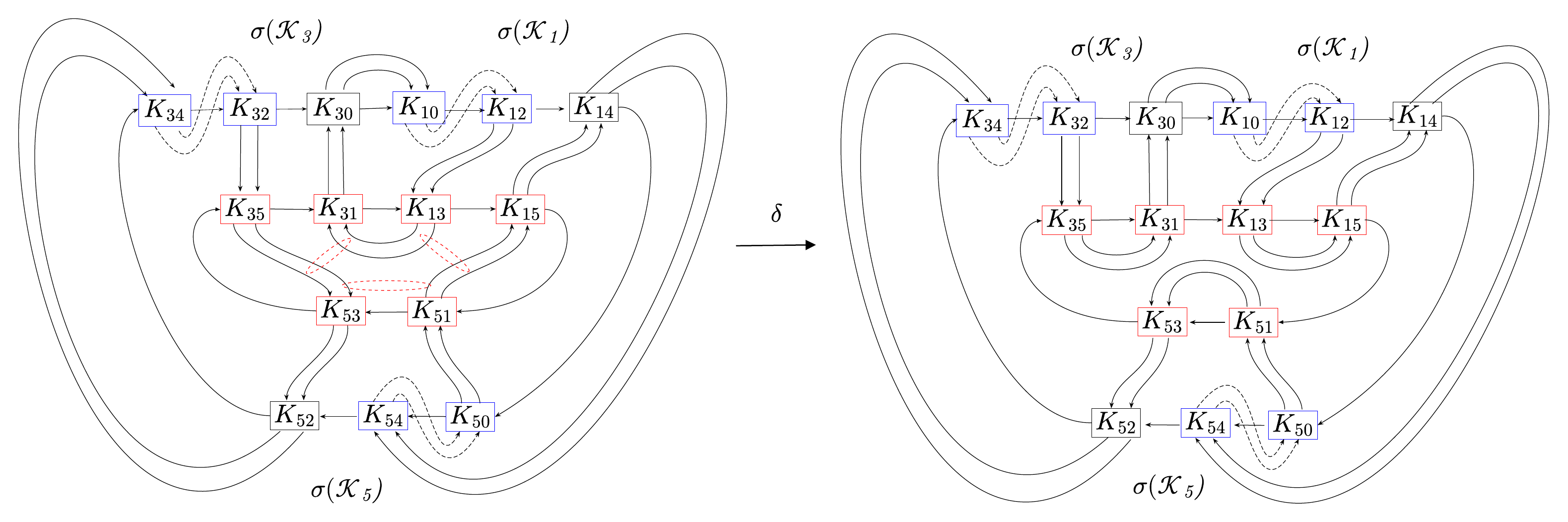}
        \caption{Tensor network diagram representing the contraction of the scattering amplitudes $\sigma({\mathcal{K}_1})$, $\sigma({\mathcal{K}_3})$ and $\sigma({\mathcal{K}_5})$ (left), and diagram after performing the $\delta$ moves highlighted in red (right).  }
        \label{fig:135step1}
    \end{figure}

   We will simplify the diagram by the identity in Fig. \ref{fig:contracted} and then embed the diagram inside $S^3$. 
   Following the Fig. \ref{fig:135step2}, we apply several isotopies and $\delta$ moves. 
   The last diagram of Fig. \ref{fig:135step2} can be obtained from the scattering amplitude of $\sigma({\mathcal{K}_0})$, $\sigma({\mathcal{K}_2})$ and $\sigma({\mathcal{K}_4})$ in Fig. \ref{fig:024}. 
     The factor $(\delta_{02}\delta_{24})^2$ appears upon applying the inverse of the move \eqref{3-3identity} to the 0-2-4 contraction.


        \begin{figure}[!htp]
        \centering
\includegraphics[width=0.8\linewidth]{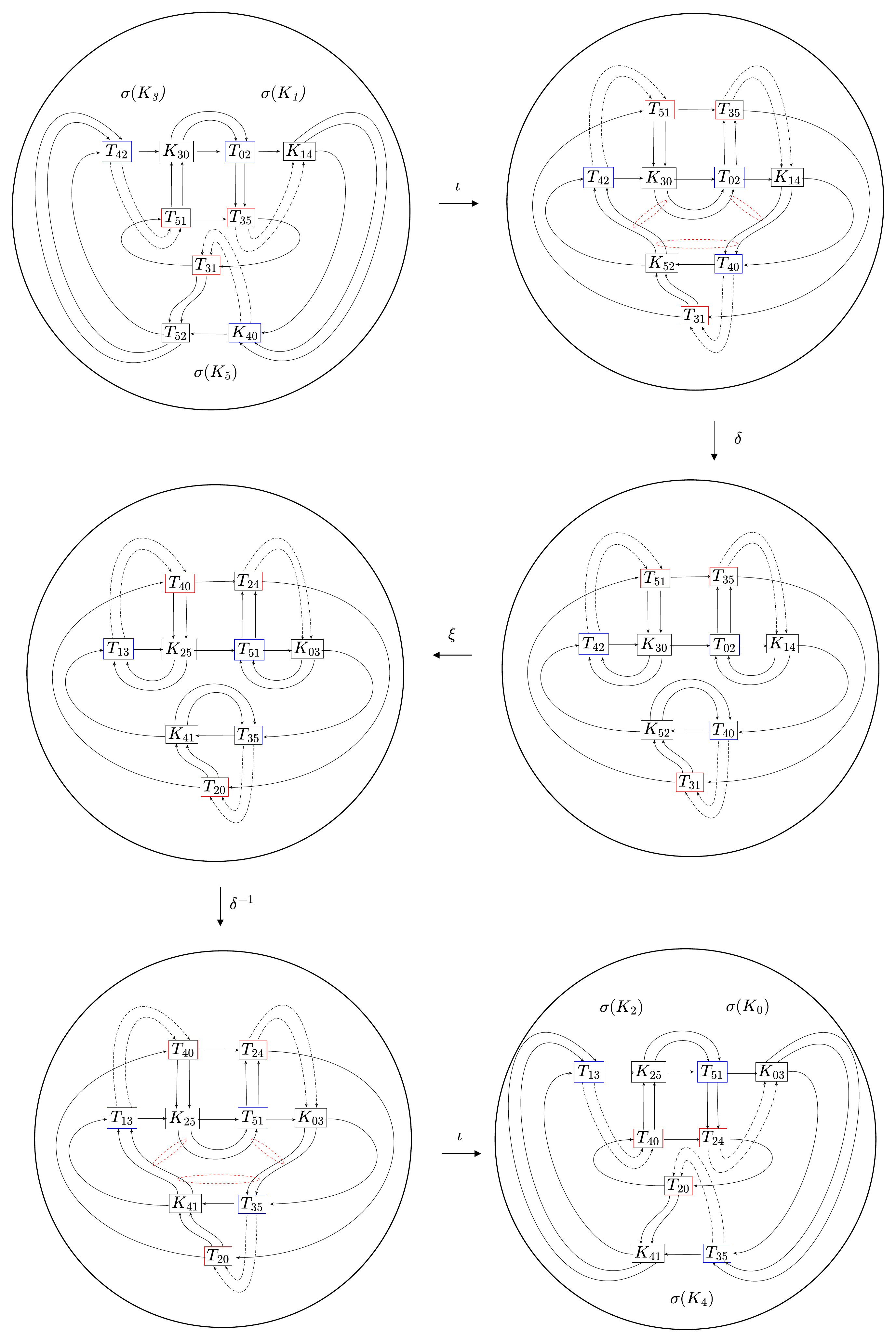}
        \caption{Manipulations of diagrams where  $\iota$ represents the isotropy by pulling the tensor outwards (red) or inwards (blue) inside $S^3$,  the operation $\xi$ is from cyclic permutation of labels $\kappa=[012345]$, and $\delta$ represents the manipulation \eqref{3-3identity}.    }
        \label{fig:135step2}
    \end{figure}

        \begin{figure}[!htp]
        \centering
\includegraphics[width=1.0\linewidth]{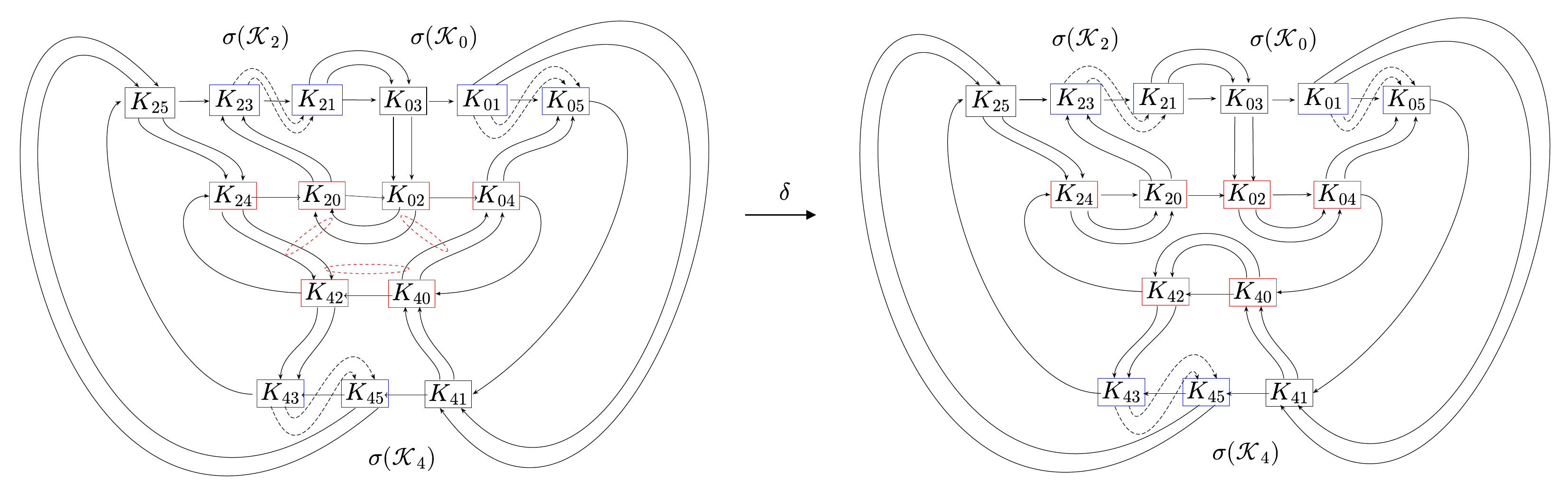}
        \caption{Tensor network diagram representing the contraction of the scattering amplitudes $\sigma({\mathcal{K}_0})$, $\sigma({\mathcal{K}_2})$ and $\sigma({\mathcal{K}_4})$ (left), and diagram after performing the $\delta$ moves highlighted in red (right).}
        \label{fig:024}
    \end{figure}

\end{proof}

We will now use the 3-3 Pachner move to deduce the rest of the Pachner moves. This explains why the 3-3 Pachner move is often regarded as the most important one \cite{Kashaev:2015,Kashaev2018,Douglas:2018}.

\subsubsection{4-2 and 5-1 Pachner moves}
We follow closely the proof strategy in \S 4.3.4 of \cite{Douglas:2018}.
\begin{theorem}
Given \eqref{33move}, the 4-simplex scattering amplitudes satisfy the \textbf{4-2 Pachner move}
    \begin{equation}
        \sum (\delta_{02}\delta_{24}\delta_{13}\delta_{35})^{-2}N_5\langle\sigma({\mathcal{K}_1})\ast\sigma({\mathcal{K}_3})\rangle = \sum \langle\sigma({\mathcal{K}_0})\ast\sigma({\mathcal{K}_2})\ast\sigma({\mathcal{K}_4})\ast\sigma(\bar{\mathcal{K}}_5)\rangle \,\label{42move}
    \end{equation}    
    up to isotopy. Given \eqref{42move}, the 4-simplex scattering amplitudes satisfy the \textbf{5-1 Pachner move}
    \begin{equation}
         \sum\left(\frac{\prod_{e\in\partial\kappa}\delta_e}{\prod_{\tau\in\sigma_3}|K(\tau)|\prod_{\tau\in\sigma_5}|K(\tau)|}\right)^2\langle\sigma({\mathcal{K}_1})\rangle =\sum \langle\sigma({\mathcal{K}_0})\ast\sigma({\mathcal{K}_2})\ast\sigma({\mathcal{K}_4})\ast\sigma(\bar{\mathcal{K}}_5)\ast\sigma(\bar{\mathcal{K}_3})\rangle \,\label{51move}
    \end{equation}    
    up to isotopy.
\end{theorem}
\begin{proof}
    The strategy is the following. By the 3-3 Pachner move \eqref{33move}, the contractions on the right-hand side of \eqref{42move} can be written as 
    $$\text{r.h.s.} = \sum (\delta_{02}\delta_{24}\delta_{13}\delta_{35})^{-2}\langle\sigma({\mathcal{K}_1})\ast\sigma({\mathcal{K}_3})\ast\sigma({\mathcal{K}_5})\ast\sigma(\bar{\mathcal{K}}_5)\rangle\,.$$ However, by  \textbf{Proposition \ref{sectionretraction}} and \eqref{invertibility} the two contracted tensor networks $\sum_{\mathcal{K}_5}(\sigma({\mathcal{K}_5}),\sigma(\bar{\mathcal{K}}_5))$ form a section-retraction pair up to the constant $N=N_5\in\mathbb{C}^\times$ \eqref{constant}, hence we have $$ \text{r.h.s} = \sum (\delta_{02}\delta_{24}\delta_{13}\delta_{35})^{-2}N_5\langle\sigma({\mathcal{K}_1})\ast\sigma({\mathcal{K}_3})\rangle= \text{l.h.s}\,.$$
    
    Similarly for \eqref{51move}, by using the 4-2 Pachner move \eqref{42move} the contraction on the right-hand side can be written as $$\text{r.h.s} = \sum(\delta_{02}\delta_{24}\delta_{13}\delta_{35})^{-2}N_5 \langle\sigma({\mathcal{K}_1})\ast\sigma({\mathcal{K}_3}) \ast \sigma(\bar{\mathcal{K}}_3)\rangle \stackrel{\eqref{invertibility}}{=} \sum(\delta_{02}\delta_{24}\delta_{13}\delta_{35})^{-2}N_5N_3\langle\sigma(\mathcal{K}_1)\rangle = \text{l.h.s.}\,.$$ Now by \eqref{constant} we have
    \begin{equation*}
        (\delta_{02}\delta_{24}\delta_{13}\delta_{35})^{-2}N_5N_3 = \frac{\prod_{\tau\in\sigma_5}|K(\tau)|^2\prod_{\tau\in\sigma_3}|K(\tau)|^2}{\prod_{e\in\partial\kappa}\delta_e^2}\,,
    \end{equation*}
    where the denominator is a product over both the cyclic and non-cyclic edges (cf. Tab. \ref{tab:types})  in the 4-simplices contained in $\partial\kappa,$ whence the result follows. We leave the fun diagrammatic manipulations to the reader.
\end{proof}

\section{$\cC$-decorated cells}
\label{latticemodel}

In this section, we construct the membrane-net lattice model with the input datum $\cC$ \eqref{input}.

\subsection{Membrane-net lattice Hilbert space}
\label{3dlevinwen}

Consider an oriented combinatorial cell complex. We shall focus on its 3-skeleton, which we refer to as a \textit{membrane graph} $\Gamma$. 
\begin{definition}
    A \textbf{$\cC$-decoration on $\Gamma$} consists of the following assignments:
\begin{itemize}
    \item the 3-cells 
    are assigned with objects $\ttU\in\cC$,
    \item a 2-cell interfacing two 3-cells with decorations $\ttU,\ttV\in\cC$ is assigned a 1-morphism $F:\ttU\to \ttV$,
    \item a 1-cell interfacing two 2-cells with decorations $F,G\in\operatorname{Hom}_\cC(\ttU,\ttV)$ is assigned a 2-morphism $\eta:F\Rightarrow G$;
    \item 2-cells with the opposite orientation is assigned the 1-morphism adjoint $F^\dagger: {\ttV}\to {\ttU}$, and 
    \item 1-cells with the opposite orientation is assigned the 2-morphism adjoint $\eta^T: G\Rightarrow F$.
\end{itemize}
We denote such $\cC$-decorated cells by $\gamma\in \cC^\Gamma$, and call it a \textbf{membrane-net}. 
\end{definition}

The Hilbert space at a graph vertex/0-cell $v\in \Gamma$ is constructed as follows; see also Fig. \ref{fig:lwhilbspace}.
\begin{enumerate}
    \item \textbf{Edge degrees-of-freedom}: let the oriented surfaces $f_1,\dots,f_n$ incident upon an edge $e$, such that they partition a 3-ball containing $e$ into disjoint 3-cells. We decorate these 3-cells with objects $\ttU_1,\dots,\ttU_{n}\in\cC$. Given an ordering of these surfaces starting at, say, $f_1$, we set (note $F_n: \ttU_n\to \ttU_1$) $$ \cH_{e}(\gamma) =\operatorname{2Hom}_\cC\left(F_{f_n}\circ ( \dots\circ (F_{f_2}\circ F_{f_1}))\dots ),1_{\ttU_1}\right),\qquad F_i: \ttU_i\to \ttU_{i+1}\,.$$

    \item \textbf{Vertex Hilbert space}: let $\mathsf{st}^1(v)=\mathsf{st}^{1,+}(v)\cup\mathsf{st}^{1,-}(v)$ denote the set of oriented edges in $\Gamma$ incident upon a given vertex $v$.  Given an ordering of how the edges $e^\pm\in\mathsf{st}^{1,\pm}(v)$ are attached at $v$, we define 
    \begin{equation}
        \mathcal{H}_v(\gamma) = \bigotimes_{e\in\mathsf{st}^1(v)}\mathcal{H}_{e^\pm}(\gamma) \cong \operatorname{Hom}\left(\bigotimes_{e\in\mathsf{st}^{1,-}(v)}\cH_e(\gamma),\bigotimes_{e\in\mathsf{st}^{1,+}(v)}\cH_e(\gamma)\right) \,,\label{vertexspace}
    \end{equation} 
    where $\mathcal{H}_{e^+}(\gamma)=\mathcal{H}_e(\gamma)$ denotes the edge Hilbert space for incoming edges $e\in\mathsf{st}^{1,+}(v)$ and the dual space $\mathcal{H}_{e^-}=\mathcal{H}_e^*$ for outgoing edges $e\in\mathsf{st}(v)^{1,-}$. 
\end{enumerate}
The \textit{orientation reversal} $v^-$ of a vertex $v$ is one in which all edges and surfaces are attached to $v$ in the opposite orientation. 

\begin{figure}
    \centering
    \includegraphics[width=0.85\linewidth]{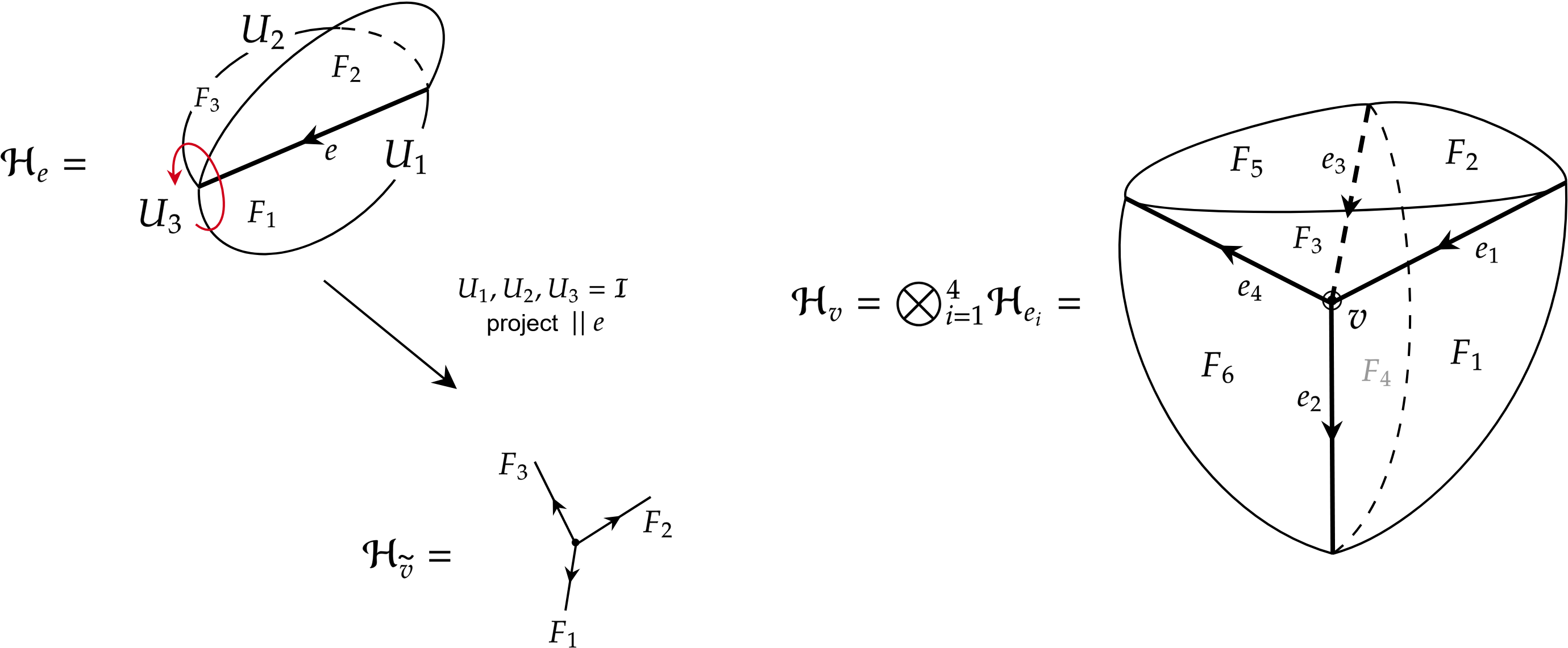}
    \caption{For the lattice graph $\Gamma$ Poincar{\'e} dual to a tetrahedron $\tau\subset D^3$ inscribed within a 3-ball $D^3$, these figures display the Hilbert spaces $\mathcal{H}_e,\mathcal{H}_v$ attached to the edges $e$ and vertices $v$ in our 3d higher Levin-Wen model. On the left, we also demonstrate the fact that, when the 3-cells are trivially-decorated, we recover the usual 2d Levin-Wen Hilbert space $\mathcal{H}^{\tilde v}$ by projecting the graph along the edge $e$.}
    \label{fig:lwhilbspace}
\end{figure}

\begin{definition}
For an oriented combinatorial cell complex $\Gamma$, we define the \textbf{membrane-graph Hilbert space} as
 \begin{equation}
     \cH^\Gamma=\bigoplus_{\gamma\in \operatorname{Irr}(\cC)^{\Gamma}}\,\bigotimes_{v^\pm\in\Gamma} \mathcal{H}_{v^\pm}(\gamma)\cong \bigotimes_{v^\pm\in\Gamma}\bigoplus_{\gamma\in\operatorname{Irr}\cC}\cH_{v^\pm}(\gamma)\label{totalvertex}
 \end{equation}
the direct sum over {all} simple $\cC$-decorations on the graph $\Gamma$.
\end{definition}
\noindent Specifically for the graph $\Gamma$ dual to a tetrahedron $\tau$, we have one unique vertex $v$ with $|\mathsf{st}^1(v)| =4$ edges incident upon it, and four edges $e_1,\dots,e_4$ each with $n=3$ surfaces $f_{ij}$ attached to it; see Fig. \ref{fig:lwhilbspace}. 

\begin{rmk}
    Let $\gamma\in (\Omega\cC)^\Gamma$ denote a $\cC$-decoration on $\Gamma$ where all 2-cell labels take values in the (braided) fusion category $\Omega\cC= \operatorname{End}_\cC(\ttI)$. By projecting/compactifying the configuration of surfaces along a chosen oriented edge (see Fig. \ref{fig:lwhilbspace}), we obtain a trivalent vertex $\tilde v$. By thinking of 1-endomorphisms $F_{f_i}\in\Omega\cC$ as objects $a,b,c,\dots$ in the fusion category $C=\Omega\cC$, then the edge Hilbert space becomes $$\cH_e(\gamma) =\operatorname{2Hom}_\cC(F_{f_3}\circ F_{f_2}\circ F_{f_1},1_{{\ttI}})  \xrightarrow{\sim}  H_{\tilde v} =\operatorname{Hom}_{C}(a_{\tilde e_1}\otimes b_{\tilde e_2}\otimes c_{\tilde e_3},\ttI)\,,$$ which is precisely the string-net Hilbert space \cite{Levin2004} with input $C$. See also \cite{fuchs2025spherical}.
\end{rmk}





\subsubsection{Equivalence with the triangulated Hilbert space}
We now consider a particular class of membrane graphs $\Gamma$ associated to a 3-manifold. 
\begin{theorem}\label{levinwentet}
    Let $(M^3,T_{M^3})$ be a triangulated closed oriented 3-manifold. If the graph complex $\Gamma$ is dual to the triangulation $T_{M^3}$, then the membrane-graph Hilbert space \eqref{totalvertex} and the triangulated Hilbert space \eqref{totaltetspace} are isomorphic.
\end{theorem}
\begin{proof}
    Consider first the graph $\Gamma$ dual to a (positively-oriented) tetrahedron $\tau^+=\tau\subset D^3$ inscribed within a 3-ball $D^3$; see Fig. \ref{fig:lwhilbspace}. For an edge $e\in\Gamma$ dual to an ordered 2-simplex $\Delta=[012]$, its attached three surfaces $f_1,f_2,f_3$ are dual to respectively the 1-simplices $[01],[12],[20]$ respectively, and the 3-cells in $\pi_0(D^3\setminus\Gamma)$ are dual to the 0-simplices $[0],[1],[2]$. These correspondences are understood implicitly in the following.
    
    Consider a $\cC$-decoration on $\Gamma$. We pick a "starting" 3-cell corresponding to the vertex decoration $\ttU_1=\ttU_{[0]}$. For each edge $e$, define a linear map which sends a 2-morphism $\eta: F_{f_3}\circ (F_{f_2}\circ F_{f_1})\Rightarrow 1_{\ttU_1}$ to a $\cC$-decorated 2-simplex 
    \begin{equation*}
        \eta_{[012]}=(e_{F_{f_3}}\circ (F_{f_2}\circ F_{f_1}))\ast(F_{f_3}^\dagger\circ \eta_\Delta)
    \end{equation*}
     given by \eqref{Deltastates}. By dagger adjunction \eqref{unitaryadjoint}, this define bijections which give
    \begin{align}
        \cH_{e^+}(\gamma)&=\operatorname{2Hom}(F_{3}\circ (F_{2}\circ F_{1}),1_{\ttU_1})\to \operatorname{2Hom}(F_{[12]}\circ F_{[01]},F_{[02]}), \nonumber\\
         \cH_{e^-}(\gamma) &=\operatorname{2Hom}((F_{2}\circ F_{1})\circ F_{3},1_{\ttU_1})\to  \operatorname{2Hom}(F_{[02]},F_{[12]}\circ F_{[01]})\label{dualmap}
    \end{align}
    for the two different orientations on the edge $e$.

    Now consider the vertex space $\mathcal{H}_v$. As can be seen in Fig. \ref{fig:lwhilbspace}, the constituent edge Hilbert spaces (with specified surface ordering) are given by 
    \begin{align*}
        \cH_{e_1^+}(\gamma) = \operatorname{2Hom}(F_3\circ (F_2\circ F_{1}),1_{\ttU_1}),\qquad \cH_{e_4^-}(\gamma)=\operatorname{2Hom}((F_5\circ F_6)\circ F_3,1_{\ttU_3}),\\
        \cH_{e_2^+}(\gamma) = \operatorname{2Hom}(F_2\circ (F_5\circ F_4),1_{\ttU_2}),\qquad \cH_{e_3^-}(\gamma) = \operatorname{2Hom}((F_4\circ F_1)\circ F_6,1_{\ttU_4}).
    \end{align*}
    The above maps \eqref{dualmap} then induce linear isomorphisms 
    \begin{align*}
        \cH_{e_1^+}(\gamma)\otimes\cH_{e_4^-}(\gamma)\xrightarrow{\sim} \operatorname{2Hom}_\cC(F_{[02]},F_{[32]}\circ F_{[03]} )\otimes  \operatorname{2Hom}_\cC(F_{[12]}\circ F_{[01]},F_{[02]}) \\ 
         \cH_{e_2^+}(\gamma)\otimes \cH_{e_3^-}(\gamma)\xrightarrow{\sim}  \operatorname{2Hom}_\cC(F_{[13]},F_{[03]}\circ F_{[10]})\otimes \operatorname{2Hom}_\cC(F_{[23]}\circ F_{[12]},F_{[13]})
    \end{align*}
    with the space spanned by a $\cC$-decoration $\eta_{\Delta^\pm}$ defined in \eqref{Deltastates}. Putting these together, up to a composition associator $\alpha$ we thus achieve a map
    \begin{equation*}
        \cH_{v^+}(\gamma)=\cH_v(\gamma) \xrightarrow{\sim}\cH_{K(\tau)}= \cH_{K(\tau)}^+\,,\qquad K: \tau\to \cC\,,
    \end{equation*}
    which lands precisely upon the tetrahedron states \eqref{taustate+}. 
    
    Through a similar argument with the opposite orientation, we have $\cH_{K(\tau)}^-\cong \cH_{v^-}(\gamma)$. Summing up all (simple) $\cC$-decorations on all oriented tetrahedra in a triangulation $T_{M^3}$ thus gives us $$\cH^\Gamma=\bigoplus_{\gamma\in\operatorname{Irr}(\cC)^{\Gamma}}\bigotimes_{v^\pm\in\Gamma}\cH_{v^\pm}(\gamma)\xrightarrow{\sim} \bigoplus_{K_s: T_{M^3}\to\cC}\,\,\bigotimes_{\tau^\pm\subset T_{M^3}}\cH_{K_s(\tau)}^\pm=\cH_{T_{M^3}}\,,$$ as desired.
    
\end{proof}

\begin{figure}[t]
\begin{subfigure}[h]{0.4\linewidth}
\hspace*{25pt}
\includegraphics[width=1\linewidth]{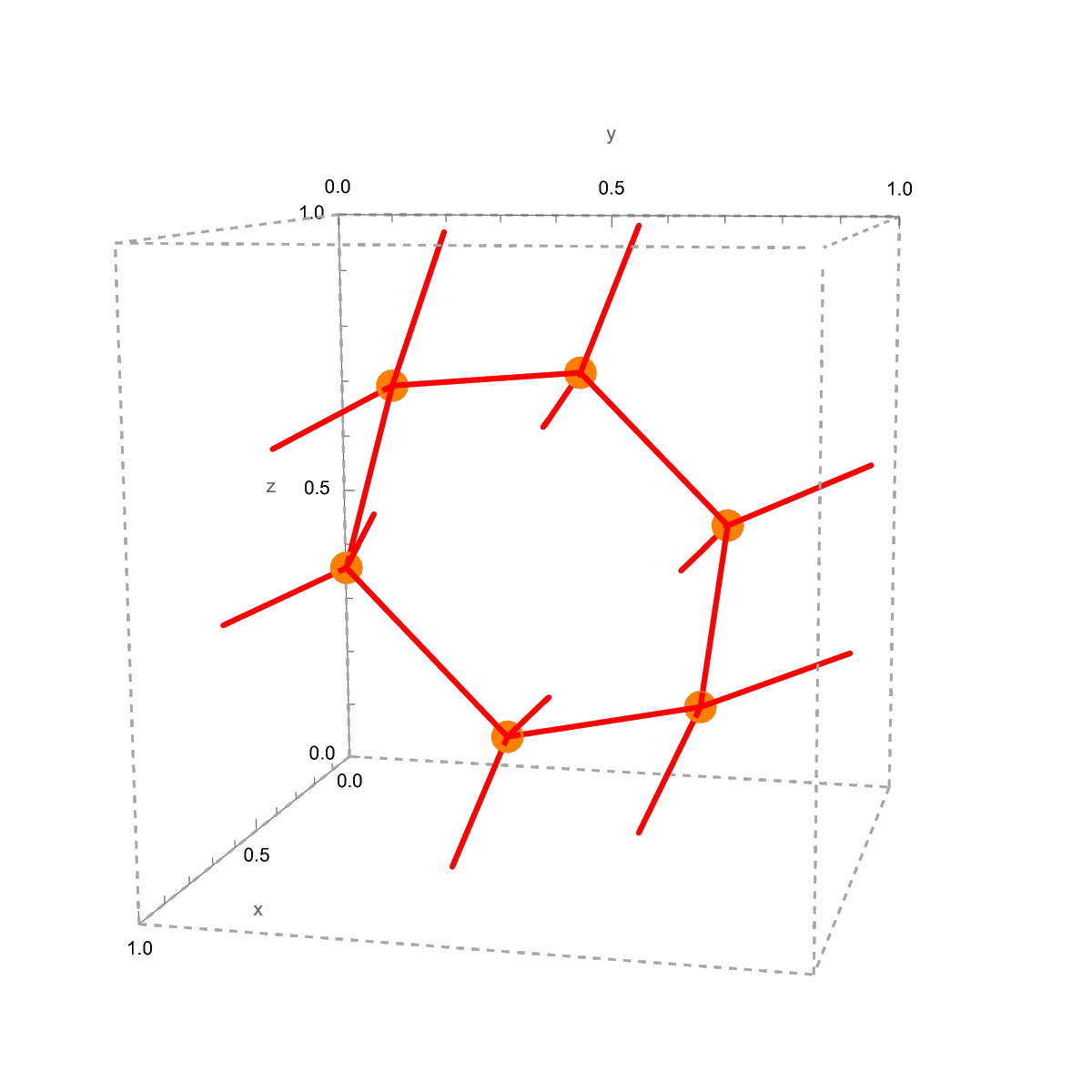}
\end{subfigure}
\hfill
\begin{subfigure}[h]{0.4\linewidth}
\includegraphics[width=0.65\linewidth]{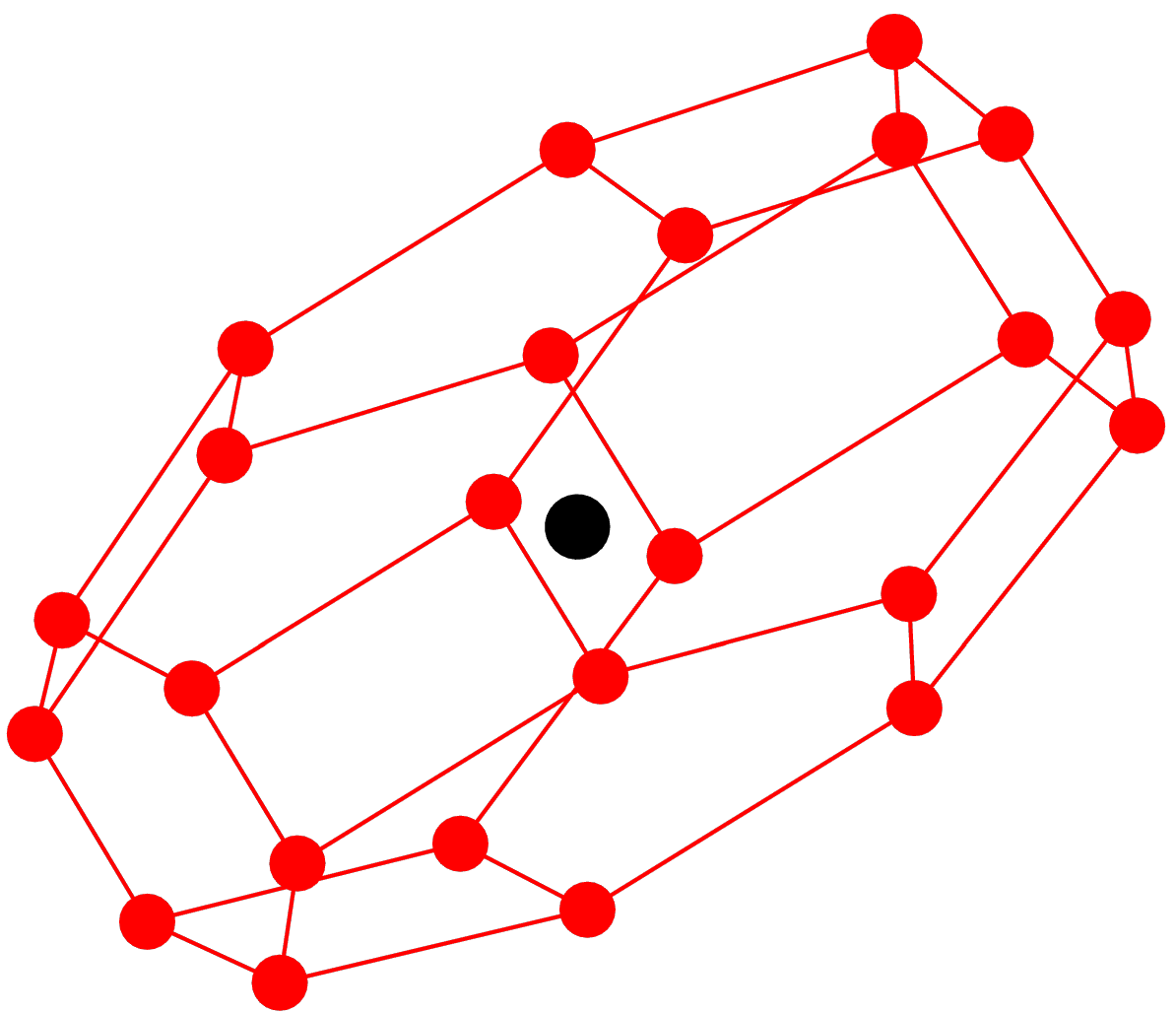}
\end{subfigure}%
\vspace*{-0.5cm}\caption{Illustration of dual lattice model of 3-simplexes forming cubes \cite{Xi:2021wto}. The left hand side figure shows the dualized triangulation of a single cube. The right hand side presents the 3d dual cell structure.}
    \label{fig:MembraneNet}
\end{figure}

    

The lattice given by the dual graph complex of a triangulation is homeomorphic to the so-called \textbf{diamond cubic}. In this case, the (orientation-preserving) point group symmetry $\Lambda^+$ of the graph $\Gamma$, preserving a vertex $v$, is isomorphic to the group of (orientation preserving) symmetries $A_4$ of the tetrahedron $\tau$ dual to $v$. If the 3d manifold admits a good cover, we can always partition the space into cubes, and then into tetrahedra. The corresponding lattice then takes shape of a multi-layer honeycomb-parallelepipeds as shown in Fig. \ref{fig:MembraneNet}.

\begin{rmk}\label{branenetpivotal}
    Note that, strictly speaking, the above definition of the vertex Hilbert space $\cH_v$ depends on the choice of the ordering of the surfaces $f$ and edges $e$ attached to each edge $e$ and the vertex $v$. Different orderings are related by an action of $x\in\Lambda^+$ --- the orientation-preserving lattice point group --- and hence $\cH_v$ should furnish a unitary $\Lambda^+$-representation. In analogy with the string-net construction \cite{Fuchs:2023pyg}, this representation (and equivalently the $A_4$-representation \eqref{orientationrep}) should arise from the {canonical} pivotality of $\cC$ as described in \S \ref{pivotality}.  The dependency on the ordering can then be removed by taking the limit
    \begin{equation*}
        \cH_v(\gamma) = \lim_{x\in\Lambda^+} \cH_{v}(x\cdot \gamma),\qquad x\cdot-:\Gamma^\cC\to {(x\cdot \Gamma)}^\cC\,.
    \end{equation*} 
     of the diagram formed by the isomorphisms induced by $x\in\Lambda^+$.
\end{rmk}

\subsubsection{Lattice local moves}\label{hamiltonian}

As in the original Levin-Wen string-net model, the total Hilbert space of the membrane-net model is defined as the linear span of all admissible labelings of the underlying lattice. Explicitly,
\[
\mathcal{H}_{\rm tot}
=
\operatorname{span}
\bigl\{
\text{fully labeled membrane-net configurations}
\bigr\}\,,
\]
which is isomorphic to $\bigoplus_{\gamma\in(\operatorname{Irr}\cC)^\Gamma}\cH^\Gamma(\gamma)\cong \cH^\Gamma$ \eqref{totalvertex}.

\begin{figure}
    \centering
    \includegraphics[width=0.85\linewidth]{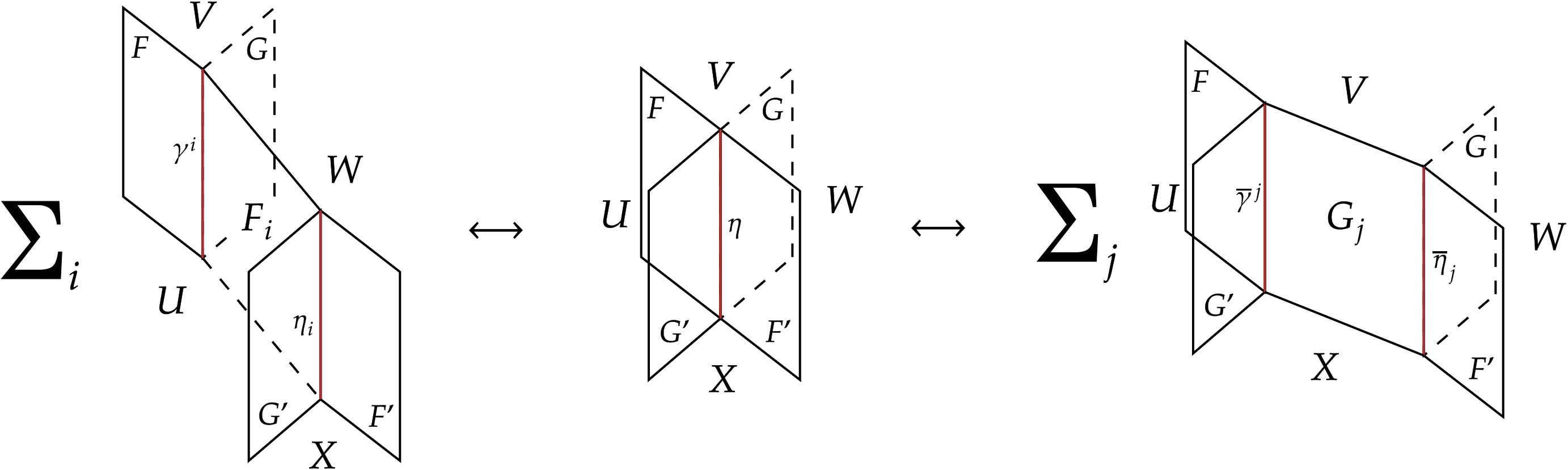}
    \caption{The graphical representation of the natural isomorphism implementing the F-move.}
    \label{fig:Fmove}
\end{figure}

\begin{proposition}\label{Fmove}
    Let $\ttU,\ttV,\ttX,\ttW\in\cC$ and $F:\ttU\to \ttV,\, G:\ttV\to \ttW $ and $G':\ttU\to \ttX,\,F': \ttX\to \ttW$. There are linear isomorphisms
$$\bigoplus_i \operatorname{2Hom}_\cC(G\circ F,F_i)\otimes\operatorname{2Hom}_\cC(F_i,F'\circ G') \cong \bigoplus_j \operatorname{2Hom}_\cC(F,G_j\circ G')\otimes\operatorname{2Hom}_\cC(G\circ G_j,F')\,,$$
where $F_i\in \operatorname{Hom}_\cC(\ttU, \ttW),\, G_j\in \operatorname{Hom}_\cC(\ttX,\ttV)$ are simple. This is the \textbf{F-move}, and is graphically represented by Fig. \ref{fig:Fmove}.
\end{proposition}
\begin{proof}
    By local semisimplicity, there is a natural isomorphisms
\begin{align*}
\operatorname{2Hom}_\cC(G\circ F,F'\circ G')&\cong \bigoplus_i \operatorname{2Hom}_\cC(G\circ F,F_i)\otimes\operatorname{2Hom}_\cC(F_i,F'\circ G'),\\
\operatorname{2Hom}_\cC(F\circ G'^\dagger,  G^\dagger\circ F')&\cong \bigoplus_j \operatorname{2Hom}_\cC(F\circ G'^\dagger,G_j)\otimes\operatorname{2Hom}_\cC(G_j,G^\dagger\circ F')
\end{align*}
 for simple 1-morphisms $F_i\in \operatorname{Hom}_\cC(\ttU, \ttW),\, G_j\in \operatorname{Hom}_\cC(\ttX,\ttV)$. Together with \eqref{conjugation}, we then have
\begin{align*}
    \bigoplus_i \operatorname{2Hom}_\cC(G\circ F,F_i)\otimes\operatorname{2Hom}_\cC(F_i,F'\circ G')&\cong \operatorname{2Hom}_\cC(G\circ F,F'\circ G')\stackrel{\eqref{conjugation}}{\cong}\operatorname{2Hom}_\cC(F\circ G'^\dagger,  G^\dagger\circ F')\\
    &\cong \bigoplus_j \operatorname{2Hom}_\cC(F\circ G'^\dagger,G_j)\otimes\operatorname{2Hom}_\cC(G_j,G^\dagger\circ F')\\
    &\cong  \bigoplus_j \operatorname{2Hom}_\cC(F,G_j\circ G')\otimes\operatorname{2Hom}_\cC(G \circ G_j,F')
\end{align*}
as desired.
\end{proof}
Upon a choice of basis, we will let $\mathsf{M}$ denote the matrix that implements the local semisimplicity decomposition,
\begin{equation*}
    \eta = \bigoplus_{F_i\in\operatorname{Irr}\cC(\mathtt{U},\mathtt{W})} \sum_{\substack{\gamma_i^a:F\circ G\Rightarrow F_i \\ \eta^{i,b}: F_i\Rightarrow F\circ G}} \mathsf{M}(\eta)^{ab}_{i} \gamma_i^a\otimes\eta^{i,b}\,,
\end{equation*}
where $\{\gamma^a_i\},\{\eta^i_b\}$ are bases for $\operatorname{2Hom}_\cC(G\circ F,F_i),\operatorname{2Hom}_\cC(F_i,F'\circ G')$ respectively.

The \textbf{parallel move} is a special case of the F-move, where $\mathtt{V}=\mathtt{X}$ and $\eta=\alpha\circ\beta$  is a composite given by  and $\alpha:F\Rightarrow F'$ and $\beta:G\Rightarrow G'$. 

\subsubsection{The 3-cell operator}
Given a fully-decorated lattice $\Psi\in\cH_\text{tot}$, we define the following operator.

\begin{definition}\label{celloper}
    Let $c \subset \Gamma$ denote a fundamental unit cell (right side of Fig. \ref{fig:MembraneNet}) labelled by an object $\ttU_c \in\cC$. Given a decorated 3-ball $B_3^F$ labelled by an object $\ttV$ and a 1-morphism $F=F_{ c}:\ttU_c\to\ttV$ on its boundary $\partial B^3=S^2$, we define an operator $B_c^F$ which inserts the $F$-decorated $B^3$ into the fundamental cell $c$,
    \begin{equation*}
        B_c^F \Psi  = \Psi\coprod B^3_F\,.
    \end{equation*}
    Of course, if no non-zero 1-morphism $F$ exists between the interior/exterior object labels $\ttV,\ttU_c$, then we put $B_c^F=0$.
\end{definition}

Let us show that the operator $\sum_FB^F_c = B^{\oplus F}_c$ is an idempotent. For simplicity, we assume $\cC$ is strictly planar-pivotal $\phi_F=\id_F$. Consider the product $B_c^FB_c^G$ for which $\ttU\xrightarrow{F}\ttV\xrightarrow{G}\ttW$ are composable simple 1-morphisms in $\cC$. By the parallel move \eqref{idparallel}, we attach the $F$- and $G$-decorated surfaces. Using the pivotality of $\cC$, we then isotopically enlarge this surface of attachment, turning the configuration into a sum of the $F_i$-decorated 3-balls with a decorated 2-sphere attached. This 2-sphere can be evaluated by contracting its 1-cell boundary, giving a sum of $\gamma_i\ast\gamma^i$ over the 2-morphisms $\gamma_i: F\circ G\Rightarrow F_i$.

Now since $\cC$ is $\mathsf{Hilb}$-enriched, it has a canonical dagger structure and hom-space inner product for which, given $\gamma^i=\gamma_i^T$, we have $\gamma_i\ast\gamma^i = \frac{\langle\gamma_i,\gamma_i\rangle}{\operatorname{dim}F_i}\id_{F_i}$. As such, provided the composition of simple 1-morphisms is normalized such that $\gamma_i\in \cC(F\circ G,F_i)$ has norm one, we have
\begin{align*}
         B_c^FB_c^G~ & = \sum_{F_i\in\operatorname{Irr}\cC(\ttU,\ttW)}\sum_{\gamma_i: F\circ G\Rightarrow F_i}\frac{\langle\gamma_i,\gamma_i\rangle}{\operatorname{dim}F_i}B^{F_i}_c = \sum_{F_i\in\operatorname{Irr}\cC(\ttU,\ttW)}\sum_{\gamma_i: F\circ G\Rightarrow F_i}\frac{1}{\operatorname{dim}F_i}B^{F_i}_c \\
         &=\sum_{F_i\in\operatorname{Irr}\cC(\ttU,\ttW)}B^{F_i}_c= B_c^{\oplus_iF_i}\,.
\end{align*}
By summing over simples $F,G$ we have
\begin{equation*}
    B_c^{\oplus F}B_c^{\oplus G} = B_c^{\oplus_iF_i}\,,
\end{equation*}
as desired.
     

\subsection{Membrane-net lattice Hamiltonian}
Given the aforementioned cell operator $B_c$, we can then begin the construction of the Hamiltonian.  Consider a membrane graph $\Gamma$. The \textit{unconstrained} Hilbert space 
$$V=\bigotimes_{v\in \Gamma}V_v = \bigotimes_{v\in \Gamma}\left(\bigotimes_{e\in\mathsf{st}(v)}\mathbb{C}^{L_e}\right)$$
is a composite of local Hilbert spaces $V_v$ at each vertex $v\in \Gamma$, which is itself composite; the dimension at each edge $e$ is given by
$$L_e =  n_e\cdot \#\{\text{simple 1-morphisms of $\cC$}\}\,,$$
where $n_e$ counts the number of faces of $\Gamma$ coincident upon $e$. 
\begin{definition}
      The \textbf{membrane-net lattice Hamiltonian} is
    \begin{equation*}
        \boldsymbol{H} = -\sum_v\mathcal{A}_v - \sum_c\mathcal{B}_c \in \mathcal{B}(V)\,,
    \end{equation*}
    where $\mathcal{A}_v: V_v\to \cH_v=\bigoplus_\gamma\cH_v(\gamma)$ is the projector onto the total vertex Hilbert spaces \eqref{vertexspace}, and
    \begin{equation*}
        \mathcal{B}_c = \sum_{F_c}B^{F_c}_c\,
    \end{equation*}
    in terms of the cell operator of \textbf{Definition \ref{celloper}}.
\end{definition}
The idea is that $\cA_v$ can be written as a collection of delta functions which enforce the decoration rules of \S \ref{3dlevinwen}. As such, the $+1$ eigenspace of $\prod_v\cA_v$ is precisely the membrane-net Hilbert space $\cH^\Gamma\subset V$,
\begin{equation*}
    \prod_v\mathcal{A}_v\Psi = \Psi,\qquad \forall~ \Psi\in\cH_\text{tot}\,.
\end{equation*}
We now evaluate the cell operator on this subspace.

\subsubsection{Pasting decorated surfaces with circle defects}
We begin with two composable decorated surfaces $\mathtt{U}\xrightarrow{F}\mathtt{V}\xrightarrow{G}\mathtt{W}$ with  "circle" 1-cells labelled by 2-morphisms $\alpha: F\Rightarrow H,\, \beta: G\Rightarrow L$. The goal is to perform a parallel move with this decorated stratification.

To describe these decorated circles, we must make use of the 2-dagger structure on $\cC$. We write
\begin{equation*}
    \alpha^T\ast \alpha = |\alpha|^2\cdot \id_F
\end{equation*}
in terms of the identity on $F$. The decorated semicircles with "upper/lower arc" decorated by $\alpha$/$\alpha^T$, and the 0-cell interfacing between these arcs is given by the 2-morphism composition map $\ast$, which produces the scalar $|\alpha|^2$; similarly for $\beta$. When $\cC$ is unitary, these scalars are trivial. Such a stratified surface is displayed on the left of Fig. \ref{fig:stackedUVW}.

\begin{figure}
\begin{subfigure}[h]{0.4\linewidth}
\hspace*{25pt}
\includegraphics[width=0.8\linewidth]{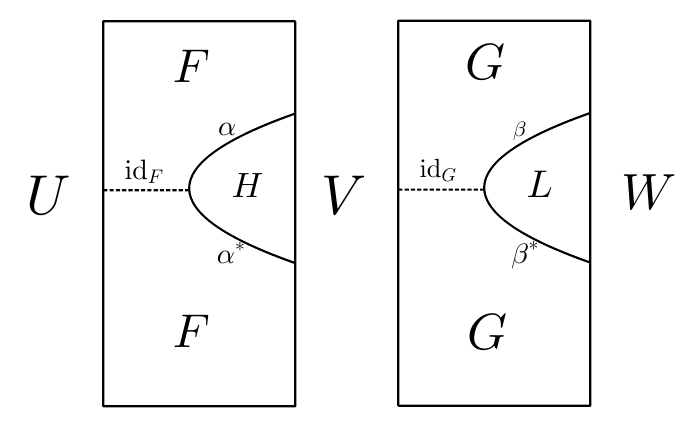}
\end{subfigure}
\hfill
\begin{subfigure}[h]{0.4\linewidth}
\includegraphics[width=0.5\linewidth]{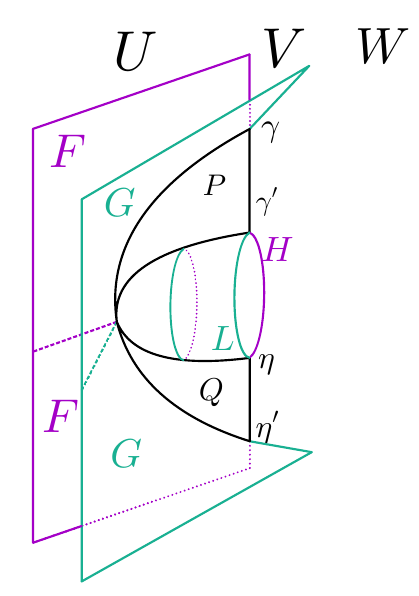}
\end{subfigure}%
\caption{Left: The decorated 2-cells with semicircular arcs as embedded decorated 1-cells. Right: The picture upon performing a parallel move at the 1-cell junctions $\alpha\circ\beta$ and $\alpha^*\circ \beta^* = \alpha^T\circ \beta^T$.}
\label{fig:stackedUVW}
\end{figure}

Then, we can proceed to perform a parallel move with the configurations 
$$F\xRightarrow{\alpha} H\xRightarrow{\alpha^T}F ,\qquad G\xRightarrow{\beta} L\xRightarrow{\beta^T}G\,, $$
which gives rise to a sum over the intermediate simple surfaces $P,Q: \mathtt{U}\to\mathtt{W}$, as well as their interfaces 
$$\gamma: F\circ G\Rightarrow P,\qquad\eta': Q\Rightarrow F\circ G$$
$$\gamma': P\Rightarrow H\circ L,\qquad\eta: H\circ L\Rightarrow Q$$
This is graphically presented as a "cone" made out of the two parallel faces $\mathtt{U}\xrightarrow{H}\mathtt{V}\xrightarrow{L}\mathtt{W}$, as displayed in the right of Fig. \ref{fig:stackedUVW}.

By attaching this configuration with its inverse, we obtain a local move which pastes together two parallel surfaces with circular decorations. This defines the fusion operation for two decorated surfaces, which after choosing a basis for the 2-hom-spaces involved reads (repeated indices are summed)
$$  (\alpha\circ \beta)\otimes(\alpha^T\circ\beta^T) = \bigoplus_{F_i,F_j\in\operatorname{Irr}\cC(\mathtt{U},\mathtt{W})} \sum_{\substack{\gamma_i^a:F\circ G\Rightarrow F_i \\ \eta_j^b: H\circ L\Rightarrow F_j}} \sum_{\substack{\eta^{i,c}:F_i\Rightarrow H\circ L \\ \gamma^{j,d}: F_j\Rightarrow F\circ G}}\mathsf{M}(\alpha\circ\beta)^{ac}_{i}\mathsf{M}(\alpha^T\circ\beta^T)^{bd}_j \gamma^a_i\otimes\eta^{i,c}\otimes\eta^b_j\otimes\gamma^{j,d} \,.$$
By contracting the "bubble" decorated by $H\circ L$, we get a 2-morphism $\eta^b_j\ast\eta^{i,c}:  F_i\Rightarrow H\circ L \Rightarrow F_j$ which by simplicity must be proportional to $\delta_{ij}$. Denote by $\sigma^b_c\delta_{ij}$ this scalar, we then have 
\begin{align*}
(\alpha\circ \beta)\otimes(\alpha^T\circ\beta^T) \xmapsto{\text{eval.}} &\bigoplus_{F_i,F_j\in\operatorname{Irr}\cC(\mathtt{U},\mathtt{W})} \sum_{\substack{\gamma_i^a:F\circ G\Rightarrow F_i \\ \gamma^{j,d}: F_j\Rightarrow F\circ G}}\left(\sum_{b,c}\mathsf{M}(\alpha\circ\beta)^{ac}_{i}\sigma^b_c\delta_{ij}\mathsf{M}(\alpha^T\circ\beta^T)^{bd}_i\right) \gamma^a_i\otimes\gamma^{j,d}\,.
\end{align*}
In the special case where $\alpha=\id_F,\beta=\id_G$ are both the identity (which are self-adjoint), then we put
\begin{equation}
    \id_F\circ \id_G= \bigoplus_{F_i\in\operatorname{Irr}\cC(\mathtt{U},\mathtt{W})} \sum_{\substack{\gamma_i^a:F\circ G\Rightarrow F_i \\ \gamma^{i,d}: F_i\Rightarrow F\circ G}}\mathcal{M}^{ad}_i \gamma^a_i\otimes\gamma^{i,d}\label{idparallel}
\end{equation}
where $\gamma^{i,d}=\gamma_{i,d}^T$ and the matrix $\mathcal{M}$ is constructed out of $\mathsf{M}(\id_{F}\circ \id_{G})=\mathsf{M}(\id_{F\circ G})$.

\medskip

We now apply \eqref{idparallel} to $B_c^F\Psi$, which pastes parts of the 3-ball $B^3$ onto the 2-cell boundaries $f_k\subset\partial c$ of the fundamental cell $c\subset \Gamma$. To be more explicit, given the 1-morphism decorations $\{G_{f_k}\}_{f_k\subset\partial c}$ on these boundary 2-cells, we then have (locally around $c)$
\begin{align}
    B_c^F\Psi\mid_c &= \bigotimes_{k} \id_F\circ \id_{G_{f_k}} \stackrel{\eqref{idparallel}}{=}\bigotimes_k \left(\bigoplus_{F_{i}\in\operatorname{Irr}\cC(\mathtt{U},\mathtt{W})} \sum_{\substack{\gamma_i^a:F\circ G_{f_k}\Rightarrow F_i \\ \gamma^{i,d}: F_i\Rightarrow F\circ G_{f_k}}}\mathcal{M}^{ad}_i \gamma^a_i\otimes\gamma^{i,d}\right)\nonumber\\
    &= \bigoplus_{\{F_{i_k}\}_k\in\operatorname{Irr}\cC(\mathtt{U},\mathtt{W})^{|\partial c|}}\bigotimes_k\left(\sum_{\substack{\gamma_{i_k}^a:F\circ G_{f_k}\Rightarrow F_{i_k} \\ \gamma^{i_k,d}: F_{i_k}\Rightarrow F\circ G_{f_k}}}\mathcal{M}^{ad}_{i_k} \gamma^a_{i_k}\otimes\gamma^{i_k,d}\right),\label{parallelball}
\end{align}
where the 1-morphism $F_{i_k}$ decorate the interior of the $k$-th boundary 2-cell and the 2-morphisms $\gamma_{i_k},\gamma^{i_k}=\gamma_{i_k}^T$ decorate its 1-cell boundary.

\begin{figure}
\begin{subfigure}[h]{0.4\linewidth}
\hspace*{25pt}
\includegraphics[width=0.8\linewidth]{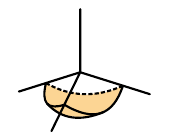}
\end{subfigure}
\hfill
\begin{subfigure}[h]{0.4\linewidth}
\includegraphics[width=0.8\linewidth]{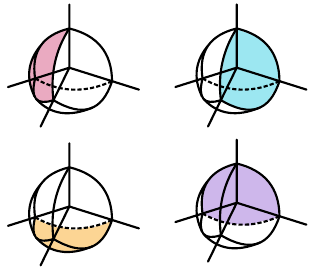}
\end{subfigure}%
\caption{Left: the local picture around a vertex of a fundamental cell $c$ upon an isotopy on the decorated stratification $B_c^F\Psi$, which pushed three of the combined 2-cells to the position of the coloured 2-cell shown here. Right: The contribution of all four cell operators $\prod_{c\in \mathsf{st}^3(v)}B_c^{F_c}\Psi$ around a central vertex $v$. The contribution from different cells are colour-coded.}
\label{fig:4branes}
\end{figure}

\subsubsection{Passing decorated surfaces through corners}
We now make the crucial assumption that an isotopy of the decorated stratification does not change its value. Under this assumption, we can then isotopically deform the decorated stratification obtained above $B_c^F\Psi\mid_c$. Concretely, this involves pushing the boundary 1-cells $\gamma_{i_k}$ to meet with the edges of the unit cell $c$. This can be accomplished away from a neighbourhood of the vertices. See the left of Fig. \ref{fig:4branes}.

Denote by $\mathsf{st}^n(v)$ the set of $n$-cells neighbouring a vertex $v\in\Gamma$. By applying $\prod_{c\in \mathsf{st}^3(v)}B_c^F$ to $\Psi$, we obtain a local picture of the form shown on the right of Fig. \ref{fig:4branes}. This stratified web of four decorated surfaces $G_{f_{k_1}}=G_{f^1},\dots,G_{f_{k_4}}=G_{f^4}$ surrounding $v$ pair-wise interface the \textit{new} decorations on the 2-cell boundaries, which are now given by the intermediate surfaces $F_{i_k}$ appearing in \eqref{parallelball}.

We enumerate these new boundary 2-cell simple decorations by $\left\{F_{i_m}\right\}_{{f_m\in\mathsf{st}^2(v)}}$. The 1-skeleton of the stratified web is labelled by the following 2-morphism
\begin{gather*}
    \gamma_{1}: F_{i_1}\rightarrow G_{f^1}\circ G_{f^2},\qquad \gamma_{2}: F_{i_2}\rightarrow G_{f^2}\circ G_{f^3},\qquad \gamma_{3}: F_{i_3}\rightarrow G_{f^3}\circ G_{f^4}\nonumber\\
    \gamma_{4}: F_{i_4}\rightarrow G_{f^4}\circ G_{f^1},\qquad
    \gamma_{5}: F_{i_5}\rightarrow G_{f^1}\circ G_{f^3},\qquad \gamma_{6}: F_{i_6}\rightarrow G_{f^2}\circ G_{f^4}
\end{gather*}
To evaluate this web, we decompose it into three pieces, as depicted in Fig. \ref{fig:piecutfromabove}. The gluing of a pair of these pieces can be  rearranged through the F-move for decorated surfaces \textbf{Proposition \ref{Fmove}} --- for instance, we have
\begin{gather}
  \cC(F_{i_3}\circ G_{f^1},G_{f^3}\circ F_{i_4})\xrightarrow{\sim}\cC(F_{i_3}\circ F_{i_4},G_{f^1}\circ G_{f^3}) \nonumber  \\
  (G_{f^3}\circ \gamma_4^T) \ast(\gamma_3\circ G_{f^1}) \mapsto \tilde\gamma_4^T\ast\tilde\gamma_3\,. \label{collar2morphis}
\end{gather}
This operation is depicted in Fig. \ref{fig:cuttingandgluing}.

\begin{figure}
    \centering
      \includegraphics[width=.6\linewidth]{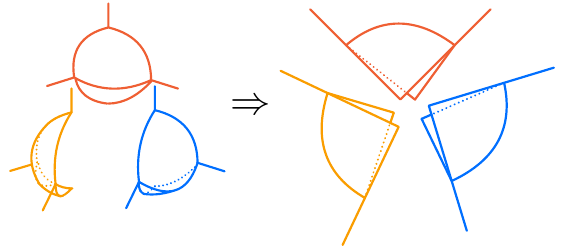}
  \caption{The decomposition of the decorated web arising from $\prod_{c\in\mathsf{st}^3(v)}B^{F_c}_c$ into three pieces  around a vertex $v$.}
  \label{fig:piecutfromabove}
\end{figure}

\begin{figure}
    \centering
  \includegraphics[width=0.8\linewidth]{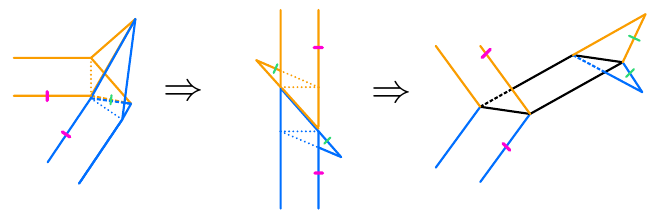}
  \caption{The F-move rearranges the partial-gluing of two pieces of the decorated web surrounding the vertex $v$. The edges that were originally attached prior to the decomposition are marked.} 
  \label{fig:cuttingandgluing}
\end{figure}

Upon rearranging the decorated surfaces, we reattach them according to the original gluing rules. This gives rise to a decorated web as depicted in Fig. \ref{fig:squeezing}. Let us  denote by $S_{v}^{\{F_c\}_c}(\{\tilde\gamma\})$ the resulting decorated web, where $c\in\mathsf{st}^3(v)$ and $\{\tilde\gamma\}$ is the collection of the set of 2-morphisms obtained via eg. \eqref{collar2morphis}. 

\begin{figure}
    \centering
    \includegraphics[width=0.8\linewidth]{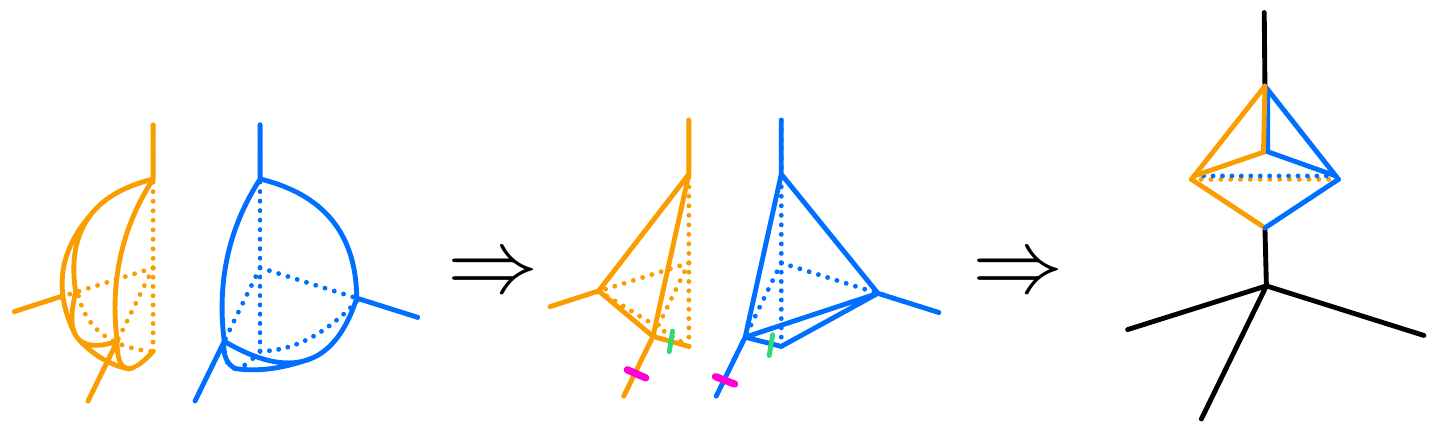}
    \caption{Upon applications of the F-move (Fig. \ref{fig:cuttingandgluing}) and reattaching the piece, we obtain a decorated web along an edge of which the vertex $v$ is attached. This operation in essence "squeezes" the web Fig. \ref{fig:4branes} into a neighbouring unit cell.}
    \label{fig:squeezing}
\end{figure}

\medskip

To summarize, the action of $\prod_{c\in\mathsf{st}^3(v)}B_c^{F_c}$ brings $\Psi$  to a linear combination of new configurations $\Psi'=\Psi_{\{F_{i_k}\}}\in\cH_\text{tot}$ of fully-decorated lattices. The above computations then imply that the coefficients must be proportional to the decorated web,
\begin{equation*}
        \prod_{c\in\mathsf{st}^3(v)}B_c^{F_c}\Psi = \sum_{{\tilde\gamma}} S_{v}^{\{F_c\}_c}(\{\tilde\gamma\})\Psi_{\{F_{i_k}\}} \,.
    \end{equation*}
        Given a complete configuration $\mathsf{F}=\{F_c\}_{c\subset\Gamma}$ of decorated 3-balls, the product over all vertices $v\in\Gamma$ gives rise to a linear map 
        $$\mathcal{B}_\mathsf{F}= \prod_{c\subset\Gamma}B_c^{F_{c}}: \mathcal{H}^\Gamma \to \mathcal{H}^\Gamma$$ 
        which can be evaluated as a contraction of a product of webs associated to each vertex; formally, we write this as
        $$\mathcal{B}_\mathsf{F}\Psi = \sum\left(\prod_vS_{v}^\mathsf{F}\right) \Psi',$$
        where the sum is taken over the data on all intermediate surfaces and their boundaries.

\subsubsection{Evaluation of the bubble $S_v$}

As explained above, the bubble $S_v$ can be explicitly computed. However, here we provide an alternative argument based on the geometry. To start, we observe that the graph $\Gamma_v$ associated to the decorated bubble $S_v$ (top right side of Fig. \ref{fig:squeezing}) is graph isomorphic to \eqref{dual4-simplex}, as shown in Fig. \ref{fig:bubble20j}. \textbf{Theorem \ref{levinwentet}} then tells us that the decorations $(\mathsf{F},\{\tilde\gamma\})$ on this web induce a corresponding decoration $\mathcal{K}=\mathcal{K}_{\mathsf{F},\{\tilde\gamma\}}$ on the dual 4-simplex $\sigma_v$, in accordance with \textbf{Definition \ref{2decorate}}. 

\begin{figure}
    \centering
    \includegraphics[width=0.65\linewidth]{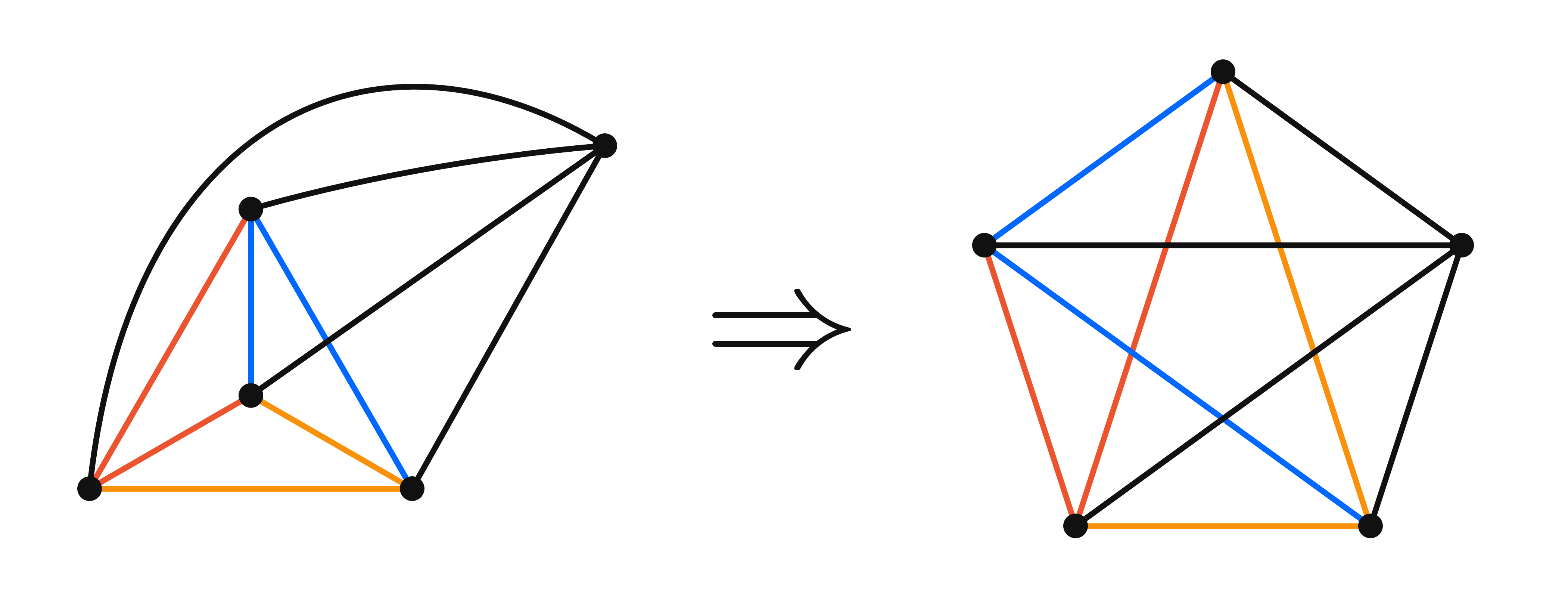}
    \caption{A flat projection of the graph associated to $S_{v}$.}
    \label{fig:bubble20j}
\end{figure}

Formally, this implies that $S_v$ must give the 4-simplex amplitude/20$j$-symbols \eqref{amplitude},
    $$S_{v}^{\mathsf{F}}(\{\tilde\gamma\})=\langle\sigma_v(\mathcal{K}_{\mathsf{F},\{\tilde\gamma\}})\rangle\,,$$
as computed in \S \ref{20jsymbols}. This is reasonable, as given the correspondence $(\mathsf{F},\{\tilde\gamma\})\mapsto\mathcal{K}_{\mathsf{F},\{\tilde\gamma\}}$ is  bijective,  one can then compute the partition function \eqref{partition} on $M^4=M^3\times S^1$ directly from the Hamiltonian $\boldsymbol{H}$,
\begin{align*}
    \operatorname{tr}_{V}e^{-\boldsymbol{H}}&=\operatorname{tr}_{\mathcal{H}^\Gamma}e^{\sum_cB_c}=\sum_{\mathsf{F},{\tilde\gamma}}\prod_{v\in\Gamma}S_v^{\mathsf{F}}(\{\tilde\gamma\})\stackrel{}{=} \sum_{\cal K}\prod_{\sigma\subset T_{M^4}}\langle \sigma( \mathcal K)\rangle=Z_\mathscr{C}(M^4)\,.
\end{align*}
This can be leveraged to carry out the universal construction of the TQFT $Z_\mathscr{C}$.

\section{Outlook}\label{outlook}

It is interesting to work out physically relevant examples of the general explicit construction developed here. As explained in \S \ref{highergauexamples}, for any finite $2$-group $\mathbb{G}$ we can evaluate the $20j$-symbols of the Yetter-Dijkgraaf-Witten theory directly in terms of the higher-categorical data $\mathcal{C}=2\mathrm{Rep}(\mathbb{G})$ \cite{Bartsch:2023wvv}; similarly for the state-sum amplitudes, and membrane-net ground-state Hilbert spaces on simple manifolds. In particular, it would be useful  to identify explicitly the point-, loop-, and membrane-like excitations encoded by the lattice operators. The same can also be done for $\mathbb{G}$ any Lie 2-group, up to the choice of a model for $H_\mathbb{G}$ as described in \textit{Remark \ref{bicategorybgd}}; we emphasize here that, to the best of our knowledge, the computation of the 20$j$-symbols for non-split Lie 2-groups has never been done in the literature.


Given the 3d lattice model $Z_\cC(M^3)$ in \S \ref{latticemodel}, one can  verify  the duality of \textbf{Conjecture \ref{2-duality}} at the level of the Hamiltonian,
\begin{gather*}
    \cH^\Gamma=Z_\cC(M^3,\Gamma) \xrightarrow{\text{EM duality}} Z^\cD(M^3,\Gamma),\qquad \forall~ \text{closed }M^3\,,
\end{gather*}
where $\cC,\cD$ are Morita dual fusion 2-categories and $Z^\cD(M^3,\Gamma)$ is the membrane-net lattice Hilbert space given in, eg., \cite{xi2021latticerealizationgeneralthreedimensional}. This, as well as explicit results mentioned in the previous paragraph, are important for our understanding on $(3+1)$-dimensional
topological phases \cite{Wen:2019}. We will pursue this direction in a followup work.


\medskip

Having a lattice model also provides a concrete physical framework in which the theory of higher-categorical algebra can be probed. By studying the topological defects it hosts, one can explicitly construct the associated tube algebra $A$. Similar task has been undertaken from the dual perspective in \cite{xi2021latticerealizationgeneralthreedimensional}; indeed, one produces the Drinfeld centre $\operatorname{Mod}_\mathsf{2Vect}(A)\simeq Z_1(\cC)\simeq Z_1(\cD)$. This makes $A$ a key algebraic gadget in capturing the duality defect associated to higher-group symmetries. However, in contrast to the dual perspective, our lattice model realizes (in the case $\cC= \operatorname{2Rep}(H)$) the tube algebra $A$ as a model of the categorical double $D(H)$ of the underlying Hopf category $H$ \cite{Huang2026-dr,huang2023tannaka}. This fact is significant, as the canonical quasitriangularity \cite{neuchl1997representation,Chen:2025?} of $D(H)$ directly provides a class of solutions of the Zamolodchikov tetrahedron equations \cite{Zamolodchikov:1980,Kapranov:1994} associated to $H$.

Viewing the $(3+1)$d TQFT as a SymTFT suggests extending the lattice construction to spatial manifolds with physical boundaries and interfaces, providing a microscopic framework for studying $3$d phases with 2-group symmetry $\bbG$ or Morita-dual categorical symmetries. A concrete goal is to derive boundary face, edge, and corner operators compatible with the bulk $20j$ amplitudes, determine which bulk excitations can terminate or condense at each boundary, and identify the fusion and junction degrees of freedom of interfaces between different boundary conditions. With a fixed topological symmetry boundary, varying the physical boundary conditions and interactions would then allow one to investigate symmetry-protected topological phases~\cite{Delcamp:2018wlb}, spontaneously symmetry-breaking phases~\cite{Liu:2024znj,Brito:2026eqi}, and symmetry-enriched topological or gapless phases. Cases with nontrivial 2-group anomalies could be studied through the corresponding anomaly inflow, using the appropriate oriented or spin bordism-based description~\cite{Gu:2026trc}; a spin realization would require a corresponding refinement of the lattice construction. This approach would make the boundary realization of higher categorical symmetries explicit and provide a setting for exploring dualities between these phases through gauging suitable $E_1$-algebras in the symmetry categories.

The local $20j$ amplitudes also suggest an explicit tensor-network representation of the ground-state subspace. A useful next step is to obtain the ground-state tensors from the cylinder state sum, identify their virtual defect operators, and derive the local transformations induced by changes of triangulation. This would provide a computational formulation of ground-state wavefunctions and linked-defect amplitudes, and could make the proposed Morita duality explicit at the level of ground-state projectors and observables. The resulting tensors would also provide a setting for testing the conjectured relation between genuine multipartite entanglement and TQFT partition functions proposed in~\cite{DelZotto:2026fpw}. In the present $(3+1)$d setting, a concrete goal is to construct five-partite replica invariants associated with triangulated four-manifolds and determine whether their genuine entanglement signals reproduce $|Z_{\mathcal C}(M^4)|^2$ with the prescribed sphere normalization. In finite 2-group examples, these entanglement diagnostics could complement comparisons of models with different Postnikov and action classes through suitable loop-braiding processes.

\newpage

\printbibliography

\end{document}